\documentclass[12pt,notitlepage]{article}

\usepackage[T1]{fontenc}
\usepackage{lmodern}
\usepackage[margin=1.05in]{geometry}
\usepackage{setspace}
\AtBeginDocument{%
    \setlength{\abovedisplayskip}{5pt}%
    \setlength{\belowdisplayskip}{5pt}%
}

\usepackage{mathtools,amssymb,amsthm}
\usepackage{nicematrix}
\usepackage{dsfont}
\newcommand{\one}{\mathds 1}

\usepackage{xcolor}
\usepackage{graphicx}
\usepackage{subcaption}
\usepackage{tikz}
\usetikzlibrary{arrows.meta,positioning,patterns.meta}
\usepackage[shortlabels]{enumitem}
\newcommand{\figurepanelwidth}{0.49\linewidth}
\usepackage{titlesec}
\titleformat{\subsubsection}[block]
    {\normalfont\normalsize\itshape\filcenter}
    {\thesubsubsection}
    {1em}
    {}

\usepackage{abstract}

\usepackage{natbib}
\usepackage{xurl}

\definecolor{drkred}{rgb}{0.7,0.11,0.11}
\colorlet{drkblue}{blue!61.8!black}
\definecolor{beaublue}{rgb}{0.74,0.83,0.9}

\usepackage[
    hypertexnames=false,
    colorlinks=true,
    linkcolor=drkblue,
    urlcolor=drkblue,
    citecolor=drkblue,
    anchorcolor=drkblue
]{hyperref}
\hypersetup{
    pdftitle={Pricing Expertise with Information Spillovers},
    pdfauthor={Klajdi Hoxha}
}

\usepackage{apptools}

\newtheorem{definition}{Definition}
\newtheorem{theorem}{Theorem}
\newtheorem{proposition}{Proposition}
\newtheorem{lemma}{Lemma}
\newtheorem{assumption}{Assumption}
\newtheorem{property}{Property}

\AtAppendix{%
    \counterwithin{proposition}{section}%
    \counterwithin{lemma}{section}%
    \counterwithin{definition}{section}%
}

\usepackage[nameinlink,capitalize,noabbrev]{cleveref}
\crefname{assumption}{Assumption}{Assumptions}
\crefname{property}{Property}{Properties}
\Crefname{property}{Property}{Properties}
\crefname{section}{Section}{Sections}
\crefname{subsection}{Section}{Sections}
\crefname{equation}{Eqn.}{Eqns.}
\crefname{figure}{Figure}{Figures}
\crefname{appendix}{Appendix}{Appendices}
\Crefname{appendix}{Appendix}{Appendices}

\usepackage{bookmark}

\DeclarePairedDelimiter\autobracket{(}{)}
\newcommand{\br}[1]{\autobracket*{#1}}

\DeclarePairedDelimiter\autosbracket{[}{]}
\newcommand{\sbr}[1]{\autosbracket*{#1}}

\DeclarePairedDelimiter\autocbracket{\{}{\}}
\newcommand{\cbr}[1]{\autocbracket*{#1}}

\DeclarePairedDelimiterX\autolsbracket[1]{[}{)}{#1}
\newcommand{\lsbr}[1]{\autolsbracket*{#1}}

\DeclarePairedDelimiterX\autorsbracket[1]{(}{]}{#1}
\newcommand{\rsbr}[1]{\autorsbracket*{#1}}

\DeclareMathOperator*{\argmax}{argmax}

\newcommand{\rbb}{\mathbb{R}}
\newcommand{\p}{\mathbb{P}}
\newcommand{\E}{\mathds E}
\newcommand{\mcal}{\mathcal{M}}

\newcommand{\ctheta}{\theta^{\mathrm{crit}} }
\newcommand{\sutheta}{\theta_{*} }
\newcommand{\sotheta}{\theta^{*} }

\title{\vspace{-1em}%
    Pricing Expertise with Information Spillovers\thanks{%
    I am grateful to Steven Callander, Ravi Jagadeesan, Michael Ostrovsky, Ilya Segal, Andrzej Skrzypacz, Frank Yang, Kemal Y{\i}ld{\i}z, Weijie Zhong, and seminar participants at Stanford University for helpful comments and suggestions.%
    }%
}
\author{Klajdi Hoxha\thanks{%
    Graduate School of Business, Stanford University. Email: hklajdi@stanford.edu.%
}}
\date{November 16, 2026}

\begin{document}

\maketitle

\begin{abstract}
Why do experts charge value-based fees for some services and flat fees for others? For example, an analyst charges the same fee for a stock report whether the client invests a thousand dollars or a million, while an investment bank charges a higher fee for a fairness opinion on a large merger than on a small one. We consider the problem of an expert who sells expertise to a client whose private type determines his project and how much he has at stake. Because project states are correlated, cheaper expertise about a related project can substitute for expertise about his own. We develop a Markov model of information spillovers and provide an explicit characterization of the optimal menu. With perfect spillovers, as with a stock report, a common flat fee for full information is optimal. With sufficiently weak spillovers, as with a fairness opinion, each client receives full information at a value-based fee equal to his willingness to pay. Between these extremes, low types can receive distorted expertise at their full willingness to pay, while middle and high types receive full information. Middle types pay a common base fee plus a project-specific fee that rises with what is at stake, and high types pay a value-based fee equal to their willingness to pay. Because higher types lose more by buying expertise about another project, the project-specific fee can rise faster than willingness to pay, so client surplus can peak among middle types and fall to zero at the top.
\end{abstract}

\section{Introduction}
\label{sec:introduction}

Consider two investors deciding whether to buy shares in the same company, one planning to invest a thousand dollars and the other a million. An analyst's report on the company's prospects is equally informative for both, but the second investor has far more at stake. The analyst, however, usually cannot verify how much each client plans to invest. If she charged more for larger investments, the large investor could claim a smaller investment and buy the report at the lower fee.\footnote{Many subscription services that publish stock ratings for individual investors charge fees that do not depend on how much a subscriber invests, such as \$249 a year for Morningstar Investor and \$299 for Seeking Alpha Premium (\url{https://www.morningstar.com/mm/} and \url{https://about.seekingalpha.com/premium-subscription-price-update}, accessed November 2026).}

By contrast, consider a board that is deciding whether to approve a merger. It asks an investment bank for a fairness opinion on whether the terms are fair to shareholders.\footnote{Boards obtain fairness opinions as evidence that they evaluated the deal with due care \citep{kisgen2009fairness}.} A cheaper opinion on a smaller merger in the same industry may show how comparable firms were valued, but it does not assess the terms of the board's own deal. Relying on that opinion can lead to an incorrect decision, and the cost of such a mistake rises with the size of the merger. The bank can therefore charge more for an opinion on a larger deal.\footnote{In a study of U.S. stock-for-stock mergers, \citet[pp.~132--133]{matthews2016stock} reports that average fairness opinion fees were approximately \$125,000 for targets with an equity value of \$25 million and \$1 million for targets with an equity value of \$1 billion.}

In both examples, the more a client can learn from cheaper expertise intended for someone else, the less the expert can charge him for expertise about his own decision. Such \textit{information spillovers} help determine whether the expert charges a \textit{common flat fee} or a \textit{value-based fee} that reflects each client's willingness to pay. Spillovers are perfect for the analyst's report, since every investor is deciding about the same company. They are weak for fairness opinions, since each opinion concerns the terms of a single deal. We develop a model of expertise with information spillovers, which delivers a common flat fee when spillovers are perfect and a value-based fee when they are sufficiently weak. In many settings, however, expertise about one project may help a client evaluate a related project without fully resolving his uncertainty. Our main results characterize the optimal menu under such intermediate spillovers.

In our model, a monopolist expert sells expertise to a client deciding whether to undertake a project. Each project is in one of two states, good or bad, and a correct decision means undertaking a good project or rejecting a bad one. The client's private type determines both which project he is considering and his payoff from a correct decision. We assume that every project is good with the same prior probability, which is above one half. Without expertise, the client therefore undertakes his project, and expertise is valuable because it can recommend rejecting a bad one. The expert offers a menu of information products, each a Blackwell experiment whose signal depends only on the state of the project it concerns. The client can buy at most one information product and pays its fee before he acts. Because project states are correlated, he may prefer a cheaper information product intended for another project. We characterize the menu of information products and fees that maximizes the expert's expected revenue.

We first suppose that the client is one of two types, low or high, and that spillovers are imperfect. If the expert charges both types their willingness to pay for full information, the high type may prefer the low type's cheaper information product. When spillovers are sufficiently weak, he does not, because he would lose more from the less relevant expertise than he saves in fees. The expert then achieves \textit{efficient full surplus extraction}, providing full information to both types at their full willingness to pay.

When spillovers are stronger and imitation is profitable, improving the low type's expertise raises his fee but also makes it more valuable to the high type. If the expected gain in fees from the low type is at least as large as the expected cost of leaving additional surplus to the high type, providing full information to both types is optimal. The low type pays his willingness to pay, and the high type pays that same fee plus the value he would lose by buying the low type's information product. If instead the cost exceeds the gain, the expert achieves \textit{distorted full surplus extraction}, distorting the low type's expertise and extracting the full value of the information product each type receives. This distortion takes the form of sometimes recommending that the low type reject a good project, while always recommending that he reject a bad one.

We then consider a continuum of types. Projects lie on a line ordered by the payoff from a correct decision. The correlation between two projects' states falls with the distance between them, so, unlike with two types, a client can buy expertise about an arbitrarily close project that is almost as informative as expertise about his own. If each type pays his full willingness to pay, however, fees change smoothly with type, so a client saves little by buying expertise about a nearby project. We show that efficient full surplus extraction is possible if and only if, at every type, the marginal gain from imitating a lower type, the fee saved, is no greater than the marginal cost of imitation, the information value lost. This condition holds when spillovers are sufficiently weak, and each client then pays a value-based fee equal to his willingness to pay.

When this condition fails and spillovers are imperfect, we impose additional conditions on payoffs and the type distribution. These include concavity of the payoff from a correct decision and strict monotonicity of the \textit{informational virtual value}, our counterpart of the Myersonian virtual value \citep{myerson1981optimal}. We show that the optimal menu then partitions the type space into up to three regions. Low types face distorted full surplus extraction, as in the two-type case. Middle types receive full information and earn positive surplus. Each pays a positive \textit{common base fee}, equal to the willingness to pay of the lowest type who receives full information, plus a \textit{project-specific fee}. The project-specific fee is zero for that lowest type and rises at a rate proportional to the client's willingness to pay for full information, which makes the fee strictly convex in type. High types face efficient full surplus extraction, as in the case of sufficiently weak spillovers.

The convex project-specific fee explains why surplus can peak among middle types. In standard screening problems, a higher type values every product at least as much as a lower type does. By buying the lower type's product, he can guarantee himself at least the lower type's surplus. In our setting, a lower type's information product concerns a different project, so it is less relevant to a higher type, who may value it less than the lower type does. Because the payoff is concave, willingness to pay for full information rises at a nonincreasing rate, while the fee rises at an increasing rate. Once the fee rises faster than willingness to pay, surplus begins to fall. If it returns to zero before the highest type, the fee equals willingness to pay from then on, so every higher type earns zero surplus.

\paragraph{Technical Contributions.} Our main technical contribution is a Markov model of information spillovers across projects. We hold fixed the prior probability $\mu$ that a project is good and impose Markov, homogeneity, and continuity properties on the joint prior. These properties yield a stationary, homogeneous, and reversible two-state Markov chain in which, per unit distance, the state switches from good to bad at rate $\br{1-\mu}\lambda$ and from bad to good at rate $\mu\lambda$, for some parameter $\lambda\geq0$. Writing $X_\theta$ for the indicator that project $\theta$ is good, two projects at distance $\Delta$ satisfy
\[
    \operatorname{Corr}\br{X_\theta,X_{\theta+\Delta}}=e^{-\lambda\Delta}.
\]
The parameter $\lambda$, the \textit{correlation decay rate}, thus measures how quickly expertise about one project loses relevance for another. Moreover, the Markov model yields a closed form for the value of any information product to any client.

The informational virtual value summarizes the expert's trade-off between distorting expertise and leaving surplus to higher types. Let $u\br{\theta}$ denote type $\theta$'s payoff from a correct decision, and let $F$ and $f$ denote the distribution and density of types. The informational virtual value is
\[
    u\br{\theta}
    -\br{u'\br{\theta}-2\br{1-\mu}\lambda u\br{\theta}}
    \frac{1-F\br{\theta}}{f\br{\theta}}.
\]
Each term is measured per unit gain in type $\theta$'s probability of a correct decision. The first term is the resulting increase in his value, and the second is the net information rent this gain leaves to higher types. The first part of the second term, $u'\br{\theta}\br{1-F\br{\theta}}/f\br{\theta}$, is the usual information rent, since higher types value the same gain more. With imperfect spillovers, a more informative information product also loses more value when used for a different project. Each unit gain therefore raises a nearby type's marginal cost of imitation by $2\br{1-\mu}\lambda u\br{\theta}$, which offsets part of this rent. Under perfect spillovers ($\lambda=0$), there is no offset, and the expression reduces to the Myersonian virtual value for $u\br{\theta}$.

\paragraph{Related Literature.} Our paper contributes to the literature on the design and pricing of information products for privately informed buyers \citep[see][for a survey]{bergemann2019markets}. \citet{bergemann2018design} study how a seller screens buyers with different private beliefs about a common state through a menu of experiments. In their binary-state, binary-action model with a continuum of types, an optimal menu offers full information and at most one partially informative experiment. If all types choose the same action without additional information, a common posted price for full information is optimal. Our continuum model instead has a continuum of projects with correlated binary states, and each client cares only about the state of his own project. We rule out differences in prior beliefs by assigning every project the same probability of being good. Although all clients take the same action without expertise, the optimal menu can include a continuum of information products, distorted for low types and perfectly informative for higher types.

\citet{liu2021optimal} study buyers choosing between taking an action with an uncertain payoff and remaining inactive, where the buyer's type determines his payoff from acting and is drawn independently of the state. The type initially indifferent between acting and remaining inactive, when one exists, earns the highest surplus. We likewise draw the client's type independently of the vector of project states. No type in our model is indifferent between undertaking and rejecting his project without expertise, yet surplus can still be highest for middle types.

\citet{babaioff2012optimal} study a more general framework in which the buyer's and seller's private signals about the state may be correlated, and the payoff from each action depends on both the buyer's type and the state. When the buyer's type is independent of the state, they show that a menu of experiments at fixed prices is optimal. With correlation, the optimal mechanism may instead require payments that depend on the disclosed signal. They compute optimal mechanisms through linear and convex programs. If the state is the vector of project states and each type cares only about his own project, the payoff structure of a finite version of our model fits their framework, with the client's type independent of the state. We additionally restrict each information product to depend only on the state of the project it concerns and allow a continuum of projects, which lies outside their finite framework. With this restriction and the Markov prior, the value of any information product to any client has a closed form, and our payoff and regularity conditions permit an explicit characterization of the optimal menu for a continuum of types.

Information products in our model differ horizontally in the project they concern and vertically in their Blackwell informativeness about that project. Our screening problem therefore relates to screening by quality \citep{mussa1978monopoly,maskin1984monopoly} and to nonlinear pricing with horizontal and vertical heterogeneity \citep{rochet2002nonlinear,stole2007price}. In the canonical one-dimensional model of screening by quality, a higher type values every product at least as much as a lower type does. In our model, each client's ideal product is full information about his own project, so a higher type may value a lower type's information product less than the lower type does.

\citet{tan2026screening} studies a monopolist who sells a product line of horizontally differentiated variants to buyers who privately know their ideal product. When willingness to pay is correlated with horizontal need, product mismatch deters higher-value buyers from imitating lower-value ones, since imitation leaves them with a product far from their ideal. In our model, the loss from buying expertise about another project instead depends jointly on the correlation between project states and the informativeness of the information product.

Information spillovers in our model arise because expertise about one project can inform a client's decision about another, not because buyers of information interact with one another. In \citet{admati1986monopolistic,admati1990direct}, purchased information leaks to other traders through informative prices, which dilutes its value. \citet{admati1986monopolistic} show that selling noisier versions of the information can be optimal, because the added noise limits this leakage. Our expert likewise distorts expertise, but only for low types, to make it less valuable to higher types who might imitate them. \citet{rodriguez2024strategic} extends the screening problem of \citet{bergemann2018design} to buyers who play a game after acquiring information. Even when buyers choose the same action without additional information, the optimal menu can include two informative experiments. Both her two-type menu and ours can take the form of distorted full surplus extraction, with full information for the high type and distorted expertise for the low type. Our distortion, however, arises from information spillovers across projects rather than from strategic interaction among buyers.

Our Markov model of information spillovers relates to models that represent uncertainty across a continuum of alternatives, such as technologies or policies, by a stochastic process. \citet{jovanovic1990long} derive a Brownian-motion representation of proportional productivity changes along a technological dimension from axioms on beliefs, with a single parameter that measures technological opportunity. \citet{callander2008theory} uses Brownian motion to model policy expertise, so learning the outcome of one policy is informative about nearby policies without revealing the entire mapping from policies to outcomes. Our Markov, homogeneity, and continuity properties play roles similar to the axioms of \citet{jovanovic1990long}. Given equal marginal priors, they likewise yield a single parameter, the correlation decay rate $\lambda$ (see \cref{prop:markov-prior}), which, together with the distance between projects, determines how much learning about one project reveals about another. Whereas Brownian motion assigns each alternative a real-valued outcome, each project in our model is either good or bad, so the process is a two-state Markov chain.

\paragraph{Roadmap.} \Cref{sec:model} introduces the model and reduces the expert's problem to simple menus. \Cref{sec:two-types} characterizes the optimal menu with two types. \Cref{sec:continuum} characterizes the optimal menu with a continuum of types. \Cref{sec:conclusion} concludes.

\section{Model}
\label{sec:model}

\Cref{subsec:primitives} introduces the model primitives, \cref{subsec:value-of-information} defines the value of information products, and \cref{subsec:experts-problem} states the expert's problem and reduces it to simple menus. Complete proofs are in \cref{app:model}.

\subsection{Primitives}
\label{subsec:primitives}

A monopolist expert (she) sells expertise to a client (he) who must decide whether to undertake a project. Projects are indexed by $\Theta$.
The client's private information is the index $\theta\in\Theta$ of his project of interest. We call $\theta$ his \textit{type} and assume it is drawn from a distribution $F\in\Delta\br{\Theta}$. A client of type $\theta$ earns $u\br{\theta}>0$ if he makes the correct decision about project $\theta$.

Each project $\theta$ has a state
\[
    \omega_\theta\in\Omega=\cbr{G,B},
\]
where $G$ denotes \textit{good} and $B$ denotes \textit{bad}.\footnote{The restriction to two states is sufficient for the economic trade-offs we intend to capture.} A client of type $\theta$ cares about the state $\omega_\theta$. He takes an action $a\in A=\Omega$ and receives the ex-post payoff
\[
    u\br{\theta}\,\one\cbr{a=\omega_\theta},
\]
so he earns $u\br{\theta}$ if his action matches the state and zero otherwise. We interpret action $G$ as undertaking the project and action $B$ as rejecting it, and we say that project $\theta$ is good when $\omega_\theta=G$ and bad when $\omega_\theta=B$. A correct decision therefore means undertaking a good project or rejecting a bad one.

The state of nature is the joint realization of all project states,
\[
    \omega=\br{\omega_\theta}_{\theta\in\Theta}\in\Omega^\Theta.
\]
The expert and the client share a common prior $\p$ on $\Omega^\Theta$, under which the states of different projects may be correlated.

\paragraph{Expertise.} Expertise about project $\theta\in\Theta$ is a Blackwell experiment
\[
    \pi\br{\theta}:\Omega\to\Delta\br{S}
\]
on a finite signal space $S$. We call this experiment an \textit{information product}.\footnote{By \cref{prop:simple-menus}, it is without loss of generality to restrict attention to binary signal spaces. The argument can be readily extended to arbitrary standard Borel signal spaces.} We write $\pi_G^s\br{\theta}$ and $\pi_B^s\br{\theta}$ for the probability of signal $s$ when $\omega_\theta=G$ and $\omega_\theta=B$, respectively. The expert can provide any information product at zero marginal cost.

Each information product $\pi\br{\theta}$ is \textit{customized} to its project, in that, conditional on $\omega_\theta$, its signal is independent of the other project states $\br{\omega_{\widetilde\theta}}_{\widetilde\theta\neq\theta}$.\footnote{The restriction is not without loss of generality, since information products whose signals also depend on other projects' states may improve screening. We leave this extension for future research.} Customization does not make expertise about one project useless to clients who consider other projects. Although the signal of $\pi\br{\theta}$ depends only on $\omega_\theta$, the correlation between project states makes it informative about other projects as well, which is the source of \textit{information spillovers}.

\paragraph{Menus and Timing.} The expert offers a revenue-maximizing menu $\mcal$ of information products, each customized to some project $\theta$ and offered at a price $t$. The timing is as follows:
\begin{enumerate}[label=(\roman*)]
    \item The expert commits to a menu $\mcal$.
    \item The client's type $\theta$ and the state of nature $\omega$ are drawn independently from $F$ and $\p$, respectively. The client privately learns his type but not the state of nature.
    \item The client buys at most one information product from the menu, possibly one customized to another project, and pays its price.
    \item If the client purchases an information product, he observes its signal and updates his belief about $\omega_\theta$. He then takes an action.
\end{enumerate}
In stage (iii), the client may prefer cheaper expertise about another project if it remains sufficiently informative about his own.

\paragraph{Contracting Environment.} We assume that only the client's choice from the menu and its price are contractible, while his action, his signal, and the realized project states are not.\footnote{\citet{esHo2007price} instead allow transfers to depend on the client's action. In their model, the client has a private initial estimate of his project's profitability, and the consultant can condition transfers on whether he undertakes the project. \citet{smolin2020disclosure} likewise lets a seller price both information about an object's attributes and the object itself, so the buyer's purchase decision is contractible.} The expert therefore cannot condition the price on whether the client's decision turns out to be correct.\footnote{Contingent fees are common in some advisory markets. In November 2025, a McKinsey managing partner said that about a quarter of the firm's global fees came from outcome-based pricing (\url{https://finance.yahoo.com/news/ai-reshaping-mckinsey-makes-money-195132745.html}, accessed November 2026).} As in stage~(iii) of the timing, the client pays before he takes an action.

\paragraph{Equal Marginal Priors.} We assume that, under the prior $\p$, every project is good with the same probability,
\[
    \p\br{\omega_\theta=G}=\mu\in\br{1/2,1}
    \quad\text{for all }\theta\in\Theta.
\]
Clients are therefore ex ante equally optimistic, and each expects his own project to be good with probability $\mu>1/2$.

\subsection{Value of Information Products}
\label{subsec:value-of-information}

Without expertise, a client strictly prefers to undertake his project, action $G$, since it is good with probability $\mu>1/2$. His \textit{outside option} is therefore $\mu u\br{\theta}$.

He makes an incorrect decision only when his project is bad, which happens with probability $1-\mu$. Expertise is valuable because it can recommend rejecting such a project. With a perfectly informative signal, he always makes a correct decision and earns $u\br{\theta}$.

Consider a type $\theta$ who buys information product $\pi\br{\widetilde{\theta}}$ about project $\widetilde{\theta}$. For each signal $s\in S$ with positive probability, let $\nu_s\br{\theta\mid\widetilde{\theta}}$ be his posterior probability of the good state $\omega_\theta=G$. After observing $s$, undertaking the project gives him expected payoff $\nu_s\br{\theta\mid\widetilde{\theta}}\,u\br{\theta}$, and rejecting it gives $\br{1-\nu_s\br{\theta\mid\widetilde{\theta}}}u\br{\theta}$. He takes the action with the higher expected payoff, so his expected payoff after signal $s$ is
\[
    \max\cbr{\nu_s\br{\theta\mid\widetilde{\theta}},1-\nu_s\br{\theta\mid\widetilde{\theta}}}u\br{\theta}.
\]
The \textit{value} of the information product, or his willingness to pay for it, is therefore his expected payoff relative to the outside option:
\begin{equation}
    \label{dfn:valueforinformation}
    V\br{\theta,\pi\br{\widetilde{\theta}}}
    :=u\br{\theta}\br{\E\br{\max\cbr{\nu_s\br{\theta\mid\widetilde{\theta}},1-\nu_s\br{\theta\mid\widetilde{\theta}}}}-\mu}.
\end{equation}

We refer to learning a project's state with certainty as receiving \textit{full information} about that project. Let $\overline{\pi}\br{\theta}$ denote a perfectly informative information product about project $\theta$, whose signal reveals $\omega_\theta$ with probability one, and let $\underline{\pi}\br{\theta}$ denote an uninformative one. Then
\[
    \overline V\br{\theta}:=V\br{\theta,\overline{\pi}\br{\theta}}=\br{1-\mu}u\br{\theta}
    \quad\text{and}\quad
    \underline V\br{\theta}:=V\br{\theta,\underline{\pi}\br{\theta}}=0,
\]
and every information product $\pi\br{\widetilde{\theta}}$ satisfies
\[
    \underline V\br{\theta}\leq V\br{\theta,\pi\br{\widetilde{\theta}}}\leq\overline V\br{\theta}.
\]
Since expertise is costless, providing full information about a client's project is socially efficient. A menu achieves \textit{full surplus extraction} from type $\theta$ if he pays the full value of his information product. We say that full surplus extraction is \textit{efficient} if the type receives full information about his project and \textit{distorted} if the information product is not perfectly informative.

\subsection{Expert's Problem}
\label{subsec:experts-problem}

We first state the expert's problem and then show that it suffices to consider information products with two possible recommendations, one to undertake the project and one to reject it.

\paragraph{Direct Menus.} We can restrict attention without loss of generality to direct menus in which the information product assigned to each reported type $\theta$ is customized to project $\theta$.\footnote{An information product about another project can be replaced by one about project $\theta$ that preserves type $\theta$'s value and weakly reduces its value to every other type.} A direct menu is a collection
\[
    \br{\pi,t}:=\cbr{\pi\br{\theta},t\br{\theta}}_{\theta \in \Theta},
\]
where a client who reports $\theta$ receives information product $\pi\br{\theta}$ and pays the transfer $t\br{\theta}$. The expert chooses a direct menu $\br{\pi,t}$ to maximize total expected transfers
\begin{equation}
    \label{eqn:objective}
    \tag{Obj}
    \max_{\br{\pi,t}}\int_\Theta t\br{\theta}\,dF\br{\theta}
\end{equation}
subject to incentive compatibility,
\begin{equation}
    \label{eqn:IC}
    \tag{IC}
    V\br{\theta,\pi\br{\theta}}-t\br{\theta}
    \geq V\br{\theta,\pi\br{\widetilde{\theta}}}-t\br{\widetilde{\theta}}
    \quad\text{for all }\theta,\widetilde{\theta}\in\Theta,
\end{equation}
and individual rationality,
\begin{equation}
    \label{eqn:IR}
    \tag{IR}
    V\br{\theta,\pi\br{\theta}}-t\br{\theta}\geq0
    \quad\text{for all }\theta\in\Theta.
\end{equation}
We write
\[
    U\br{\theta}:=V\br{\theta,\pi\br{\theta}}-t\br{\theta}
\]
for the \textit{surplus}, or information rent, of a truth-telling type $\theta$.

\paragraph{Simple Menus.} Next, we show that it is sufficient for each information product to recommend an action that a truth-telling type finds optimal to follow.
\begin{definition}[Simple Menus]
    \label{dfn:simple-menu}
    We say that a direct menu $\br{\pi,t}$ is a \textit{simple menu} if, for every $\theta\in\Theta$, the information product $\pi\br{\theta}$ has a binary signal space $S=\cbr{G,B}$ and is \textit{responsive}, meaning that taking $a=s$ is optimal for type $\theta$ after receiving signal $s$.
\end{definition}

\begin{proposition}[Sufficiency of Simple Menus]
    \label{prop:simple-menus}
    For any IC and IR direct menu $\br{\pi,t}$, there exists an IC and IR simple menu $\br{\widehat{\pi},t}$ such that
    \[
        V\br{\theta,\widehat{\pi}\br{\theta}}=V\br{\theta,\pi\br{\theta}}
        \quad\text{for all }\theta\in\Theta.
    \]
\end{proposition}
Following the revelation principle for Bayesian games of communication \citep{myerson1982optimal}, as adapted by \citet{bergemann2018design}, we show this by garbling each $\pi\br{\theta}$, combining the signals that lead type $\theta$ to the same action into a single recommendation, $s=G$ or $s=B$.

For the rest of the paper, we consider simple menus with $S=\cbr{G,B}$. For ease of notation, we write
\[
    \pi_G\br{\theta}:=\pi_G^G\br{\theta}
    \quad\text{and}\quad
    \pi_B\br{\theta}:=\pi_B^B\br{\theta}
\]
for the probabilities of a correct recommendation conditional on states $G$ and $B$, respectively. Each information product $\pi\br{\theta}$ can therefore be represented by the following matrix
\[
    \begin{array}{c|cc}
        \omega_\theta\backslash s & G & B \\
        \hline
        G & \pi_G\br{\theta} & 1-\pi_G\br{\theta}\\
        B & 1-\pi_B\br{\theta} & \pi_B\br{\theta}
    \end{array}
\]
Thus $\pi_G\br{\theta}$ is the probability that the expert recommends undertaking a good project, and $\pi_B\br{\theta}$ is the probability that she recommends rejecting a bad one. A perfectly informative information product therefore has $\overline{\pi}_G\br{\theta}=\overline{\pi}_B\br{\theta}=1$.

An information product is uninformative when each recommendation is equally likely in both states, that is, when $\pi_G\br{\theta}+\pi_B\br{\theta}=1$. Let $\underline{\pi}\br{\theta}$ denote the uninformative information product with $\underline{\pi}_G\br{\theta}=1$ and $\underline{\pi}_B\br{\theta}=0$, which always recommends undertaking the project. An uninformative information product leaves the client's belief at the prior $\mu>1/2$ after every recommendation, so he undertakes the project whatever the expert recommends. Responsiveness therefore requires that recommendation $B$ never be sent, which makes $\underline{\pi}\br{\theta}$ the only responsive uninformative information product.

\paragraph{Responsiveness.} Fix a type $\theta$ and his information product $\pi\br{\theta}$. Let $\nu_s\br{\theta}:=\nu_s\br{\theta\mid\theta}$ denote his posterior of the good state after recommendation $s$. By Bayes' rule, whenever recommendations $s=G$ and $s=B$ are sent with positive probability,
\[
    \nu_G\br{\theta}
    =\frac{\mu\pi_G\br{\theta}}{\mu\pi_G\br{\theta}+\br{1-\mu}\br{1-\pi_B\br{\theta}}}
    \quad\text{and}\quad
    \nu_B\br{\theta}
    =\frac{\mu\br{1-\pi_G\br{\theta}}}{\mu\br{1-\pi_G\br{\theta}}+\br{1-\mu}\pi_B\br{\theta}}.
\]
Following a recommendation is optimal when the recommended state is at least as likely as the other, that is, $\nu_G\br{\theta}\geq1-\nu_G\br{\theta}$ and $1-\nu_B\br{\theta}\geq\nu_B\br{\theta}$, or equivalently, when
\begin{align*}
    \mu\pi_G\br{\theta}&\geq\br{1-\mu}\br{1-\pi_B\br{\theta}},\\
    \br{1-\mu}\pi_B\br{\theta}&\geq\mu\br{1-\pi_G\br{\theta}}.
\end{align*}
Because $\mu>1/2$, the second inequality implies the first: the prior already favors undertaking the project, so only a recommendation to reject it must be persuasive. Rearranging the second inequality shows that an information product $\pi\br{\theta}$ is responsive if and only if
\begin{equation}
    \label{eqn:responsive}
    \tag{RSP}
    \mu\pi_G\br{\theta}+\br{1-\mu}\pi_B\br{\theta}\geq\mu.
\end{equation}
The left side of \eqref{eqn:responsive} is type $\theta$'s probability of a correct decision when he follows the expert's recommendations: undertaking a good project or rejecting a bad one. If he ignores them and undertakes the project, his probability of a correct decision is the prior $\mu$. Thus, under \eqref{eqn:responsive}, following the expert's recommendations is at least as good as ignoring them. The value of a responsive information product $\pi\br{\theta}$ to type $\theta$ is then
\begin{equation}
    \label{eqn:valueofinformation_responsive}
    V\br{\theta,\pi\br{\theta}}
    =\br{\mu\pi_G\br{\theta}+\br{1-\mu}\pi_B\br{\theta}-\mu}u\br{\theta}.
\end{equation}

Note that responsiveness concerns only the truth-telling type. If type $\theta$ buys $\pi\br{\widetilde{\theta}}$ for $\widetilde{\theta}\neq\theta$, then his posterior of the good state $\nu_s\br{\theta\mid\widetilde{\theta}}$ after recommendation $s$ is
\begin{equation}
    \label{eqn:cross-posterior}
    \nu_s\br{\theta\mid\widetilde{\theta}}
    =\nu_s\br{\widetilde{\theta}}\,\p\br{\omega_\theta=G\mid\omega_{\widetilde{\theta}}=G}
    +\br{1-\nu_s\br{\widetilde{\theta}}}\,\p\br{\omega_\theta=G\mid\omega_{\widetilde{\theta}}=B}.
\end{equation}
When $\nu_B\br{\widetilde{\theta}}\leq1/2$, recommendation $B$ persuades type $\widetilde{\theta}$ to reject his project, but it need not persuade type $\theta$. To see this, by \eqref{eqn:cross-posterior}, $\nu_B\br{\theta\mid\widetilde{\theta}}$ can remain above $1/2$ when the truth-telling type's posterior $\nu_B\br{\widetilde{\theta}}$ is close to $1/2$, or when spillovers are weak, so that $\p\br{\omega_\theta=G\mid\omega_{\widetilde{\theta}}=B}$ is high. Type $\theta$ may then undertake his project despite the recommendation to reject it.

\paragraph{Reduced Problem.} The expert's \textit{reduced problem} is therefore to choose a simple menu $\br{\pi,t}$, with $\pi_G\br{\theta},\pi_B\br{\theta}\in\sbr{0,1}$ for all $\theta\in\Theta$, that maximizes \eqref{eqn:objective} subject to \eqref{eqn:IC}, \eqref{eqn:IR}, and \eqref{eqn:responsive}. A truth-telling type's value from his own information product, $V\br{\theta,\pi\br{\theta}}$, is given by \eqref{eqn:valueofinformation_responsive}. His value from an information product about another project, $V\br{\theta,\pi\br{\widetilde{\theta}}}$, follows from \eqref{dfn:valueforinformation} after substituting the posterior \eqref{eqn:cross-posterior}. Lastly, note that the reduced problem is feasible: offering $\underline{\pi}\br{\theta}$ with $t\br{\theta}=0$ to each type $\theta$ implements the outside option.

In the next two sections, we characterize solutions to the reduced problem for two types and for a continuum of types.

\section{Two Types}
\label{sec:two-types}

\Cref{subsec:two-types-spillovers} describes information spillovers when the client has one of two types, and \cref{subsec:two-types-optimal-menu} characterizes the optimal menu. Complete proofs are in \cref{app:two-types}.

\subsection{Information Spillovers}
\label{subsec:two-types-spillovers}

Consider a finite type space $\Theta=\cbr{\ell,h}$. We refer to $\theta=\ell$ as the \textit{low} type and $\theta=h$ as the \textit{high} type. The low type has probability $p\in\br{0,1}$, and $u\br{h}>u\br{\ell}>0$.

Let $q$ denote the probability that the high type's project is bad given that the low type's project is good,
\[
    q:=\p\br{\omega_h=B\mid\omega_{\ell}=G}.
\]
We characterize the optimal menu for $0\leq q\leq1-\mu$, so that project states are \textit{positively correlated}: learning one project is good weakly raises the probability that the other is good, that is,
\[
    \p\br{\omega_h=G\mid\omega_{\ell}=G}=1-q\geq\mu,
\]
and likewise learning one project is bad weakly raises the probability that the other is bad. The states are \textit{perfectly} correlated at $q=0$ and \textit{independent} at $q=1-\mu$. We omit the case of \textit{negatively correlated} states ($q>1-\mu$), where a similar analysis applies except that the expert then distorts the recommendation in the bad state, $\pi_B$, rather than in the good state, $\pi_G$, as in \cref{thm:optimal-two-types}.

The joint prior of the two project states is then determined by $\mu$ and $q$ (see \cref{app:two-types-priors}):
\begin{equation}
    \label{eqn:two-types-joint}
    \begin{array}{c|cc}
        \omega_{\ell}\backslash\omega_h & G & B \\
        \hline
        G & \mu-\mu q & \mu q \\
        B & \mu q & 1-\mu-\mu q
    \end{array}
\end{equation}
Because both projects are good with probability $\mu$, the prior is symmetric in the two projects, so the two off-diagonal entries of \eqref{eqn:two-types-joint} equal $\mu q$.

\subsection{Optimal Menu}
\label{subsec:two-types-optimal-menu}

\subsubsection{Efficient Full Surplus Extraction}
The first result gives conditions under which the expert can achieve efficient full surplus extraction. The efficient full surplus extraction menu is $\br{\overline{\pi},\overline t}$, where for each $\theta\in\cbr{\ell,h}$,
\[
    \overline{\pi}_G\br{\theta}=\overline{\pi}_B\br{\theta}=1
    \quad\text{and}\quad
    \overline t\br{\theta}=\overline V\br{\theta}=\br{1-\mu}u\br{\theta}.
\]
The low type cannot profitably buy the high type's information product because its price $\overline V\br{h}$ exceeds his maximum willingness to pay $\overline V\br{\ell}$.

Now consider type $h$ buying $\overline{\pi}\br{\ell}$. This perfectly reveals $\omega_{\ell}$, while his decision concerns $\omega_h$. After learning that $\omega_{\ell}=G$, he takes action $G$, since his own project is then good with probability $1-q\geq\mu$. After learning that $\omega_{\ell}=B$, he either takes action $B$ or continues to take $G$. Following both recommendations gives a correct decision whenever the two project states coincide, which by \eqref{eqn:two-types-joint} occurs with probability $1-2\mu q$. Taking $G$ regardless gives a correct decision with probability $\mu$. The value of this information product is therefore
\[
    V\br{h,\overline{\pi}\br{\ell}}
    =\br{\max\cbr{\mu,1-2\mu q}-\mu}u\br{h}
    =\max\cbr{0,1-\mu-2\mu q}u\br{h}.
\]
Under $\br{\overline{\pi},\overline t}$, the high type receives zero surplus from his own information product, so imitating the low type is profitable if and only if $V\br{h,\overline{\pi}\br{\ell}}$ exceeds the price $\overline t\br{\ell}=\br{1-\mu}u\br{\ell}$.

The preceding arguments give our first result, whose proof is therefore omitted.

\begin{proposition}[Efficient Full Surplus Extraction]
    \label{prop:fullsurplustwotypes}
    The efficient full surplus extraction menu $\br{\overline{\pi},\overline t}$ satisfies \eqref{eqn:IC}, and hence is the unique optimal simple menu, if and only if
    \begin{equation}
        \max\cbr{0,1-\mu-2\mu q}u\br{h}\leq\br{1-\mu}u\br{\ell}.
        \tag{C1}
    \end{equation}
\end{proposition}

If $q\geq\br{1-\mu}/\br{2\mu}$, correlation is weak enough that, even after learning that the low type's project is bad, the high type finds his own project at least as likely to be good as bad. It is then optimal for him to undertake his project whatever $\overline{\pi}\br{\ell}$ recommends, so it has no value to him, and condition (C1) holds for every $u\br{h}$ and $u\br{\ell}$, which gives the region above the dotted horizontal line in \cref{fig:optimalmenu_twotypes}.

\begin{figure}[t!]
    \centering

\begin{tikzpicture}[x=4cm,y=13.5cm,>=Stealth,font=\normalsize]
    \tikzset{
        expertise boundary/.style={black,line width=.9pt},
        expertise dashed/.style={black,dashed,line width=.8pt},
        expertise dotted/.style={black!65,densely dotted,line width=.6pt},
        expertise hatching/.style={pattern={Lines[angle=45,distance=5pt,line width=.28pt]},pattern color=drkblue!40}
    }


    \fill[beaublue!35] (1,0)
    -- plot[domain=1:2.8,samples=160]
    (\x,{.25*(1-1/(1+8*(\x-1)))})
    -- (2.8,.333) -- (1,.333) -- cycle;

    \fill[expertise hatching] (1.43,0)
    -- plot[domain=1.43:1.86,samples=80]
    (\x,{.5*(1-4.44/(1+8*(\x-1)))})
    -- plot[domain=1.86:2.8,samples=100]
    (\x,{.25*(1-1/(1+8*(\x-1)))})
    -- (2.8,0) -- cycle;

    \fill[beaublue] (1,0)
    -- plot[domain=1:1.86,samples=120]
    (\x,{.25*(1-1/(1+8*(\x-1)))})
    -- plot[domain=1.86:1.43,samples=80]
    (\x,{.5*(1-4.44/(1+8*(\x-1)))})
    -- cycle;

    \draw[expertise dotted] (1,.25) -- (2.8,.25);
    \draw[expertise boundary] plot[domain=1:2.8,samples=160]
    (\x,{.25*(1-1/(1+8*(\x-1)))});
    \draw[expertise dashed] plot[domain=1.43:1.86,samples=80]
    (\x,{.5*(1-4.44/(1+8*(\x-1)))});
    \draw[black!45,line width=.45pt] (1,.333) -- (2.8,.333);

    \draw[line width=.95pt,->] (.96,0) -- (2.9,0)
    node[right] {$u\br{h}/u\br{\ell}$};
    \draw[line width=.95pt,->] (1,-.012) -- (1,.365) node[above] {$q$};
    \draw[line width=.5pt] (1,.0035) -- (1,-.0035) node[below=3pt] {$1$};
    \draw[line width=.5pt] (1.43,.0035) -- (1.43,-.0035) node[below=3pt] {$1/(1-p)$};
    \node[left=3pt] at (1,0) {$0$};
    \node[left=3pt] at (1,.25) {$\frac{1-\mu}{2\mu}$};
    \node[left=3pt] at (1,.333) {$1-\mu$};

    \begin{scope}[shift={(3.08,.35)},x=1cm,y=1cm,font=\footnotesize]
        \draw[black!45,line width=.45pt] (-.18,.30) rectangle (6.0,-3.95);

        \draw[fill=beaublue!35,draw=black!25,line width=.3pt] (0,-.12) rectangle (.55,.12);
        \node[anchor=west,inner sep=0pt] at (.78,0) {Efficient full surplus extraction};

        \draw[expertise hatching,draw=black!25,line width=.3pt] (0,-.84) rectangle (.55,-.60);
        \node[anchor=west,inner sep=0pt] at (.78,-.72) {Distorted full surplus extraction};

        \draw[fill=beaublue,draw=black!25,line width=.3pt] (0,-1.62) rectangle (.55,-1.38);
        \node[anchor=west,inner sep=0pt] at (.78,-1.50) {Full information and $U(h)>0$};

        \draw[black!20,line width=.35pt] (0,-1.94) -- (5.38,-1.94);

        \draw[expertise boundary] (0,-2.47) -- (.55,-2.47);
        \node[anchor=west,inner sep=0pt] at (.78,-2.47)
        {$\displaystyle q=\frac{1-\mu}{2\mu}\left(1-\frac{1}{u\br{h} / u\br{\ell}}\right)$};

        \draw[expertise dashed] (0,-3.42) -- (.55,-3.42);
        \node[anchor=west,inner sep=0pt] at (.78,-3.42)
        {$\displaystyle \frac{u\br{h}}{u\br{\ell}}=\frac{1}{\br{1-p}\br{1-2q}}$};
    \end{scope}
\end{tikzpicture}
    \caption{Illustration of the optimal menu $\br{\pi^*,t^*}$ as a function of $u\br{h}/u\br{\ell}$ and $q$.}
    \label{fig:optimalmenu_twotypes}
\end{figure}

When correlation is strong enough, so that $q<\br{1-\mu}/\br{2\mu}$, condition (C1) bounds the relative payoffs of the two types,
\[
    u\br{h}/u\br{\ell}\leq\br{1-\mu}/\br{1-\mu-2\mu q},
\]
which gives the region between the dotted horizontal line and the solid curve in \cref{fig:optimalmenu_twotypes}. The high type compares his saving in fees from imitation, $\overline t\br{h}-\overline t\br{\ell}=\br{1-\mu}\br{u\br{h}-u\br{\ell}}$, with his imitation cost, $\overline V\br{h}-V\br{h,\overline{\pi}\br{\ell}}=2\mu q\,u\br{h}$. Condition (C1) then requires that the saving not exceed the cost. As $u\br{h}$ grows relative to $u\br{\ell}$, the saving eventually exceeds the imitation cost, so he prefers the cheaper expertise unless correlation weakens as well ($q$ rises). As $u\br{h}/u\br{\ell}\to\infty$, efficient full surplus extraction requires correlation to be weak enough that $q\geq\br{1-\mu}/\br{2\mu}$.

\subsubsection{General Case}
Consider now the case in which condition (C1) does not hold, so efficient full surplus extraction is not possible. The following lemma characterizes the properties of an optimal menu.

\begin{lemma}[Structure of $\br{\pi^*,t^*}$]
    \label{lmm:two-types-structure}
    When condition (C1) does not hold, there exists an optimal menu $\br{\pi^*,t^*}$ satisfying the following:
    \begin{enumerate}[label=(\roman*)]
        \item No distortion at the top: $\pi^*\br{h}=\overline{\pi}\br{h}$.
        \item No surplus at the bottom: $U^*\br{\ell}=0$.
        \item The high type's incentive constraint binds: $U^*\br{h}=V\br{h,\pi^*\br{\ell}}-t^*\br{\ell}$.
        \item No distortion in the bad state at the bottom: $\pi_B^*\br{\ell}=1$.
    \end{enumerate}
\end{lemma}

Parts (i)--(iii) are familiar from standard screening problems with one-dimensional types. The high type receives full information, the low type receives no surplus, and the high type's (downward) incentive constraint binds. Part (iv) specifies which recommendation the expert distorts: the expert may distort $\pi_G\br{\ell}$, but she can always recommend that the low type reject a bad project.

For $q>0$, the intuition for distorting the recommendation only in the good state is that $\omega_{\ell}=G$ is more informative about the high type's project than $\omega_{\ell}=B$. A good low-type project raises the probability that the high type's project is good to $\p\br{\omega_h=G\mid\omega_{\ell}=G}=1-q$, whereas a bad low-type project raises the probability that his project is bad only to $\p\br{\omega_h=B\mid\omega_{\ell}=B}=1-\mu q/\br{1-\mu}$, which is smaller because $\mu>1/2$. Since the low type earns the same payoff $u\br{\ell}$ from a correct decision in either state, the expert distorts the recommendation in the state that is more informative about the high type's project, the good state, which weakens his incentive to imitate the most.

As a starting point, consider perfectly correlated states, $q=0$. The expert has no incentive to distort expertise, so she finds it optimal to offer full information at a common flat fee, which is either $\br{1-\mu}u\br{\ell}$ or $\br{1-\mu}u\br{h}$. At the lower fee, both types purchase, and revenue is $\br{1-\mu}u\br{\ell}$. At the higher fee, only the high type purchases, the low type receives $\underline{\pi}\br{\ell}$ at zero price, and revenue is $\br{1-p}\br{1-\mu}u\br{h}$. The lower fee is therefore optimal if and only if $\br{1-p}u\br{h}\leq u\br{\ell}$. The payoff ratio $1/\br{1-p}$ at which the two fees yield the same revenue is the intercept of the dashed curve in \cref{fig:optimalmenu_twotypes}.

The next theorem characterizes the optimal menu for $q>0$. By \cref{lmm:two-types-structure}, it remains to determine the optimal recommendation probability $\pi_G\br{\ell}$, since the transfers $t\br{\ell}$ and $t\br{h}$ then follow from parts (ii) and (iii).

\begin{theorem}[Optimal Menu]
    \label{thm:optimal-two-types}
    Suppose $q>0$ and condition (C1) does not hold. Then every optimal simple menu has $\pi^*\br{h}=\overline{\pi}\br{h}$, $\pi_B^*\br{\ell}=1$, and $t^*\br{\ell}=\br{\mu\pi_G^*\br{\ell}+1-2\mu}u\br{\ell}$. If
    \begin{equation}
        \label{eqn:thm1-inequality}
        \frac{u\br{h}}{u\br{\ell}}\leq\frac{1}{\br{1-p}\br{1-2q}},
    \end{equation}
    there is an optimal menu with no distortion at the bottom and positive surplus for the high type, with $\pi_G^*\br{\ell}=1$ and $t^*\br{h}=2\mu q\,u\br{h}+\br{1-\mu}u\br{\ell}<\overline t\br{h}$.

    Otherwise, there is distorted full surplus extraction, i.e., $t^*\br{h}=\overline t\br{h}$ and the optimal $\pi_G^*\br{\ell} \in \lsbr{(2\mu-1)/\mu, 1}$ with
    \begin{equation}
        \label{eqn:optimal-piG-low}
        \pi_G^*\br{\ell}=\frac{\br{2\mu-1}\br{u\br{h}-u\br{\ell}}}{\mu\br{\br{1-2q}u\br{h}-u\br{\ell}}}.
    \end{equation}

    The optimal simple menu is unique unless \eqref{eqn:thm1-inequality} holds with equality. At equality, each value of $\pi_G^*\br{\ell}$ between the distorted level in \eqref{eqn:optimal-piG-low} and $1$ determines an optimal simple menu.
\end{theorem}

Improving the low type's expertise raises the price he pays but also the high type's gain from imitating him. When \eqref{eqn:thm1-inequality} holds, the expected gain in fees from the low type is at least as large as the expected cost of leaving additional surplus to the high type, so providing full information to both types is optimal. Otherwise, the expert distorts the low type's expertise until the high type's surplus is zero. The distortion cannot be too large: responsiveness \eqref{eqn:responsive} with $\pi_B\br{\ell}=1$ bounds $\pi_G\br{\ell}$ below by $\br{2\mu-1}/\mu$.

The dashed curve in \cref{fig:optimalmenu_twotypes} is the boundary at which \eqref{eqn:thm1-inequality} holds with equality. When the payoff ratio $u\br{h}/u\br{\ell}$ is small, distorting the low type's expertise is less profitable than leaving the high type positive surplus, so the expert provides full information to both types, which is illustrated by the full-information region in \cref{fig:optimalmenu_twotypes}. This region expands as the low type becomes more likely (higher $p$) or correlation weakens (higher $q$). When the payoff ratio is large, providing full information to the low type is not worth the surplus left to the high type, so the expert distorts the low type's expertise and extracts the full surplus from both types, which is illustrated by the distorted full surplus extraction region in \cref{fig:optimalmenu_twotypes}.

When correlation is sufficiently weak, $q\geq\br{1-\mu}/\br{2\mu}$, full information about the low type's project has no value to the high type. Learning the low type's state can change the high type's belief, but undertaking his project remains optimal after either observation. The continuum model in \cref{sec:continuum} defines a distance between projects under which full information about a sufficiently nearby project always has positive value. We show that efficient full surplus extraction is nevertheless possible under conditions on payoffs. When it is not, the optimal menu can generate surplus patterns that cannot arise with two types. In particular, ``middle'' types can be the ones who gain, because their payoffs are not so high that imitation costs them more than it saves in fees, nor so low that the expert distorts their expertise.

\section{Continuum of Types}
\label{sec:continuum}

\Cref{subsec:continuum-spillovers} describes information spillovers when the prior satisfies Markov, homogeneity, and continuity properties across projects. \Cref{subsec:continuum-optimal-menu} characterizes the optimal menu. Complete proofs are in \cref{app:continuum}.

\subsection{Information Spillovers}
\label{subsec:continuum-spillovers}

We normalize the project indices so that $\Theta=\sbr{0,1}$. Suppose the client's type $\theta$ is drawn from a distribution $F$ with a continuous density $f$ that is strictly positive on $\Theta$. His payoff from a correct decision is $u\br{\theta}$, where $u:\Theta\to\rbb_{++}$ is strictly increasing and continuously differentiable.

Consider two projects indexed by $\theta$ and $\theta+\Delta$ in $\Theta$, where $\Delta\geq0$. We call the difference $\Delta$ between their indices the \textit{distance} between the projects.\footnote{If the original index is $\hat{\theta}\in\sbr{0,L}$, we set $\theta=\hat{\theta}/L$. Then a normalized distance $\Delta$ corresponds to an original distance $L\Delta$, so the correlation decay rate $\lambda$ in \cref{prop:markov-prior} is $L$ times the original rate.} The following three properties of the common prior $\p$ describe how information spillovers vary with this distance. Let $\mathcal{F}_\theta:=\sigma\br{\omega_s:0\leq s\leq\theta}$ denote the $\sigma$-algebra generated by the states of all projects with indices up to and including $\theta$. In what follows, $i$ and $j$ denote states in $\Omega$.

\begin{property}[Markov]
    \label{property:markov}
    The prior $\p$ satisfies
    \[
        \p\br{\omega_{\theta+\Delta}=j\mid\mathcal{F}_\theta}
        =\p\br{\omega_{\theta+\Delta}=j\mid\omega_\theta}
        \quad\text{almost surely}.
    \]
\end{property}
By the Markov property, given the state of project $\theta$, the states of projects with lower indices carry no additional information about the state of project $\theta+\Delta$. Learning the states of several projects below $\theta+\Delta$ is therefore no more informative about $\omega_{\theta+\Delta}$ than learning the state of the nearest one alone.

\begin{property}[Homogeneity]
    \label{property:homogeneity}
    The prior $\p$ satisfies
    \[
        \p\br{\omega_{\theta+\Delta}=j\mid\omega_\theta=i}
        =\p\br{\omega_\Delta=j\mid\omega_0=i}.
    \]
\end{property}
By homogeneity, what the state of one project reveals about a project $\Delta$ above it depends only on their distance. For fixed recommendation probabilities, the gain in a client's probability of a correct decision depends only on the distance $\Delta$ between his own project and the one he learns about, while its value to the client scales with his payoff $u\br{\theta}$.\footnote{Allowing spillovers to depend on the projects' indices would introduce an additional screening channel. We shut it down to isolate the trade-off between willingness to pay, determined by $u\br{\theta}$, and information spillovers, which depend only on the distance $\Delta$ between projects and not on their indices.}

\begin{property}[Continuity]
    \label{property:continuity}
    The prior $\p$ satisfies, for every $0 \leq \theta<1$,
    \[
        \lim_{\Delta\downarrow0}\p\br{\omega_{\theta+\Delta}\ne\omega_\theta}=0.
    \]
\end{property}
By continuity, projects at a vanishing distance share the same state with probability approaching one. A client who learns the state of a nearby project can therefore predict his own project's state with a probability of error that approaches zero as the distance between them approaches zero.

Under homogeneity, the conditional state probabilities depend only on $\Delta$. Let
\[
    P_{ij}\br{\Delta}:=\p\br{\omega_{\theta+\Delta}=j\mid\omega_\theta=i}
\]
denote the probability that project $\theta+\Delta$ is in state $j$ given that project $\theta$ is in state $i$, and write $P\br{\Delta}:=\br{P_{ij}\br{\Delta}}_{i,j\in\Omega}$ for the matrix of these probabilities. The next proposition shows that, under \cref{property:markov,property:homogeneity,property:continuity} and equal marginal priors, the common prior $\p$ is the law of a stationary, homogeneous, and reversible two-state Markov chain whose transition matrix $P\br{\Delta}$ is determined by a single parameter $\lambda\geq0$, the rate at which the \textit{correlation} between project states decays with distance.

\begin{proposition}[Markov Prior]
    \label{prop:markov-prior}
    If each project is good with probability $\mu\in\br{0,1}$ and $\p$ satisfies \cref{property:markov,property:homogeneity,property:continuity}, then $\p$ is the law of a stationary, homogeneous, and reversible two-state Markov chain. There is a unique $\lambda\geq0$ such that, for any two projects $\theta,\widetilde{\theta}\in\Theta$ at distance $\Delta=|\widetilde{\theta}-\theta|$ and any $i,j\in\Omega$, $\p\br{\omega_{\widetilde{\theta}}=j\mid\omega_\theta=i}=P_{ij}\br{\Delta}$, for transition matrix
    \begin{equation}
        \label{eqn:transitionfunction}
        P\br{\Delta}=
        \begin{pNiceMatrix}[first-row,first-col]
              & G & B \\
            G & \mu+\br{1-\mu}e^{-\lambda\Delta} & \br{1-\mu}\br{1-e^{-\lambda\Delta}} \\
            B & \mu\br{1-e^{-\lambda\Delta}} & 1-\mu+\mu e^{-\lambda\Delta}
        \end{pNiceMatrix}.
    \end{equation}
\end{proposition}
Let $X_\theta:=\one\cbr{\omega_\theta=G}$ denote the good-state indicator for project $\theta$. Using \eqref{eqn:transitionfunction}, we show that $\operatorname{Corr}\br{X_\theta,X_{\theta+\Delta}}=e^{-\lambda\Delta}$, where $\lambda$ is the \textit{correlation decay rate}. In particular,
\[
    \lim_{\Delta\downarrow0}\frac{1-\operatorname{Corr}\br{X_\theta,X_{\theta+\Delta}}}{\Delta}=\lambda.
\]
That is, the correlation decays exponentially at rate $\lambda$ from one at $\Delta=0$ to $e^{-\lambda}$ at $\Delta=1$, the distance between projects $0$ and $1$. When $\lambda=0$, project states are perfectly correlated at every distance. When $\lambda>0$, correlation falls as projects move farther apart. At any given distance $\Delta>0$, a higher decay rate $\lambda$ means weaker correlation between their states.

Now suppose a client interested in project $\theta+\Delta$ learns that project $\theta$ is good. If he undertakes his own project, he makes an incorrect decision with probability $P_{GB}\br{\Delta}=\br{1-\mu}\br{1-e^{-\lambda\Delta}}$, which is the counterpart of $q$ in the two-type model. This probability is zero at $\Delta=0$ and approaches the prior probability $1-\mu$ of the bad state as $\lambda\Delta\to\infty$. At any fixed distance $\Delta>0$, a higher $\lambda$ makes the state of project $\theta$ less informative about the client's own project, and the two project states become independent as $\lambda\to\infty$. Likewise, for fixed $\lambda>0$, learning that project $\theta$ is good is less informative about the client's own project when the distance $\Delta$ between them is greater.

Stationarity implies that, for fixed states $i$ and $j$, the joint probability
\[
    \p\br{\omega_\theta=i,\omega_{\theta+\Delta}=j}
\]
depends only on the distance $\Delta$ between the projects. Lastly, suppose $\theta-\Delta,\theta+\Delta\in\Theta$. Reversibility implies that, for either state $i\in\Omega$, a client of type $\theta$ has the same posterior about his own project after observing $i$ at $\theta-\Delta$ as after observing $i$ at $\theta+\Delta$.

\paragraph{Posteriors.} Consider a client of type $\theta$ who buys the information product $\pi\br{\widetilde{\theta}}$ about project $\widetilde{\theta}$, and let $\Delta=|\theta-\widetilde{\theta}|$. For each recommendation $s$ sent with positive probability,
substituting \eqref{eqn:transitionfunction} into \eqref{eqn:cross-posterior} gives the client's posterior probability that his own project is good (see the proof of \cref{lmm:cross-value}),
\begin{equation}
    \label{eqn:posterior}
    \nu_s\br{\theta\mid\widetilde{\theta}}
    =e^{-\lambda\Delta}\nu_s\br{\widetilde{\theta}}+\br{1-e^{-\lambda\Delta}}\mu.
\end{equation}
When $\lambda>0$, the client's posterior moves closer to the common prior $\mu$ as the distance between the projects increases and approaches $\nu_s\br{\widetilde{\theta}}$ as that distance falls to zero.

To characterize the optimal menu, we restrict attention to $\mu>1/2$ without loss of generality.\footnote{If $\mu < 1/2$, relabeling the states and actions in an optimal menu for prior $1-\mu$ gives an optimal menu for prior $\mu$, because the symmetry of payoffs preserves each type's willingness to pay for every information product. Some results in this section also hold at $\mu=1/2$, so we sometimes state them for $1/2\le\mu<1$.} Without expertise, the client strictly prefers to undertake his project when $\mu>1/2$ and is indifferent between undertaking and rejecting it when $\mu=1/2$. Suppose a client of type $\theta$ receives recommendation $s=G$ about project $\widetilde{\theta}$. Responsiveness gives $\nu_G\br{\widetilde{\theta}}\geq1/2$, and since $\mu\geq1/2$ as well, by \eqref{eqn:posterior} his posterior $\nu_G\br{\theta\mid\widetilde{\theta}}$ is also at least $1/2$. He therefore undertakes his project after $s=G$, as he does without expertise, so only recommendation $s=B$ can change his decision.

Suppose now that he receives recommendation $s=B$ about project $\widetilde{\theta}$ and $\nu_B\br{\widetilde{\theta}}<1/2$. By \eqref{eqn:posterior}, his posterior $\nu_B\br{\theta\mid\widetilde{\theta}}$ is also below $1/2$ when the projects are sufficiently close, so he strictly prefers to reject his own project. When $\mu>1/2$ and the projects are sufficiently far apart, his posterior may remain above $1/2$ even after recommendation $B$, so he still prefers to undertake his project. Since he undertakes his project after either recommendation, the information product has zero value to him even if it changes his beliefs.

When $\mu=1/2$, his posterior after $s=B$ is a weighted average of $1/2$ and $\nu_B\br{\widetilde{\theta}}<1/2$. The weight on $\nu_B\br{\widetilde{\theta}}$ is positive at every finite distance, so his posterior remains below $1/2$ and he strictly prefers to reject his project. Without expertise, the client is indifferent between undertaking and rejecting his project, so even a small change in his belief breaks the tie.

\paragraph{Value of Expertise.} Under the Markov prior, the value of a responsive information product $\pi\br{\widetilde{\theta}}$ to a client of type $\theta$ is obtained in closed form.

\begin{lemma}[Value of Expertise]
    \label{lmm:cross-value}
    Suppose $1/2\leq\mu<1$ and $\p$ is the Markov prior with correlation decay rate $\lambda\geq0$. Then
    \begin{equation}
        \label{eqn:general-cross-value}
        \begin{aligned}
            V\br{\theta,\pi\br{\widetilde{\theta}}}
            =u\br{\theta}&\Bigl(\mu\pi_G\br{\widetilde{\theta}}+\br{1-\mu}\pi_B\br{\widetilde{\theta}}-\mu\\
            &\quad-2\mu\br{1-\mu}\br{1-e^{-\lambda\Delta}}
                \br{\pi_G\br{\widetilde{\theta}}+\pi_B\br{\widetilde{\theta}}-1}\Bigr)_+,
        \end{aligned}
    \end{equation}
    where $\br{x}_+:=\max\cbr{x,0}$ denotes the positive part of $x$.
\end{lemma}
To interpret \eqref{eqn:general-cross-value}, suppose type $\theta$ buys $\pi\br{\widetilde{\theta}}$ that satisfies \eqref{eqn:responsive}. Following $G$ means undertaking the project, which is what he does without expertise, so his gain comes from following a recommendation to reject. By following $B$, he avoids undertaking a bad project when $\omega_\theta=B$ but forgoes a good project when $\omega_\theta=G$. His gain in the probability of a correct decision from following the recommendation is therefore
\begin{equation}
    \label{eqn:prob-gain}
    \p\br{s=B,\omega_\theta=B}-\p\br{s=B,\omega_\theta=G}.
\end{equation}

At $\Delta=0$, the recommendation concerns his own project, so these probabilities are $\br{1-\mu}\pi_B\br{\widetilde{\theta}}$ and $\mu\br{1-\pi_G\br{\widetilde{\theta}}}$. Their difference gives the first three terms in \eqref{eqn:general-cross-value}.

At distance $\Delta>0$, the two projects may have different states. When they do, a recommendation that is correct for project $\widetilde{\theta}$ is incorrect for project $\theta$. The probability that the states differ is
\[
    \p\br{\omega_{\widetilde{\theta}}\neq\omega_\theta}
    =\mu P_{GB}\br{\Delta}+\br{1-\mu}P_{BG}\br{\Delta}
    =2\mu\br{1-\mu}\br{1-e^{-\lambda\Delta}}.
\]
The information product recommends rejecting the project with probability $\pi_B\br{\widetilde{\theta}}$ when project $\widetilde{\theta}$ is bad and $1-\pi_G\br{\widetilde{\theta}}$ when it is good. By \eqref{eqn:responsive} and $\mu\geq1/2$, the difference $\pi_B\br{\widetilde{\theta}}-\br{1-\pi_G\br{\widetilde{\theta}}}$ is nonnegative. Multiplying this difference by $\p\br{\omega_{\widetilde{\theta}}\neq\omega_\theta}$ gives the loss in type $\theta$'s probability of a correct decision that is subtracted in \eqref{eqn:general-cross-value}.

\subsection{Optimal Menu}
\label{subsec:continuum-optimal-menu}

\subsubsection{Efficient Full Surplus Extraction}

As in the two-type model, consider the menu $\br{\overline{\pi},\overline t}$ that offers each type $\theta$ full information about his project at his willingness to pay, $\overline t\br{\theta}=\br{1-\mu}u\br{\theta}$.

Suppose type $\theta>0$ instead buys full information about project $\theta-\Delta$, so he saves $\overline t\br{\theta}-\overline t\br{\theta-\Delta}=\br{1-\mu}\br{u\br{\theta}-u\br{\theta-\Delta}}$ in fees. For sufficiently small $\Delta>0$, he follows both recommendations. Setting $\pi_G=\pi_B=1$ in \eqref{eqn:general-cross-value} then gives his loss of information value,
\[
        V\br{\theta,\overline{\pi}\br{\theta}}
        -V\br{\theta,\overline{\pi}\br{\theta-\Delta}}
        =2\mu\br{1-\mu}\br{1-e^{-\lambda\Delta}}u\br{\theta}.
\]
Then dividing the saving in fees and the loss of information value by $\Delta$, and letting $\Delta\downarrow0$, gives his marginal saving in fees and marginal loss of information value
\[
    \underbrace{\br{1-\mu}u'\br{\theta}}_{\text{marginal gain from imitation}}
    \quad\text{and}\quad
    \underbrace{2\mu\br{1-\mu}\lambda u\br{\theta}}_{\text{marginal cost of imitation}}.
\]
If the marginal gain from imitation exceeds the marginal cost, type $\theta$ prefers the cheaper expertise for all sufficiently small $\Delta>0$. Incentive compatibility therefore requires the marginal cost to be at least the marginal gain at every type. The next proposition shows that this requirement is also sufficient.

\begin{proposition}[Efficient Full Surplus Extraction]
    \label{prop:fullsurpluscontinuum}
    Suppose $1/2\leq\mu<1$. The efficient full surplus extraction menu $\br{\overline{\pi},\overline t}$ is incentive compatible, and hence is the unique optimal simple menu $dF$-almost everywhere, if and only if
    \begin{equation}
        \label{eqn:fullsurpluscontinuum}
        u'\br{\theta}\leq2\mu\lambda u\br{\theta}\quad\text{for all }\theta\in\Theta.
        \tag{C2}
    \end{equation}
\end{proposition}
Condition \eqref{eqn:fullsurpluscontinuum} compares how fast willingness to pay rises across types with how fast information spillovers decay with distance. Since $\overline t\br{\theta}=\br{1-\mu}u\br{\theta}$, fees under $\br{\overline{\pi},\overline t}$ rise at the proportional rate $u'\br{\theta}/u\br{\theta}$. When this rate is at most $2\mu\lambda$ for every type, the fee a client saves by buying full information about a nearby project is no greater than his loss of information value.

\subsubsection{General Case}

First consider perfect spillovers, $\lambda=0$, so all projects share the same state. Fix a responsive information product $\pi\br{0}$ about project $0$. Any type $\theta$ who follows its recommendations makes a correct decision with probability $\mu\pi_G\br{0}+\br{1-\mu}\pi_B\br{0}$, compared with $\mu$ without expertise. Define the quality of $\pi\br{0}$ as this gain divided by the gain from full information,
\[
    x=\frac{\mu\pi_G\br{0}+\br{1-\mu}\pi_B\br{0}-\mu}{1-\mu}
    \in\sbr{0,1}.
\]
Thus the same information product offers every type the same quality $x$, and by \eqref{eqn:general-cross-value}, type $\theta$'s willingness to pay is $V\br{\theta,\pi\br{0}}=x\overline V\br{\theta}=x\br{1-\mu}u\br{\theta}$. From the client's perspective, an information product of quality $x$ is equivalent to receiving full information with probability $x$. The expert therefore faces the standard problem of selling one good to one buyer with private value $\overline V\br{\theta}$, for which a single posted price is optimal. Since $u$ is strictly increasing, only types $\theta \geq \widehat{\theta}$ pay a flat fee $\br{1-\mu}u\br{\widehat{\theta}}$, giving revenue $\br{1-\mu}u\br{\widehat{\theta}}\br{1-F\br{\widehat{\theta}}}$.

\begin{proposition}[Flat Fee]
    \label{prop:flat-fee}
    Suppose $1/2\leq\mu<1$ and $\lambda=0$. An optimal menu offers full information at a common flat fee $t^*=\br{1-\mu}u\br{\widehat{\theta}}$ where
    \[
        \widehat{\theta} \in \argmax_{\theta\in\Theta}u\br{\theta}\br{1-F\br{\theta}}.
    \]
    That is, types $\theta\geq\widehat{\theta} $ buy $\overline{\pi}\br{\theta}$ at the price $t^*$, and types $\theta<\widehat{\theta} $ receive $\underline{\pi}\br{\theta}$ at zero price.
\end{proposition}

Next we consider imperfect spillovers, $\lambda>0$. As in parts (i) and (iv) of \cref{lmm:two-types-structure}, the expert can restrict attention to menus in which she always recommends rejecting a bad project.

\begin{lemma}[No Distortion in the Bad State]
    \label{lmm:one-parameter}
    For any IC and IR simple menu $\br{\pi,t}$, there exists an IC and IR simple menu $\br{\widehat{\pi},t}$ such that, for every $\theta\in\Theta$,
    \[
        \widehat{\pi}_B\br{\theta}=1
        \quad\text{and}\quad
        V\br{\theta,\widehat{\pi}\br{\theta}}=V\br{\theta,\pi\br{\theta}}.
    \]
\end{lemma}

Among information products that give type $\theta$ the same probability of a correct decision about his own project, one that always recommends rejecting a bad project is least valuable to every other type. Any distortion then takes the form of sometimes recommending that he reject a good project.

To characterize the optimal menu under imperfect spillovers, we impose two assumptions on payoffs and the type distribution.

\begin{assumption}[Payoffs]
    \label{asm:payoff}
    The payoff function $u$ is concave\footnote{Concavity can be relaxed to log-concavity together with the requirement that $u'\br{\theta}-2\br{1-\mu}\lambda u\br{\theta}$ be nonincreasing in $\theta$, under the same growth and regularity conditions.} and satisfies
    \begin{equation}
        \label{eqn:continuum-growth}
        u'\br{1}\geq\lambda u\br{1}.
    \end{equation}
\end{assumption}

Concavity implies that the proportional growth rate $u'\br{\theta}/u\br{\theta}$ is nonincreasing in $\theta$. Since this rate is lowest at $\theta=1$, \eqref{eqn:continuum-growth} gives $u'\br{\theta}/u\br{\theta}\geq\lambda$ for every type. Payoffs therefore grow across projects at least as fast as the correlation between project states decays. Within the full information region of the optimal menu, we show that this condition makes local downward IC sufficient for global downward IC. Concavity also implies that \eqref{eqn:fullsurpluscontinuum} holds if and only if $u'\br{0}\leq2\mu\lambda u\br{0}$, because the rate $u'\br{\theta}/u\br{\theta}$ is highest at $\theta=0$.

The second assumption and the cutoffs in the optimal menu are stated in terms of the \textit{informational virtual value} with upper cutoff $b\in\rsbr{0,1}$, $\varphi_b:\sbr{0,b}\to\rbb$, defined by
\begin{equation}
    \label{eqn:informational-virtual-value}
    \varphi_b\br{\theta}
    :=u\br{\theta}
    -\br{u'\br{\theta}-2\br{1-\mu}\lambda u\br{\theta}}
      \frac{F\br{b}-F\br{\theta}}{f\br{\theta}}
    \quad\text{for all }\theta\in\sbr{0,b}.
\end{equation}
For $b=1$, we write $\varphi:=\varphi_1$ and call it simply the informational virtual value.

\begin{assumption}[Regularity]
    \label{asm:regularity}
    The informational virtual value $\varphi:\Theta \to \rbb$ is strictly increasing.\footnote{\cref{asm:payoff} and a nondecreasing hazard rate $f\br{\theta}/\br{1-F\br{\theta}}$ imply \cref{asm:regularity}.}
\end{assumption}

Under \cref{asm:payoff}, we show that regularity extends to every upper cutoff: $\varphi_b$ is strictly increasing on $\sbr{0,b}$ for every $b\in\rsbr{0,1}$.

When $\lambda=0$, information spills over perfectly across projects. By \cref{prop:flat-fee}, the expert then faces the standard monopoly screening problem, and $\varphi$ is the Myersonian virtual value for $u\br{\theta}$ \citep{myerson1981optimal}. Each term in $\varphi\br{\theta}$ is measured per unit gain in type $\theta$'s probability of a correct decision. Such a gain raises his value by $u\br{\theta}$, and when $\lambda>0$ it also raises nearby types' marginal cost of imitation by $2\br{1-\mu}\lambda u\br{\theta}$. An improvement in type $\theta$'s expertise is therefore worth less to higher types than under perfect spillovers. It also costs the expert less in information rent, so $\varphi$ is (weakly) higher than the Myersonian virtual value.

We now characterize the optimal menu.\footnote{We state the theorem for $\mu>1/2$. At $\mu=1/2$, the same cutoff conditions and transfer formulas describe an optimal menu, although the recommendation probabilities need not be unique.}

\begin{theorem}[Optimal Menu]
    \label{thm:characterizationcontinuumtypes} Fix $\mu\in\br{1/2,1}$ and $\lambda>0$. Suppose \cref{asm:payoff,asm:regularity} hold, and \eqref{eqn:fullsurpluscontinuum} does not hold, that is, $u'\br{0}>2\mu\lambda u\br{0}$. Then there is a unique pair of cutoffs $0\leq\sutheta<\sotheta\leq1$ such that
    \begin{align}
        \varphi_{\sotheta}\br{\sutheta}&\geq0,
        &&\text{with equality if }\sutheta>0,
        \label{eqn:cutoff-a}\\
        u\br{\sotheta}-u\br{\sutheta}
        &\geq\int_{\sutheta}^{\sotheta}2\mu\lambda u\br{r}\,dr,
        &&\text{with equality if }\sotheta<1.
        \label{eqn:cutoff-b}
    \end{align}
    An optimal simple menu $\br{\pi^*,t^*}$ has $\pi_B^*\br{\theta}=1$ for every $\theta\in\Theta$ and satisfies the following:
    \begin{enumerate}[label=(\roman*)]
        \item (Distorted full surplus extraction.) For every $\theta<\sutheta$, if any,
        \[
            \pi_G^*\br{\theta}=\frac{2\mu-1}{\mu}\frac{u'\br{\theta}}{u'\br{\theta}-2\br{1-\mu}\lambda u\br{\theta}}
            \quad\text{and}\quad t^*\br{\theta}=V\br{\theta,\pi^*\br{\theta}},
        \]
        where $\br{2\mu-1}/\mu<\pi_G^*\br{\theta}<1$ and $\pi_G^*$ is strictly increasing in $\theta$.
        \item (Full information.) For every $\theta\in\sbr{\sutheta,\sotheta}$,
        \[
            \pi_G^*\br{\theta}=1
            \quad\text{and}\quad
            t^*\br{\theta}=\br{1-\mu}u\br{\sutheta}+\int_{\sutheta}^{\theta}2\mu\br{1-\mu}\lambda u\br{r}\,dr.
        \]
        \item (Efficient full surplus extraction.) For every $\theta>\sotheta$, if any,
        \[
            \pi_G^*\br{\theta}=1
            \quad\text{and}\quad t^*\br{\theta}=\br{1-\mu}u\br{\theta}.
        \]
    \end{enumerate}
    Every optimal simple menu agrees with $\br{\pi^*,t^*}$ $dF$-almost everywhere.
\end{theorem}

\begin{figure}[t!]
    \centering
\tikzset{
    menu axis/.style={black, line width=.95pt, ->},
    menu tick/.style={black, line width=.5pt},
    menu tick label/.style={anchor=base, yshift=-15pt, inner sep=0pt},
    menu intercept label/.style={scale=1.15},
    menu guide/.style={black!65, densely dotted, line width=.6pt},
    perfect curve/.style={drkred, line width=1.05pt, dash pattern=on 5pt off 3pt},
    intermediate curve/.style={drkblue, line width=1.3pt, dashdotted},
    weak curve/.style={black!45, line width=3pt}
}

\pgfmathsetmacro{\menulow}{(97-15*sqrt(33))/40}
\pgfmathsetmacro{\menuhigh}{(15*sqrt(33)-53)/40}
\def\menucritical{.55}
\def\menuflat{.4}
\def\menugapscale{2.4}
\newcommand{\menuvaluefee}[1]{.04+.2*(#1)+.065*(#1)*(1-(#1))}
\newcommand{\menudisplayfee}[1]{(\menuvaluefee{#1})-\menugapscale*2*((#1)-\menulow)*(\menuhigh-(#1))/15}
\newcommand{\menudistortedfee}[1]{(\menuvaluefee{#1})*3*(.2+(#1))/(2.8-(#1))}
\pgfmathsetmacro{\valueatcritical}{\menuvaluefee{\menucritical}}
\pgfmathsetmacro{\feeatcritical}{\menudisplayfee{\menucritical}}
\pgfmathsetmacro{\flatfee}{\menuvaluefee{\menuflat}}
\pgfmathsetmacro{\pricelimit}{\menudistortedfee{\menulow}}
\pgfmathsetmacro{\priceatlow}{\menuvaluefee{\menulow}}
\def\qualityfloor{.75}
\pgfmathsetmacro{\qualitylimit}{9/(4*(2.8-\menulow))}

\newcommand{\menuinformationpanel}{%
    \begin{tikzpicture}[x=5.95cm,y=12.75cm,>=Stealth,font=\normalsize]
        \path[use as bounding box] (-.173,.65067) rectangle (1.144,1.11600);
        \draw[menu axis] (-.02,.704) -- (1.06,.704) node[right] {$\theta$};
        \draw[menu axis] (0,.699) -- (0,1.06037) node[right=3pt,inner sep=1pt] {$\pi_G^*\br{\theta}$};
        \draw[menu guide] (\menulow,.704) -- (\menulow,1);
        \draw[menu guide] (\menuflat,.704) -- (\menuflat,1);

        \draw[weak curve] (0,1) -- (1,1);
        \draw[intermediate curve] plot[domain=0:\menulow,samples=65]
            (\x,{9/(4*(2.8-\x))});
        \draw[intermediate curve] (\menulow,1) -- (1,1);
        \draw[perfect curve] (0,\qualityfloor) -- (\menuflat,\qualityfloor);
        \draw[perfect curve] (\menuflat,1) -- (1,1);
        \draw[drkblue,fill=white,line width=.85pt]
            (\menulow,\qualitylimit) circle[radius=1.7pt];
        \fill[drkblue] (\menulow,1) circle[radius=1.7pt];
        \draw[drkred,fill=white,line width=.85pt]
            (\menuflat,\qualityfloor) circle[radius=1.7pt];
        \fill[drkred] (\menuflat,1) circle[radius=1.7pt];

        \foreach \position/\cutoff in {0/0,\menulow/\sutheta,\menuflat/\widehat{\theta},1/1}{
            \draw[menu tick] (\position,.707) -- (\position,.701);
            \node[menu tick label] at (\position,.704) {$\cutoff$};
        }
        \draw[menu tick] (-.006,\qualityfloor) -- (.006,\qualityfloor);
        \node[menu intercept label,left=3pt] at (0,\qualityfloor) {$\frac{2\mu-1}{\mu}$};
        \draw[menu tick] (-.006,1) -- (.006,1);
        \node[left=3pt] at (0,1) {$1$};
    \end{tikzpicture}%
}

\newcommand{\menutransferpanel}{%
    \begin{tikzpicture}[x=5.95cm,y=17.68cm,>=Stealth,font=\normalsize]
        \path[use as bounding box] (-.173,-.03846) rectangle (1.144,.29712);
        \fill[beaublue!75] plot[domain=\menulow:\menuhigh,samples=80] (\x,{\menuvaluefee{\x}})
            -- plot[domain=\menuhigh:\menulow,samples=80] (\x,{\menudisplayfee{\x}})
            -- cycle;
        \draw[menu axis] (-.02,0) -- (1.06,0) node[right] {$\theta$};
        \draw[menu axis] (0,-.004) -- (0,.257) node[right=3pt,inner sep=1pt] {$t^*\br{\theta}$};
        \draw[menu guide] (\menulow,0) -- (\menulow,\priceatlow);
        \draw[menu guide] (\menuflat,0) -- (\menuflat,\flatfee);
        \draw[menu guide] (\menucritical,0) -- (\menucritical,\feeatcritical);
        \draw[menu guide] (\menuhigh,0) -- (\menuhigh,{\menuvaluefee{\menuhigh}});

        \draw[weak curve] plot[domain=0:1,samples=80] (\x,{\menuvaluefee{\x}});
        \node[menu intercept label,above left=2pt,inner sep=0pt] at (.97,{\menuvaluefee{.97}})
            {$\scriptstyle t\br{\theta}=\br{1-\mu}u\br{\theta}$};
        \draw[menu tick] (-.009,.04) -- (0,.04);
        \node[menu intercept label,rotate=0,anchor=south,inner sep=0pt] at ([xshift=-5pt]-0.13,.04)
            {$\scriptstyle\br{1-\mu}u\br{0}$};
        \draw[intermediate curve] plot[domain=0:\menulow,samples=65]
            (\x,{\menudistortedfee{\x} + .006*(1-\x/\menulow)^2});
        \draw[intermediate curve] plot[domain=\menulow:\menuhigh,samples=80]
            (\x,{\menudisplayfee{\x}});
        \draw[intermediate curve] plot[domain=\menuhigh:1,samples=40] (\x,{\menuvaluefee{\x}});
        \draw[black!75,line width=.65pt] (\menucritical,\feeatcritical)
            -- (\menucritical,\valueatcritical);
        \draw[perfect curve] ([yshift=2.2pt]0,0) -- ([yshift=2.2pt]\menuflat,0);
        \draw[perfect curve] (\menuflat,\flatfee) -- (1,\flatfee);
        \draw[drkblue,fill=white,line width=.85pt]
            (\menulow,\pricelimit) circle[radius=1.7pt];
        \fill[drkblue] (\menulow,\priceatlow) circle[radius=1.7pt];
        \draw[drkred,fill=white,line width=.85pt]
            ([yshift=2.2pt]\menuflat,0) circle[radius=1.7pt];
        \fill[drkred] (\menuflat,\flatfee) circle[radius=1.7pt];
        \foreach \position/\cutoff in {0/0,\menulow/\sutheta,\menuflat/\widehat{\theta},\menucritical/\ctheta,\menuhigh/\sotheta,1/1}{
            \draw[menu tick] (\position,.0024) -- (\position,-.0024);
            \node[menu tick label] at (\position,0) {$\cutoff$};
        }
    \end{tikzpicture}%
}

\newcommand{\menusharedlegend}{%
    {\small
        \setlength{\fboxrule}{.45pt}\setlength{\fboxsep}{4pt}%
        \fcolorbox{black!45}{white}{\strut
            \tikz[baseline=-.4ex]{\draw[perfect curve] (0,0)--(.64,0);}~Perfect spillovers\hspace{1.2em}%
            \tikz[baseline=-.4ex]{\draw[intermediate curve] (0,0)--(.64,0);}~Intermediate spillovers\hspace{1.2em}%
            \tikz[baseline=-.4ex]{\draw[weak curve] (0,0)--(.64,0);}~Weak spillovers}%
    }%
}
    \menusharedlegend\par\smallskip
    \begin{subfigure}[t]{\figurepanelwidth}
        \centering
        \menuinformationpanel
        \caption{Expertise}
        \label{fig:menu-information}
    \end{subfigure}\hfill
    \begin{subfigure}[t]{\figurepanelwidth}
        \centering
        \menutransferpanel
        \caption{Fees}
        \label{fig:menu-fees}
    \end{subfigure}
    \caption{Optimal menu for perfect $\br{\lambda=0}$, intermediate (low $\lambda$), and weak (high $\lambda$) spillovers.}
    \label{fig:continuum-menus}
\end{figure}

A new feature of the continuum model is that middle types can be the only clients who earn positive surplus, even though higher types value full information about their own projects more. Middle types receive full information and pay a convex fee because the marginal cost of imitation rises with type. The fee can eventually reach the client's willingness to pay for full information. From that cutoff onward, the expert extracts the full surplus from every higher type. The expert also serves lower types, charging each of them the full value of his distorted expertise.

In part (i), a type $\theta < \sutheta$ buys distorted expertise, which sometimes recommends rejecting a good project but never recommends undertaking a bad one. The expert distorts expertise just enough that, for a marginally higher type considering the same information product, the gain from his higher payoff is exactly offset by his loss of information value, since the product concerns a different project. She can then charge each type $\theta < \sutheta$ the full value of his information product, which is a strictly positive fee, and he earns zero surplus. Since $u'\br{\theta}/u\br{\theta}$ is decreasing, the marginal cost of imitation rises relative to the marginal gain. The expert can therefore offer higher types in this region more informative expertise and still charge each client the full value of his distorted expertise.

The dash-dotted curves for intermediate spillovers in \cref{fig:menu-information} and \cref{fig:menu-fees} show the resulting increases in $\pi_G^*\br{\theta}$ and $t^*\br{\theta}$, respectively, when information spillovers are imperfect but not sufficiently weak for efficient full surplus extraction. Throughout the region, $\pi_G^*\br{\theta}$ stays above $\br{2\mu-1}/\mu$, the lower bound that \eqref{eqn:responsive} places on $\pi_G$ when $\pi_B=1$, so each client strictly gains from following the recommendations. If $\sutheta>0$, both the recommendation probability $\pi_G^*\br{\theta}$ and the fee $t^*\br{\theta}$ jump at $\sutheta$, after which the expert provides full information.

In part (ii), once type $\sutheta$ receives full information, the expert finds it optimal to provide full information to every higher type as well. The fee starts at type $\sutheta$'s full willingness to pay, $\overline V\br{\sutheta}$. To ensure incentive compatibility of types between the two cutoffs, she sets fees so that, at the margin, a client who imitates a lower type saves exactly as much in fees as he loses in information value. The fee therefore rises at the marginal cost of imitation, $t^{*\prime}\br{\theta}=2\mu\br{1-\mu}\lambda u\br{\theta}$. This rate increases with the client's payoff from a correct decision, so the fee is strictly convex.

Willingness to pay for full information, however, rises at the nonincreasing rate $\overline V'\br{\theta}=\br{1-\mu}u'\br{\theta}$ because $u$ is concave. The client's surplus, $\overline V\br{\theta}-t^*\br{\theta}$, is therefore strictly concave between the cutoffs, as shown in \cref{fig:continuum-surplus}. Surplus is zero at $\sutheta$ and grows while willingness to pay rises faster than the fee, that is, while $u'\br{\theta}>2\mu\lambda u\br{\theta}$. We call the type at which client surplus is highest the \textit{critical type}, $\ctheta$ (see \cref{fig:continuum-surplus}). At an interior maximum, $u'\br{\ctheta}=2\mu\lambda u\br{\ctheta}$. Between $\ctheta$ and $\sotheta$, surplus falls even though willingness to pay keeps rising, because the fee now rises faster.

\begin{figure}[t!]
    \centering
\tikzset{
    surplus axis/.style={black,line width=.95pt,->,>=Stealth},
    surplus tick/.style={black,line width=.5pt},
    surplus curve/.style={drkblue,line width=1.15pt,line join=round},
    surplus guide/.style={black!65,densely dotted,line width=.6pt}
}
\newcommand{\drawsurpluspanel}[5]{%
    \begin{tikzpicture}[x=6.05cm,y=188cm,>=Stealth,font=\normalsize]
        \path[use as bounding box] (-.105,-.00315) rectangle (1.145,.022);
        \pgfmathsetmacro{\surplusfactor}{(#5)*(1-(#5))*(#4)}
        \pgfmathsetmacro{\peakheight}{\surplusfactor*((#2)-(#1))^2}
        \pgfmathsetmacro{\endheight}{\surplusfactor*((#3)-(#1))*(2*(#2)-(#1)-(#3))}
        \fill[beaublue!75] (#1,0) -- plot[domain=#1:#3,samples=65]
            (\x,{\surplusfactor*(\x-(#1))*(2*(#2)-(#1)-\x)}) -- (#3,0) -- cycle;
        \draw[surplus axis] (-.025,0) -- (1.075,0) node[right] {$\theta$};
        \draw[surplus axis] (0,-.00042) -- (0,.0193) node[above] {$U^*\br{\theta}$};
        \draw[surplus guide] (#2,0) -- (#2,\peakheight);
        \ifdim#3pt=1pt
            \draw[surplus guide] (#3,0) -- (#3,\endheight);
        \fi
        \draw[surplus curve] (0,0) -- (#1,0)
            -- plot[domain=#1:#3,samples=65]
            (\x,{\surplusfactor*(\x-(#1))*(2*(#2)-(#1)-\x)});
        \ifdim#3pt<1pt
            \draw[surplus curve] (#3,0) -- (1,0);
        \else
            \fill[drkblue] (#3,\endheight) circle[radius=1.55pt];
        \fi
        \fill[drkblue] (#2,\peakheight) circle[radius=1.55pt];
        \foreach \position/\cutoff in {0/0,#1/\sutheta,#2/\ctheta}{
            \draw[surplus tick] (\position,.00023) -- (\position,-.00023)
                node[below=3pt,inner sep=1pt] {$\cutoff$};
        }
        \ifdim#3pt<1pt
            \draw[surplus tick] (#3,.00023) -- (#3,-.00023)
                node[below=3pt,inner sep=1pt] {$\sotheta$};
            \draw[surplus tick] (1,.00023) -- (1,-.00023)
                node[below=3pt,inner sep=1pt] {$1$};
        \else
            \draw[surplus tick] (1,.00023) -- (1,-.00023)
                node[below=3pt,inner sep=1pt] {$\sotheta=1$};
        \fi
    \end{tikzpicture}%
}

\newcommand{\surplusthreeregions}{\drawsurpluspanel{.270789008}{.550000000}{.829210992}{.833333333}{.80}}
\newcommand{\surpluspositivetop}{\drawsurpluspanel{0.362135402}{0.701875902}{1}{.660000000}{.84}}
    \begin{subfigure}[t]{\figurepanelwidth}
        \centering
        \surplusthreeregions
        \caption{$\theta^*<1$: Parts (i)--(iii)}
        \label{fig:surplus-three}
    \end{subfigure}\hfill
    \begin{subfigure}[t]{\figurepanelwidth}
        \centering
        \surpluspositivetop
        \caption{$\theta^*=1$: Parts (i)--(ii)}
        \label{fig:surplus-top}
    \end{subfigure}
    \caption{Client surplus under the optimal menu.}
    \label{fig:continuum-surplus}
\end{figure}

Part (iii) applies when surplus falls to zero at some $\sotheta<1$, as in \cref{fig:surplus-three}. At this cutoff, the fee again equals the client's willingness to pay for full information. The strict form of condition \eqref{eqn:fullsurpluscontinuum}, $u'\br{\theta}<2\mu\lambda u\br{\theta}$, already holds for every $\theta>\ctheta$. Types strictly between $\ctheta$ and $\sotheta$ nevertheless earn positive surplus, because condition \eqref{eqn:fullsurpluscontinuum} rules out profitable imitation only when the lower type also pays his full willingness to pay. Between $\ctheta$ and $\sotheta$, a nearby lower type pays less than his full willingness to pay. If the expert charged type $\theta$ his full willingness to pay, he could still earn positive surplus by buying full information about a sufficiently close lower type's project at the fee offered to that type. At $\sotheta$ and above, each client pays his full willingness to pay for full information. The intuition for efficient full surplus extraction in \cref{prop:fullsurpluscontinuum} therefore applies to every type at or above $\sotheta$.

Lastly, we interpret the two conditions in \cref{thm:characterizationcontinuumtypes} that determine the cutoffs. By \eqref{eqn:cutoff-b}, the upper cutoff is the type at which the convex fee catches up with the client's willingness to pay for full information, or the highest type if it never does. By \eqref{eqn:cutoff-a}, the expert distorts expertise only below $\sutheta$, the type at which the informational virtual value with upper cutoff $\sotheta$ reaches zero, or the lowest type if it is never negative. Only types between the two cutoffs earn positive information rent, which is why, unlike in standard screening problems, the informational virtual value accounts only for the mass $F\br{\sotheta}-F\br{\theta}$ of higher types up to $\sotheta$, not the mass $1-F\br{\theta}$ of all higher types. Moreover, note that $\varphi_{\sotheta}\br{\theta}\geq\varphi\br{\theta}$, where $\varphi$ is higher than the Myersonian virtual value. The expert therefore leaves less information rent than under perfect spillovers, so she provides full information to more types. That is, $\sutheta\leq\widehat{\theta}$, where $\widehat{\theta}$ is the optimal cutoff under perfect spillovers in \cref{prop:flat-fee} (see \cref{fig:continuum-menus}).

\subsubsection{Flat Fees and Value-Based Fees}

We now return to the question of when an expert charges a common flat fee or value-based fees. \Cref{fig:continuum-menus} illustrates the expertise clients receive and the fees they pay under perfect, intermediate, and weak spillovers across projects.

Under perfect spillovers ($\lambda=0$), \cref{prop:flat-fee} shows that types at or above $\widehat{\theta}$ buy full information at a common flat fee, while lower types receive $\underline{\pi}\br{\theta}$ at zero price. To keep $\pi_B=1$ across all three cases, we plot $\pi_G=\br{2\mu-1}/\mu$ in \cref{fig:menu-information} for types $\theta<\widehat{\theta}$ under perfect spillovers. At this lower bound, \eqref{eqn:responsive} binds and the probability of a correct decision remains $\mu$, so this information product is payoff-equivalent to $\underline{\pi}\br{\theta}$. The dashed line in \cref{fig:menu-fees} shows the corresponding fee schedule, with a zero fee below $\widehat{\theta}$ and a common flat fee of $\br{1-\mu}u\br{\widehat{\theta}}$ at and above $\widehat{\theta}$.

With intermediate spillovers (low $\lambda$), the expert charges value-based fees to low and high types, as shown in \cref{fig:menu-fees}. Each type $\theta<\sutheta$ pays less than his willingness to pay for full information, $\overline V\br{\theta}$. His fee nevertheless equals the full value of the distorted expertise he receives, so he earns zero surplus. Each type above $\sotheta$ receives full information and pays his full willingness to pay for it, so he too earns zero surplus.

Each client strictly between the cutoffs receives full information and pays a \textit{common base fee}, $\overline V\br{\sutheta}$, plus a \textit{project-specific fee} that rises at a rate proportional to his payoff from a correct decision, $u\br{\theta}$. Clients with more at stake therefore pay a larger project-specific fee for advice about their own projects, but the total fee still leaves each client with positive surplus. Lastly, when spillovers are sufficiently weak (high $\lambda$) for \eqref{eqn:fullsurpluscontinuum} to hold, every type buys full information at the value-based fee $\overline V\br{\theta}=\br{1-\mu}u\br{\theta}$, leaving each client with zero surplus.

\section{Conclusions}
\label{sec:conclusion}

How much an expert can charge for advice tailored to a client's own decision depends not only on his willingness to pay but also on how useful that advice would be to other clients. If expertise about one project is just as useful for every other project, the expert charges a common flat fee. If it is of little use for other projects, she charges each client his full willingness to pay. Between these extremes, clients with intermediate stakes pay a common base fee plus a project-specific fee. Because clients with more at stake lose more by buying expertise about another project, this project-specific fee can rise faster than willingness to pay. Surplus can therefore peak at intermediate stakes and fall to zero for clients with the most at stake.

The expert can also limit what imitating clients gain from information spillovers by distorting the expertise she offers clients with little at stake. Because a project is more likely to be good than bad, learning that one project is good says more about another project than learning that it is bad. A false recommendation to reject a good project is thus more likely to be wrong for a client considering another project than a failure to recommend rejecting a bad one. The expert therefore sometimes recommends that these clients reject a good project, but never that they undertake a bad one. Among distortions that reduce these clients' value by the same amount, this one makes their expertise least valuable to clients who might imitate them.

We have assumed that the client pays before he acts and that fees cannot depend on his decision or its outcome. In some advisory markets, however, fees depend on whether the client's decision turns out to be correct or on the action he takes. A fee that rewards correct decisions would make distortion costly for the expert, since distorted expertise sometimes leads the client to reject a good project. A fee contingent on the client's action would also change his incentive to follow her recommendations. In both cases, the expert would have additional instruments to screen her clients, so the optimal distortion of expertise and the distribution of surplus may differ from those in our setting.

Lastly, the Markov model of information spillovers may be useful beyond the pricing of expertise. It could be applied to problems of experimentation and search, social learning, and strategic communication in which decision makers face uncertainty about correlated alternatives.

\bibliography{references}
\bibliographystyle{ecta}
\clearpage
\appendix
\begin{center}
    \Large\bfseries Appendix
\end{center}

\addcontentsline{toc}{section}{Appendix}

\section{Proofs for the Model}
\label[appendix]{app:model}

\subsection{Proof of \texorpdfstring{\cref{prop:simple-menus}}{Proposition \ref*{prop:simple-menus}}}
\begin{proof}
    Fix an IC and IR direct menu $\br{\pi,t}$ and a type $\theta$. For a signal $s$ with probability $\mu\pi_G^s\br{\theta}+\br{1-\mu}\pi_B^s\br{\theta}>0$, type $\theta$'s posterior probabilities of states $G$ and $B$ are
    \[
        \nu_s\br{\theta}=\frac{\mu\pi_G^s\br{\theta}}{\mu\pi_G^s\br{\theta}+\br{1-\mu}\pi_B^s\br{\theta}}
        \quad\text{and}\quad
        1-\nu_s\br{\theta}=\frac{\br{1-\mu}\pi_B^s\br{\theta}}{\mu\pi_G^s\br{\theta}+\br{1-\mu}\pi_B^s\br{\theta}}.
    \]
    Since the denominators coincide, action $G$ is optimal after $s$, that is, $\nu_s\br{\theta}\geq1/2$, if and only if
    \begin{equation}
        \label{eqn:joint-comparison}
        \mu\pi_G^s\br{\theta}\geq\br{1-\mu}\pi_B^s\br{\theta}.
    \end{equation}
    Let $S_G\br{\theta}$ be the set of signals satisfying \eqref{eqn:joint-comparison} and $S_B\br{\theta}=S\setminus S_G\br{\theta}$. We assign tied and zero-probability signals to $S_G\br{\theta}$, and define the binary information product by
    \[
        \widehat{\pi}_\omega^a\br{\theta}=\sum_{s\in S_a\br{\theta}}\pi_\omega^s\br{\theta}
        \quad \text{for all } a,\omega\in\cbr{G,B}.
    \]

    \paragraph{Responsiveness.} Summing \eqref{eqn:joint-comparison} over $S_G\br{\theta}$ and its reverse inequalities over $S_B\br{\theta}$ gives
    \[
        \mu\widehat{\pi}_G^G\br{\theta}\geq\br{1-\mu}\widehat{\pi}_B^G\br{\theta}
        \quad\text{and}\quad
        \mu\widehat{\pi}_G^B\br{\theta} \leq \br{1-\mu}\widehat{\pi}_B^B\br{\theta},
    \]
    so following each recommendation is optimal for type $\theta$, and $\widehat{\pi}\br{\theta}$ is responsive.

    \paragraph{Information Value.} The information product $\widehat{\pi}\br{\theta}$ is a garbling of $\pi\br{\theta}$, obtained by mapping each signal $s\in S_a\br{\theta}$ to recommendation $a$. Since a garbling of an experiment about $\omega_\theta$ is also a garbling of the information it conveys about $\omega_{\widetilde{\theta}}$, by Blackwell's theorem, $V\br{\widetilde{\theta},\widehat{\pi}\br{\theta}}\leq V\br{\widetilde{\theta},\pi\br{\theta}}$ for every type $\widetilde{\theta}$. Moreover, type $\theta$ takes the same action after recommendation $a$ as after any signal in $S_a\br{\theta}$, so $V\br{\theta,\widehat{\pi}\br{\theta}}=V\br{\theta,\pi\br{\theta}}$.

    \paragraph{IC and IR.} Let $\br{\widehat{\pi},t}$ be the menu obtained by applying the construction to every type at unchanged transfers. Since each type's value from his own information product is unchanged and his value from every other information product is weakly lower, $\br{\widehat{\pi},t}$ is IC and IR.
\end{proof}

\section{Proofs for Two Types}
\label[appendix]{app:two-types}

\subsection{Joint Prior}
\label[appendix]{app:two-types-priors}

By definition,
\[
    \p\br{\omega_{\ell}=G,\omega_h=B}=\p\br{\omega_{\ell}=G}\,\p\br{\omega_h=B\mid\omega_{\ell}=G}=\mu q.
\]
Then equal marginal priors imply
\begin{align*}
    \mu=\p\br{\omega_{\ell}=G}&=\p\br{\omega_{\ell}=G,\omega_h=G}+\p\br{\omega_{\ell}=G,\omega_h=B},\\
    \mu=\p\br{\omega_h=G}&=\p\br{\omega_{\ell}=G,\omega_h=G}+\p\br{\omega_{\ell}=B,\omega_h=G}.
\end{align*}
Subtracting the common term $\p\br{\omega_{\ell}=G,\omega_h=G}$ from both equations gives
\[
    \p\br{\omega_{\ell}=G,\omega_h=B}=\mu-\p\br{\omega_{\ell}=G,\omega_h=G}=\p\br{\omega_{\ell}=B,\omega_h=G}.
\]
The left side equals $\mu q$ by the first display, so $\p\br{\omega_{\ell}=B,\omega_h=G}=\mu q$ as well. The diagonal entries then follow from the marginal probabilities,
\[
    \p\br{\omega_{\ell}=G,\omega_h=G}=\mu-\mu q
    \quad\text{and}\quad
    \p\br{\omega_{\ell}=B,\omega_h=B}=\br{1-\mu}-\mu q,
\]
which gives \eqref{eqn:two-types-joint}.

\subsection{Proof of \texorpdfstring{\cref{lmm:two-types-structure}}{Lemma \ref*{lmm:two-types-structure}}}

\begin{proof}
    Fix a simple menu $\br{\pi,t}$ that satisfies \eqref{eqn:IC} and \eqref{eqn:IR}, and write $x=\pi_G\br{\ell}$ and $y=\pi_B\br{\ell}$ for the low type's recommendation probabilities. Define
    \[
        v:=\mu x+\br{1-\mu}y-\mu
    \]
    as the low type's gain in the probability of a correct decision from $\pi\br{\ell}$, so that his value is $V\br{\ell,\pi\br{\ell}}=v u\br{\ell}$ by \eqref{eqn:valueofinformation_responsive}. By \eqref{eqn:responsive}, $v\geq0$, and since $x,y\leq1$, $v\leq1-\mu$, with equality under full information, $x=y=1$.

    \paragraph{Part (iv).}
    To calculate the high type's value from buying this information product, consider his actions after each recommendation. His posterior probability of the good state $\omega_h=G$ after recommendation $s$ is
    \[
        \begin{aligned}
            \nu_s\br{h\mid\ell}
            &=\nu_s\br{\ell}\,\p\br{\omega_h=G\mid\omega_{\ell}=G}
              +\br{1-\nu_s\br{\ell}}\,\p\br{\omega_h=G\mid\omega_{\ell}=B}\\
            &=\br{1-q}\nu_s\br{\ell}+\frac{\mu q}{1-\mu}\br{1-\nu_s\br{\ell}}\\
            &=\frac{q}{1-\mu}\,\mu
              +\br{1-\frac{q}{1-\mu}}\nu_s\br{\ell},
        \end{aligned}
    \]
    where the first equality is \eqref{eqn:cross-posterior}, and the second uses the conditional probabilities
    \[
        \p\br{\omega_h=G\mid\omega_{\ell}=G}=1-q
        \quad\text{and}\quad
        \p\br{\omega_h=G\mid\omega_{\ell}=B}=\mu q/\br{1-\mu}
    \]
    derived from the joint prior \eqref{eqn:two-types-joint}.

    Since $0\leq q\leq1-\mu$, the posterior $\nu_s\br{h\mid\ell}$ is a weighted average of the prior and the low type's posterior. Both are at least $1/2$ after recommendation $s=G$, since $\mu>1/2$ and $\pi\br{\ell}$ is responsive, so the high type weakly prefers to take action $a=G$. After recommendation $s=B$, he either follows it or takes $G$ regardless. Following the recommendation gives a correct decision only when $\omega_h=B$, so, relative to always taking $G$, his gain in the probability of a correct decision is
    \[
        \begin{aligned}
            &\br{\p\br{s=B,\omega_h=B}-\p\br{s=B,\omega_h=G}}_+\\
            &\quad=\br{\sbr{\mu q\br{1-x}+\br{1-\mu-\mu q}y}-\sbr{\mu\br{1-q}\br{1-x}+\mu qy}}_+\\
            &\quad=\br{\br{1-\mu-2\mu q}y-\mu\br{1-2q}\br{1-x}}_+\\
            &\quad=\br{\br{1-2q}\br{\mu x+\br{1-\mu}y-\mu}-2q\br{2\mu-1}y}_+\\
            &\quad=\br{\br{1-2q}v-2q\br{2\mu-1}y}_+,
        \end{aligned}
    \]
    where the second line uses the joint prior \eqref{eqn:two-types-joint} and the fact that recommendation $B$ is sent with probability $1-x$ in state $\omega_{\ell}=G$ and $y$ in state $\omega_{\ell}=B$.

    At fixed $v$, this gain is minimized by setting $y=1$ and $x=\br{v+2\mu-1}/\mu$, which is feasible for every $v\in\sbr{0,1-\mu}$. Replacing $\pi\br{\ell}$ in $\br{\pi,t}$ by the information product with these recommendation probabilities preserves the low type's value and weakly reduces the high type's value from imitation. Since transfers are unchanged, the new menu satisfies \eqref{eqn:IC} and \eqref{eqn:IR} and yields the same revenue, so an optimal menu can be taken to satisfy part (iv).

    Let
    \[
        \Gamma\br{v}:=\br{\br{1-2q}v-2q\br{2\mu-1}}_+
    \]
    denote the high type's gain in the probability of a correct decision from buying the low type's information product with $y=1$ and $x=\br{v+2\mu-1}/\mu$, where the low type's gain is given by $v$.

    \paragraph{The Expert's Problem.}
    The expert now chooses the low type's gain $v$, the high type's information product $\pi\br{h}=\br{\pi_G\br{h},\pi_B\br{h}}\in\sbr{0,1}^2$, and the two transfers. By \eqref{eqn:valueofinformation_responsive}, the high type's value from his own information product is
    \[
        V\br{h,\pi\br{h}}=\br{\mu\pi_G\br{h}+\br{1-\mu}\pi_B\br{h}-\mu}u\br{h}.
    \]
    Her problem is
    \[
        \begin{alignedat}{3}
            &\max_{v,\,\pi(h),\,t(\ell),\,t(h)}
                && p\,t\br{\ell}+\br{1-p}t\br{h}\\[3pt]
            &\text{subject to}
                && \br{\mu\pi_G\br{h}+\br{1-\mu}\pi_B\br{h}-\mu}u\br{h}-t\br{h}
                    \geq u\br{h}\Gamma\br{v}-t\br{\ell}
                    &\qquad& (\mathrm{IC}_{h\ell}),\\
            &&& u\br{\ell}v-t\br{\ell}
                    \geq V\br{\ell,\pi\br{h}}-t\br{h}
                    &\qquad& (\mathrm{IC}_{\ell h}),\\
            &&& \br{\mu\pi_G\br{h}+\br{1-\mu}\pi_B\br{h}-\mu}u\br{h}-t\br{h}\geq0
                    &\qquad& (\mathrm{IR}_h),\\
            &&& u\br{\ell}v-t\br{\ell}\geq0
                    &\qquad& (\mathrm{IR}_{\ell}),\\
            &&& \mu\pi_G\br{h}+\br{1-\mu}\pi_B\br{h}\geq\mu
                    &\qquad& (\mathrm{RSP}_h),\\
            &&& 0\leq v\leq1-\mu
                    &\qquad& (\mathrm{RSP}_{\ell}).
        \end{alignedat}
    \]
    With $y=1$ and $x=\br{v+2\mu-1}/\mu$, the low type's responsiveness constraint \eqref{eqn:responsive} is $v\geq0$, and $x\leq1$ is $v\leq1-\mu$, which gives $(\mathrm{RSP}_{\ell})$.

    \paragraph{Parts (i) and (ii).}
    Fix $v$ and relax the problem by temporarily dropping $(\mathrm{IC}_{\ell h})$. In the relaxed problem, increasing $t\br{\ell}$ raises revenue and relaxes $(\mathrm{IC}_{h\ell})$, so $(\mathrm{IR}_{\ell})$ binds and $t\br{\ell}=u\br{\ell}v$. The two remaining constraints on the high type's transfer are
    \[
        \begin{alignedat}{2}
            t\br{h}&\leq\br{\mu\pi_G\br{h}+\br{1-\mu}\pi_B\br{h}-\mu}u\br{h}
                &\qquad& (\mathrm{IR}_h),\\
            t\br{h}&\leq\br{\mu\pi_G\br{h}+\br{1-\mu}\pi_B\br{h}-\mu}u\br{h}-u\br{h}\Gamma\br{v}+u\br{\ell}v
                && (\mathrm{IC}_{h\ell}).
        \end{alignedat}
    \]
    Both upper bounds increase in $\pi_G\br{h}$ and $\pi_B\br{h}$, and full information satisfies $(\mathrm{RSP}_h)$, so the expert sets $\pi_G\br{h}=\pi_B\br{h}=1$, giving the high type $\overline{\pi}\br{h}$. She then charges the smaller of the two bounds,
    \[
        t\br{h}=\min\cbr{\overline V\br{h},\;\overline V\br{h}-u\br{h}\Gamma\br{v}+u\br{\ell}v}
        =\overline V\br{h}-\br{u\br{h}\Gamma\br{v}-u\br{\ell}v}_+,
    \]
    where $\overline V\br{h}=\br{1-\mu}u\br{h}$.

   To verify $(\mathrm{IC}_{\ell h})$, observe that $\Gamma\br{v}\leq v$ since $q,v\geq0$. The high type's surplus $\br{u\br{h}\Gamma\br{v}-u\br{\ell}v}_+$ is therefore at most $\br{u\br{h}-u\br{\ell}}v$, which gives
    \[
        t\br{h}\geq\br{1-\mu}u\br{h}-\br{u\br{h}-u\br{\ell}}v
        \geq\br{1-\mu}u\br{\ell},
    \]
    where the second inequality uses $v\leq1-\mu$ and $u\br{h}>u\br{\ell}$. Since $\br{1-\mu}u\br{\ell}$ is the low type's maximum willingness to pay, buying the high type's information product gives him at most zero surplus, while $(\mathrm{IR}_{\ell})$ ensures him zero surplus from his own information product $\pi\br{\ell}$. Thus $(\mathrm{IC}_{\ell h})$ holds for every feasible $v$.

    For fixed $v$ and $y=1$, the relaxed solution is feasible in the original problem and is the unique optimum conditional on these choices. Indeed, both types have positive probability, so the expert charges the largest feasible transfers, and full information uniquely maximizes both upper bounds on $t\br{h}$. This proves parts (i) and (ii).

    \paragraph{Part (iii).}
    Suppose that $(\mathrm{IC}_{h\ell})$ is slack at an optimum with $y=1$. If $(\mathrm{IR}_h)$ is also slack, the expert can raise $t\br{h}$ and increase revenue, a contradiction. Hence $(\mathrm{IR}_h)$ binds and $t\br{h}=\overline V\br{h}$.

    We next claim that $v<1-\mu$. To see this, suppose instead that $v=1-\mu$, so that the low type receives full information at the transfer $t\br{\ell}=\br{1-\mu}u\br{\ell}$. Since $\Gamma\br{1-\mu}=\br{1-\mu-2\mu q}_+$, the high type's surplus from buying the low type's information product is
    \[
        u\br{h}\Gamma\br{1-\mu}-t\br{\ell}
        =\max\cbr{0,1-\mu-2\mu q}u\br{h}-\br{1-\mu}u\br{\ell},
    \]
    which is strictly positive since condition (C1) does not hold. Since $(\mathrm{IR}_h)$ binds, surplus from his own information product is $\overline V\br{h}-t\br{h}=0$, so $(\mathrm{IC}_{h\ell})$ is violated, a contradiction. Hence $v<1-\mu$.

    But then the expert can raise $v$ slightly, and this raises $t\br{\ell}=u\br{\ell}v$ and leaves $t\br{h}$ unchanged. By continuity of $\Gamma$, $(\mathrm{IC}_{h\ell})$ remains satisfied for a sufficiently small increase in $v$ because it is slack. Revenue therefore increases, which contradicts optimality. Thus $(\mathrm{IC}_{h\ell})$ binds at every such optimum, proving part (iii).

    \paragraph{Existence.}
    Each $v\in\sbr{0,1-\mu}$ defines a menu that offers $\pi_G\br{\ell}=\br{v+2\mu-1}/\mu$ and $\pi_B\br{\ell}=1$ at the transfer $t\br{\ell}=u\br{\ell}v$, and offers $\overline{\pi}\br{h}$ at the transfer $t\br{h}=\overline V\br{h}-\br{u\br{h}\Gamma\br{v}-u\br{\ell}v}_+$. Substituting these transfers into the expert's revenue $p\,t\br{\ell}+\br{1-p}t\br{h}$ reduces her problem to one in $v$ alone:
    \begin{equation}
        \label{eqn:two-types-positive-revenue}
        \max_{0\leq v\leq1-\mu}\;
        \br{1-p}\overline V\br{h}+p\,u\br{\ell}v
        -\br{1-p}\br{u\br{h}\Gamma\br{v}-u\br{\ell}v}_+.
    \end{equation}
    The objective is continuous in $v$ on a compact interval, so a solution exists, and each solution gives an optimal menu with properties (i)--(iv).

    \paragraph{Uniqueness: $\pi_B\br{\ell}=1$.}
    Suppose $q>0$ as in \cref{thm:optimal-two-types}. Since $\Gamma\br{v}=0$ for all sufficiently small $v\geq0$, the expert can raise $v$ slightly above zero and increase $t\br{\ell}=u\br{\ell}v$ while keeping $t\br{h}=\overline V\br{h}$. Hence $v=0$ cannot be optimal.

    Now take any optimal simple menu. By part (iv), setting $y=1$ at the same $v$ and transfers gives another optimal menu. In this menu, $(\mathrm{IC}_{h\ell})$ binds by part (iii), so $(\mathrm{IR}_h)$ implies
    \[
        u\br{h}\Gamma\br{v}\geq t\br{\ell}=u\br{\ell}v>0.
    \]
    If the original menu had $y<1$, part (iv) gives the high type's value from buying the low type's information product as
    \[
        u\br{h}\br{\br{1-2q}v-2q\br{2\mu-1}y}_+=u\br{h}\Gamma\br{v}+2q\br{2\mu-1}\br{1-y}u\br{h}>u\br{h}\Gamma\br{v}.
    \]
    The equality uses $\Gamma\br{v}>0$, and the strict inequality follows from $q>0$, $\mu>1/2$, and $y<1$. Since the high type's own information product and both transfers are the same in the two menus, he would strictly prefer to imitate in the original menu, contradicting $(\mathrm{IC}_{h\ell})$.

    Thus every optimal simple menu has $\pi_B\br{\ell}=1$ and $\pi_G\br{\ell}=\br{v+2\mu-1}/\mu$. Parts (i)--(iii) then determine the high type's information product and both transfers, so each optimal $v$ determines a unique optimal simple menu.
\end{proof}

\subsection{Proof of \texorpdfstring{\cref{thm:optimal-two-types}}{Theorem \ref*{thm:optimal-two-types}}}
\begin{proof}
    By \cref{lmm:two-types-structure}, it remains to solve \eqref{eqn:two-types-positive-revenue} in $v$. Since $q\leq1-\mu<1/2$, we have $1-2q>0$, so $\Gamma\br{v}$ is nondecreasing in $v$. We first find the value $v_0$ at which the high type's surplus from imitation, $\br{u\br{h}\Gamma\br{v}-u\br{\ell}v}_+$, turns positive. When $\Gamma\br{v}>0$,
    \[
        \begin{aligned}
            u\br{h}\Gamma\br{v}-u\br{\ell}v
            &=\br{\br{1-2q}u\br{h}-u\br{\ell}}v-2q\br{2\mu-1}u\br{h}\\
            &=\br{\br{1-2q}u\br{h}-u\br{\ell}}\br{v-v_0},
        \end{aligned}
    \]
    where
    \[
        v_0:=\frac{2q\br{2\mu-1}u\br{h}}{\br{1-2q}u\br{h}-u\br{\ell}}.
    \]
    Condition (C1) fails exactly when $\br{1-\mu-2\mu q}u\br{h}>\br{1-\mu}u\br{\ell}$. Expanding $1-\mu-2\mu q=\br{1-2q}\br{1-\mu}-2q\br{2\mu-1}$ and rearranging gives
    \[
        \br{\br{1-2q}u\br{h}-u\br{\ell}}\br{1-\mu}>2q\br{2\mu-1}u\br{h}\geq0.
    \]
    Since $1-\mu>0$, the denominator $\br{1-2q}u\br{h}-u\br{\ell}$ of $v_0$ is positive, and dividing by it gives
    \[
        1-\mu>\frac{2q\br{2\mu-1}u\br{h}}{\br{1-2q}u\br{h}-u\br{\ell}}=v_0\geq0.
    \]

    By the proof of \cref{lmm:two-types-structure}, every maximizing $v$ is positive and satisfies $u\br{h}\Gamma\br{v}\geq u\br{\ell}v>0$. Hence $\Gamma\br{v}>0$, and the expression above for the high type's surplus from imitation gives $v\geq v_0$. Moreover, $q>0$ implies $v_0>0$, and the definition of $v_0$ gives $\br{1-2q}v_0-2q\br{2\mu-1}=u\br{\ell}v_0/u\br{h}$, so $\Gamma\br{v_0}=u\br{\ell}v_0/u\br{h}>0$. Hence $\Gamma\br{v}\geq \Gamma\br{v_0}>0$ on $\sbr{v_0,1-\mu}$, and the same expression applies throughout this interval. The positive part in \eqref{eqn:two-types-positive-revenue} can therefore be omitted, and the objective equals
    \[
        \br{1-p}\overline V\br{h}+p\,u\br{\ell}v-\br{1-p}\br{\br{1-2q}u\br{h}-u\br{\ell}}\br{v-v_0}.
    \]
    Dropping the terms that do not depend on $v$ leaves
    \[
        \max_{v_0\leq v\leq1-\mu}\;\sbr{p\,u\br{\ell}-\br{1-p}\br{\br{1-2q}u\br{h}-u\br{\ell}}}v
        =\max_{v_0\leq v\leq1-\mu}\;\sbr{u\br{\ell}-\br{1-p}\br{1-2q}u\br{h}}v.
    \]

    If the coefficient is strictly positive, the optimal $v$ is $1-\mu$, so both types receive full information. Since $\Gamma\br{1-\mu}=1-\mu-2\mu q$,
    \[
        t\br{h}=\overline V\br{h}-\br{1-\mu-2\mu q}u\br{h}+\br{1-\mu}u\br{\ell}
        =2\mu q\,u\br{h}+\br{1-\mu}u\br{\ell}
        <\overline V\br{h},
    \]
    where the inequality is the failure of condition (C1), so the high type receives positive surplus.

    If the coefficient is strictly negative, $v=v_0$. Then $t\br{h}=\overline V\br{h}$, so both types receive zero surplus, and
    \[
        \pi_G\br{\ell}=\frac{v_0+2\mu-1}{\mu}
        =\frac{\br{2\mu-1}\br{u\br{h}-u\br{\ell}}}{\mu\br{\br{1-2q}u\br{h}-u\br{\ell}}}
        \in\lsbr{\frac{2\mu-1}{\mu},1}
    \]
    where the bounds follow from $0\leq v_0<1-\mu$.

    \paragraph{Uniqueness.}  If the coefficient is zero, every $v\in\sbr{v_0,1-\mu}$ is optimal, including $v=1-\mu$, which gives the full-information menu in the statement of the theorem. If the coefficient is nonzero, the maximizing endpoint is unique. By the proof of \cref{lmm:two-types-structure}, this endpoint determines a unique optimal simple menu.
\end{proof}

\section{Proofs for a Continuum of Types}
\label[appendix]{app:continuum}

\subsection{Proof of \texorpdfstring{\cref{prop:markov-prior}}{Proposition \ref*{prop:markov-prior}}}
\label[appendix]{app:transition-probabilities}

\begin{proof}
The proof has three main steps. The Markov and homogeneity properties give the Chapman--Kolmogorov equations, and continuity then gives a unique $2\times2$ generator $Q$ with $P\br{\Delta}=e^{\Delta Q}$. Then equal marginal priors reduce $Q$ to the single rate $\lambda$ and yield \eqref{eqn:transitionfunction}. Finally, we show that the chain is stationary and reversible.

\paragraph{The Generator.} Fix $\Delta_1,\Delta_2\geq0$ with $\Delta_1+\Delta_2\leq1$. Conditional on $\omega_0=i$, the state $\omega_{\Delta_1}$ equals $k$ with probability $P_{ik}\br{\Delta_1}$. Given this state, the Markov property and homogeneity imply that project $\Delta_1+\Delta_2$ is in state $j$ with probability $P_{kj}\br{\Delta_2}$. Summing over $k$ gives the Chapman--Kolmogorov equations
\[
    P_{ij}\br{\Delta_1+\Delta_2}=\sum_{k\in\Omega}P_{ik}\br{\Delta_1}P_{kj}\br{\Delta_2},
\]
or, in matrix form, $P\br{\Delta_1+\Delta_2}=P\br{\Delta_1}P\br{\Delta_2}$.

By \cref{property:continuity},
\[
    \p\br{\omega_\Delta\ne\omega_0}=\mu P_{GB}\br{\Delta}+\br{1-\mu}P_{BG}\br{\Delta}\longrightarrow0
    \quad\text{as }\Delta\downarrow0.
\]
Since $0<\mu<1$, both off-diagonal entries of $P\br{\Delta}$, $P_{GB}\br{\Delta}$ and $P_{BG}\br{\Delta}$, tend to zero, so $P\br{\Delta}\to I$. The Chapman--Kolmogorov equations and continuity at zero give a unique generator $Q=P'\br{0+}$ on the finite state space $\Omega$ \citep[Section~II.2]{chung1967markov}.
Since $P_{GB}\br{0}=P_{BG}\br{0}=0$, the off-diagonal entries of $Q$ are
\[
    \alpha:=Q_{GB}=\lim_{\Delta\downarrow0}\frac{P_{GB}\br{\Delta}}{\Delta}\geq0
    \quad\text{and}\quad
    \beta:=Q_{BG}=\lim_{\Delta\downarrow0}\frac{P_{BG}\br{\Delta}}{\Delta}\geq0,
\]
the transition rates from $G$ to $B$ and from $B$ to $G$ per unit distance. Each row of $P\br{\Delta}$ sums to one for every $\Delta$, so $P_{iG}\br{\Delta}+P_{iB}\br{\Delta}$ is constant in $\Delta$ and its right derivative at zero, $Q_{iG}+Q_{iB}$, is zero for each $i\in\Omega$. Hence $Q_{GG}=-\alpha$ and $Q_{BB}=-\beta$, and
\[
    Q=
    \begin{pmatrix}
        -\alpha & \alpha\\
        \beta & -\beta
    \end{pmatrix}.
\]
For $0\leq\Delta<1$, differentiating $P\br{\Delta+h}=P\br{\Delta}P\br{h}$ at $h=0+$ gives the forward equation $P'\br{\Delta}=P\br{\Delta}Q$ \citep[Section~1.4]{kelly1979reversibility}. With $P\br{0}=I$, its unique solution is $P\br{\Delta}=e^{\Delta Q}$. The identity also holds at $\Delta=1$, since the Chapman--Kolmogorov equations give $P\br{1}=P\br{1/2}^2=e^{Q}$.

\paragraph{The Transition Matrix.} Since every project is good with probability $\mu$, conditioning on the state of project $0$ gives $\mu P_{GG}\br{\Delta}+\br{1-\mu}P_{BG}\br{\Delta}=\mu$ for every $\Delta$. Differentiating at $\Delta=0+$ gives $-\mu\alpha+\br{1-\mu}\beta=0$, or $\mu\alpha=\br{1-\mu}\beta$. Setting $\lambda=\alpha+\beta$ therefore gives $\alpha=\br{1-\mu}\lambda$ and $\beta=\mu\lambda$. Thus
\begin{equation}
    \label{eqn:generator}
    Q=\lambda
    \begin{pmatrix}
        -\br{1-\mu} & 1-\mu\\
        \mu & -\mu
    \end{pmatrix}.
\end{equation}

For either initial state $i\in\Omega$, the forward equation gives
\[
    P'_{iG}\br{\Delta}=-\alpha P_{iG}\br{\Delta}+\beta P_{iB}\br{\Delta}=\lambda\br{\mu-P_{iG}\br{\Delta}}.
\]
Solving with $P_{iG}\br{0}=\one\cbr{i=G}$ gives
\[
    P_{iG}\br{\Delta}=\mu+\br{\one\cbr{i=G}-\mu}e^{-\lambda\Delta}.
\]
Together with $P_{iB}\br{\Delta}=1-P_{iG}\br{\Delta}$, this gives the transition matrix in \eqref{eqn:transitionfunction}. The rate $\lambda=\alpha+\beta$ is unique because $Q$ is unique.

\paragraph{Stationarity.} Fix projects $\theta_1<\dots<\theta_n$ in $\Theta$ and states $i_1,\dots,i_n\in\Omega$. By the Markov property and homogeneity, for each $1\leq k<n$,
\[
    \p\br{\omega_{\theta_1}=i_1,\dots,\omega_{\theta_{k+1}}=i_{k+1}}
    =\p\br{\omega_{\theta_1}=i_1,\dots,\omega_{\theta_k}=i_k}\,P_{i_k i_{k+1}}\br{\theta_{k+1}-\theta_k}.
\]
Iterating and using $\p\br{\omega_{\theta_1}=i_1}=\p\br{\omega_0=i_1}$ gives
\[
    \p\br{\omega_{\theta_1}=i_1,\dots,\omega_{\theta_n}=i_n}
    =\p\br{\omega_0=i_1}\prod_{k=1}^{n-1}P_{i_k i_{k+1}}\br{\theta_{k+1}-\theta_k}.
\]
For fixed states $i_1,\dots,i_n$, this probability depends only on the distances between consecutive projects, so the chain is stationary.

\paragraph{Reversibility.} Fix projects $\theta$ and $\theta+\Delta$ in $\Theta$. By \eqref{eqn:transitionfunction},
\[
    \mu P_{GB}\br{\Delta}=\br{1-\mu}P_{BG}\br{\Delta}=\mu\br{1-\mu}\br{1-e^{-\lambda\Delta}},
\]
so the joint probability $\p\br{\omega_\theta=i,\omega_{\theta+\Delta}=j}$ is symmetric in $i$ and $j$. Hence
\[
    \p\br{\omega_\theta=j\mid\omega_{\theta+\Delta}=i}
    =\frac{\p\br{\omega_\theta=j,\omega_{\theta+\Delta}=i}}{\p\br{\omega_{\theta+\Delta}=i}}
    =\frac{\p\br{\omega_\theta=i,\omega_{\theta+\Delta}=j}}{\p\br{\omega_\theta=i}}
    =P_{ij}\br{\Delta},
\]
where the second equality uses the symmetry of the joint probability and equal marginal priors, and the third is the definition of $P_{ij}\br{\Delta}$. Thus $P\br{\Delta}$ describes conditional state probabilities in either direction between projects, so $\p\br{\omega_{\widetilde{\theta}}=j\mid\omega_\theta=i}=P_{ij}\br{|\widetilde{\theta}-\theta|}$ for any two projects.
\end{proof}

\subsection{Correlation between Project States}
\label[appendix]{app:state-correlation}

Let $X_\theta:=\one\cbr{\omega_\theta=G}$. Since $\E\sbr{X_\theta}=\mu$ and $X_\theta^2=X_\theta$,
\[
    \operatorname{Var}\br{X_\theta}=\E\sbr{X_\theta^2}-\E\sbr{X_\theta}^2=\mu-\mu^2=\mu\br{1-\mu}.
\]
The same calculation gives $\operatorname{Var}\br{X_{\theta+\Delta}}=\mu\br{1-\mu}$.

The product $X_\theta X_{\theta+\Delta}$ indicates that both projects are good. Hence
\[
    \begin{aligned}
        \E\sbr{X_\theta X_{\theta+\Delta}}
        &=\p\br{\omega_\theta=G,\omega_{\theta+\Delta}=G}\\
        &=\p\br{\omega_\theta=G}\,\p\br{\omega_{\theta+\Delta}=G\mid\omega_\theta=G}\\
        &=\mu P_{GG}\br{\Delta}.
    \end{aligned}
\]
Subtracting the product of the means gives
\[
    \operatorname{Cov}\br{X_\theta,X_{\theta+\Delta}}
    =\E\sbr{X_\theta X_{\theta+\Delta}}-\E\sbr{X_\theta}\E\sbr{X_{\theta+\Delta}}
    =\mu P_{GG}\br{\Delta}-\mu^2.
\]
Then using $P_{GG}\br{\Delta}$ from \eqref{eqn:transitionfunction} gives
\[
    \operatorname{Corr}\br{X_\theta,X_{\theta+\Delta}}
    =\frac{\operatorname{Cov}\br{X_\theta,X_{\theta+\Delta}}}{\sqrt{\operatorname{Var}\br{X_\theta}}\sqrt{\operatorname{Var}\br{X_{\theta+\Delta}}}}
    =\frac{\mu P_{GG}\br{\Delta}-\mu^2}{\mu\br{1-\mu}}
    =\frac{P_{GG}\br{\Delta}-\mu}{1-\mu}
    =e^{-\lambda\Delta}.
\]

\subsection{Proof of \texorpdfstring{\cref{lmm:cross-value}}{Lemma \ref*{lmm:cross-value}}}

\begin{proof}
    Fix an information product $\pi\br{\widetilde{\theta}}$ that satisfies the responsiveness constraint \eqref{eqn:responsive} for type $\widetilde{\theta}$. Then fix a client of type $\theta$, and let $\Delta=|\theta-\widetilde{\theta}|$. By \eqref{dfn:valueforinformation}, his willingness to pay for $\pi\br{\widetilde{\theta}}$ is
    \[
        V\br{\theta,\pi\br{\widetilde{\theta}}}
        =u\br{\theta}\br{\sum_{s\in\cbr{G,B}}
            \br{\mu\pi_G^s\br{\widetilde{\theta}}+\br{1-\mu}\pi_B^s\br{\widetilde{\theta}}}
            \max\cbr{\nu_s\br{\theta\mid\widetilde{\theta}},1-\nu_s\br{\theta\mid\widetilde{\theta}}}
            -\mu}.
    \]
    We next derive the posteriors explicitly and use them to determine the client's optimal action after each recommendation.

    \paragraph{The Client's Posterior.}
    Let $\pi_i^s\br{\theta\mid\widetilde{\theta}}$ denote the probability of recommendation $s$ conditional on $\omega_\theta=i$. By customization, conditional on $\omega_{\widetilde{\theta}}$, the recommendation is independent of $\omega_\theta$. Conditioning on the two possible states of project $\widetilde{\theta}$ and applying \cref{prop:markov-prior} therefore gives
    \begin{equation}
        \label{eqn:induced-experiment}
        \pi_i^s\br{\theta\mid\widetilde{\theta}}
        =P_{iG}\br{\Delta}\pi_G^s\br{\widetilde{\theta}}
         +P_{iB}\br{\Delta}\pi_B^s\br{\widetilde{\theta}},
        \qquad i,s\in\cbr{G,B}.
    \end{equation}

    Setting $i=G$ and substituting the transition probabilities from \eqref{eqn:transitionfunction} gives
    \[
        \begin{aligned}
            \pi_G^s\br{\theta\mid\widetilde{\theta}}
            &=\br{\mu+\br{1-\mu}e^{-\lambda\Delta}}
              \pi_G^s\br{\widetilde{\theta}}
              +\br{1-\mu}\br{1-e^{-\lambda\Delta}}
               \pi_B^s\br{\widetilde{\theta}}\\
            &=e^{-\lambda\Delta}\pi_G^s\br{\widetilde{\theta}}
              +\br{1-e^{-\lambda\Delta}}
               \br{\mu\pi_G^s\br{\widetilde{\theta}}
                   +\br{1-\mu}\pi_B^s\br{\widetilde{\theta}}}.
        \end{aligned}
    \]
    For each $s$ sent with positive probability, Bayes' rule then gives
    \[
        \begin{aligned}
            \nu_s\br{\theta\mid\widetilde{\theta}}
            &=\frac{\mu\pi_G^s\br{\theta\mid\widetilde{\theta}}}
                   {\mu\pi_G^s\br{\widetilde{\theta}}
                    +\br{1-\mu}\pi_B^s\br{\widetilde{\theta}}}\\
            &=e^{-\lambda\Delta}
              \frac{\mu\pi_G^s\br{\widetilde{\theta}}}
                   {\mu\pi_G^s\br{\widetilde{\theta}}
                    +\br{1-\mu}\pi_B^s\br{\widetilde{\theta}}}
              +\br{1-e^{-\lambda\Delta}}\mu\\
            &=e^{-\lambda\Delta}\nu_s\br{\widetilde{\theta}}
              +\br{1-e^{-\lambda\Delta}}\mu.
        \end{aligned}
    \]
    The last equality applies Bayes' rule to project $\widetilde{\theta}$. This proves \eqref{eqn:posterior}.

\paragraph{Optimal Actions.}
Whenever $G$ is sent, responsiveness gives $\nu_G\br{\widetilde{\theta}}\geq1/2$. Since $\mu\geq1/2$ as well, the posterior \eqref{eqn:posterior} satisfies $\nu_G\br{\theta\mid\widetilde{\theta}}\geq1/2$, so type $\theta$ undertakes his project after $G$. Equivalently, $\mu\pi_G^G\br{\theta\mid\widetilde{\theta}}\geq\br{1-\mu}\pi_B^G\br{\theta\mid\widetilde{\theta}}$, which also holds trivially if $G$ is never sent.

We now evaluate $V\br{\theta,\pi\br{\widetilde{\theta}}}$. For each recommendation $s$ sent with positive probability, Bayes' rule gives
\[
    \begin{aligned}
        \br{\mu\pi_G^s\br{\widetilde{\theta}}
            +\br{1-\mu}\pi_B^s\br{\widetilde{\theta}}}
        \nu_s\br{\theta\mid\widetilde{\theta}}
        &=\mu\pi_G^s\br{\theta\mid\widetilde{\theta}},\\
        \br{\mu\pi_G^s\br{\widetilde{\theta}}
            +\br{1-\mu}\pi_B^s\br{\widetilde{\theta}}}
        \br{1-\nu_s\br{\theta\mid\widetilde{\theta}}}
        &=\br{1-\mu}\pi_B^s\br{\theta\mid\widetilde{\theta}}.
    \end{aligned}
\]
Summing over recommendations and subtracting the outside option gives
\[
    \begin{aligned}
        V\br{\theta,\pi\br{\widetilde{\theta}}}
        &=u\br{\theta}\left(
            \sum_{s\in\cbr{G,B}}
            \max\cbr{\mu\pi_G^s\br{\theta\mid\widetilde{\theta}},
                     \br{1-\mu}\pi_B^s\br{\theta\mid\widetilde{\theta}}}
            -\mu\right)\\
        &=u\br{\theta}\Bigl(
            \mu\pi_G^G\br{\theta\mid\widetilde{\theta}}
            +\max\cbr{\mu\pi_G^B\br{\theta\mid\widetilde{\theta}},
                     \br{1-\mu}\pi_B^B\br{\theta\mid\widetilde{\theta}}}
            -\mu\Bigr),
    \end{aligned}
\]
where the second equality uses that undertaking is optimal after recommendation $G$.

Next, note that
\[
    \mu\pi_G^G\br{\theta\mid\widetilde{\theta}}-\mu
    =-\mu\pi_G^B\br{\theta\mid\widetilde{\theta}}.
\]
Substituting into the preceding display and using $\max\cbr{a,b}-a=\br{b-a}_+$ gives
\[
    \begin{aligned}
        V\br{\theta,\pi\br{\widetilde{\theta}}}
        &=u\br{\theta}\Bigl(
            \max\cbr{\mu\pi_G^B\br{\theta\mid\widetilde{\theta}},
                     \br{1-\mu}\pi_B^B\br{\theta\mid\widetilde{\theta}}}
            -\mu\pi_G^B\br{\theta\mid\widetilde{\theta}}\Bigr)\\
        &=u\br{\theta}\br{\br{1-\mu}\pi_B^B\br{\theta\mid\widetilde{\theta}}
                          -\mu\pi_G^B\br{\theta\mid\widetilde{\theta}}}_+\\
        &=u\br{\theta}\br{\p\br{s=B,\omega_\theta=B}
                          -\p\br{s=B,\omega_\theta=G}}_+,
    \end{aligned}
\]
which also verifies \eqref{eqn:prob-gain}, the gain in the probability of a correct decision from following the recommendation.

\paragraph{Value of the Information Product.}
Substituting \eqref{eqn:induced-experiment} with $s=B$ into the expression for $V\br{\theta,\pi\br{\widetilde{\theta}}}$ gives
\[
    \begin{aligned}
        V\br{\theta,\pi\br{\widetilde{\theta}}}
        &=u\br{\theta}\br{\br{1-\mu}\pi_B^B\br{\theta\mid\widetilde{\theta}}
                          -\mu\pi_G^B\br{\theta\mid\widetilde{\theta}}}_+\\
        &=u\br{\theta}\Bigl(
            \br{1-\mu}\br{P_{BG}\br{\Delta}\br{1-\pi_G\br{\widetilde{\theta}}}
                         +P_{BB}\br{\Delta}\pi_B\br{\widetilde{\theta}}}\\
        &\qquad
            -\mu\br{P_{GG}\br{\Delta}\br{1-\pi_G\br{\widetilde{\theta}}}
                   +P_{GB}\br{\Delta}\pi_B\br{\widetilde{\theta}}}\Bigr)_+\\
        &=u\br{\theta}\Bigl(
            \mu\pi_G\br{\widetilde{\theta}}+\br{1-\mu}\pi_B\br{\widetilde{\theta}}-\mu\\
        &\qquad
            -\br{\br{1-\mu}P_{BG}\br{\Delta}+\mu P_{GB}\br{\Delta}}
             \br{\pi_G\br{\widetilde{\theta}}+\pi_B\br{\widetilde{\theta}}-1}\Bigr)_+,
    \end{aligned}
\]
where the last equality uses $P_{BB}\br{\Delta}=1-P_{BG}\br{\Delta}$ and
$P_{GG}\br{\Delta}=1-P_{GB}\br{\Delta}$. By \eqref{eqn:transitionfunction},
\[
    \br{1-\mu}P_{BG}\br{\Delta}+\mu P_{GB}\br{\Delta}
    =2\mu\br{1-\mu}\br{1-e^{-\lambda\Delta}}.
\]
Substituting this equality gives \eqref{eqn:general-cross-value}.
\end{proof}

\subsection{Proof of \texorpdfstring{\cref{prop:fullsurpluscontinuum}}{Proposition \ref*{prop:fullsurpluscontinuum}}}
\begin{proof}
    Fix the efficient full surplus extraction menu $\br{\overline{\pi},\overline t}$. Under this menu, truth-telling types receive zero surplus, so \eqref{eqn:IR} holds, and \eqref{eqn:IC} holds if and only if no type earns positive surplus from any other report. We show that this holds if and only if \eqref{eqn:fullsurpluscontinuum} holds, and then that every optimal simple menu agrees with $\br{\overline{\pi},\overline t}$ for $dF$-almost every $\theta$.

    \paragraph{Necessity.}
    The argument preceding \cref{prop:fullsurpluscontinuum} shows that if the menu satisfies \eqref{eqn:IC}, then \eqref{eqn:fullsurpluscontinuum} holds at every $\theta>0$. Continuity of $u$ and $u'$ extends the inequality to $\theta=0$.

    \paragraph{Sufficiency.}
    Suppose \eqref{eqn:fullsurpluscontinuum} holds. An upward deviation to $\widetilde{\theta}>\theta$ costs
    \[
        \overline t\br{\widetilde{\theta}}=\br{1-\mu}u\br{\widetilde{\theta}}.
    \]
    Since $u$ is strictly increasing, this exceeds $\br{1-\mu}u\br{\theta}$, so an upward deviation gives negative surplus.

    For downward deviations, fix a project $\widetilde{\theta}$ and consider any type $\theta>\widetilde{\theta}$, with $\Delta=\theta-\widetilde{\theta}$. We show that $\overline{\pi}\br{\widetilde{\theta}}$ is worth at most its price $\overline t\br{\widetilde{\theta}}$ to type $\theta$. Setting $\pi_G=\pi_B=1$ in \eqref{eqn:general-cross-value} gives
    \[
        V\br{\theta,\overline{\pi}\br{\widetilde{\theta}}}
        =\br{1-\mu}u\br{\theta}
         \br{1-2\mu\br{1-e^{-\lambda\Delta}}}_+.
    \]
    At $\theta=\widetilde{\theta}$ this equals $\overline t\br{\widetilde{\theta}}$, so it suffices to show that the value is nonincreasing in $\theta$. Wherever the value is positive, differentiating with respect to $\theta$ gives
    \[
        \frac{1}{1-\mu}\frac{\partial}{\partial\theta}
        V\br{\theta,\overline{\pi}\br{\widetilde{\theta}}}
        =\br{1-2\mu\br{1-e^{-\lambda\Delta}}}u'\br{\theta}
         -2\mu\lambda e^{-\lambda\Delta}u\br{\theta}.
    \]
    The coefficient of $u'\br{\theta}$ is positive wherever the value is positive, so substituting $u'\br{\theta}\leq2\mu\lambda u\br{\theta}$ from \eqref{eqn:fullsurpluscontinuum} and rearranging the terms gives
    \begin{equation}
        \label{eqn:full-information-monotonicity}
        \frac{\partial}{\partial\theta}
        V\br{\theta,\overline{\pi}\br{\widetilde{\theta}}}
        \leq 2\mu\br{1-\mu}\lambda\br{2\mu-1}
             \br{e^{-\lambda\Delta}-1}u\br{\theta}
        \leq 0,
    \end{equation}
    where the last inequality uses $\mu\geq1/2$ and $\Delta\geq0$. Once the value reaches zero, it stays zero, since the expression inside the positive part is nonincreasing in $\Delta$. Hence
    \[
        V\br{\theta,\overline{\pi}\br{\widetilde{\theta}}}
        \leq V\br{\widetilde{\theta},\overline{\pi}\br{\widetilde{\theta}}}
        =\overline t\br{\widetilde{\theta}},
    \]
    so a downward deviation gives type $\theta$ at most zero surplus. The menu therefore satisfies \eqref{eqn:IC}.

    \paragraph{Uniqueness.}
    By \eqref{eqn:IR}, every feasible simple menu satisfies
    \[
        t\br{\theta}\leq V\br{\theta,\pi\br{\theta}}
        \leq\overline V\br{\theta}=\br{1-\mu}u\br{\theta}
        \quad\text{for every }\theta\in\Theta.
    \]
    The menu $\br{\overline{\pi},\overline t}$ attains this upper bound for every type, so it is optimal. Any other optimal simple menu must attain the bound for $dF$-almost every $\theta$, since charging less on a set of positive $F$-measure would give strictly lower expected revenue. Thus $t\br{\theta}=V\br{\theta,\pi\br{\theta}}=\overline V\br{\theta}$ for $dF$-almost every $\theta$. By \eqref{eqn:valueofinformation_responsive}, $V\br{\theta,\pi\br{\theta}}=\overline V\br{\theta}$ requires
    \[
        \mu\pi_G\br{\theta}+\br{1-\mu}\pi_B\br{\theta}=1.
    \]
    Since $0<\mu<1$ and $\pi_G\br{\theta},\pi_B\br{\theta}\leq1$, this holds only if $\pi_G\br{\theta}=\pi_B\br{\theta}=1$. Hence every optimal simple menu agrees with $\br{\overline{\pi},\overline t}$ for $dF$-almost every $\theta$.
\end{proof}

\subsection{Proof of \texorpdfstring{\cref{lmm:one-parameter}}{Lemma \ref*{lmm:one-parameter}}}
\begin{proof}
    Fix a simple menu $\br{\pi,t}$ that satisfies \eqref{eqn:IC} and \eqref{eqn:IR}. For each $\widetilde{\theta}$, define $\widehat{\pi}$ by
    \[
        \widehat{\pi}_G\br{\widetilde{\theta}}=\pi_G\br{\widetilde{\theta}}-\frac{1-\mu}{\mu}\br{1-\pi_B\br{\widetilde{\theta}}}
        \quad\text{and}\quad \widehat{\pi}_B\br{\widetilde{\theta}}=1.
    \]
    By \eqref{eqn:responsive}, $\br{2\mu-1}/\mu\leq\widehat{\pi}_G\br{\widetilde{\theta}}\leq1$, so $\widehat{\pi}$ is feasible and responsive. It also preserves type $\widetilde{\theta}$'s value because
    \[
        \mu\widehat{\pi}_G\br{\widetilde{\theta}}+1-\mu=\mu\pi_G\br{\widetilde{\theta}}+\br{1-\mu}\pi_B\br{\widetilde{\theta}}.
    \]

    It remains to show that every other type values $\widehat{\pi}\br{\widetilde{\theta}}$ weakly less than $\pi\br{\widetilde{\theta}}$. To compare the loss of information value, using \eqref{eqn:general-cross-value} and $\widehat{\pi}_B\br{\widetilde{\theta}}=1$, we have
    \[
       \br{\widehat{\pi}_G\br{\widetilde{\theta}}
                 +\widehat{\pi}_B\br{\widetilde{\theta}}-1}
            -\br{\pi_G\br{\widetilde{\theta}}+\pi_B\br{\widetilde{\theta}}-1}
            =\frac{2\mu-1}{\mu}\br{1-\pi_B\br{\widetilde{\theta}}}\geq0.
    \]
    Thus $\widehat{\pi}\br{\widetilde{\theta}}$ gives type $\widetilde{\theta}$ the same gain in the probability of a correct decision about his own project, while the loss in information value is weakly larger for every other type. By \eqref{eqn:general-cross-value},
    \[
        V\br{\theta,\widehat{\pi}\br{\widetilde{\theta}}}
        \leq V\br{\theta,\pi\br{\widetilde{\theta}}}
        \quad\text{for every }\theta\neq\widetilde{\theta}.
    \]
    Hence $\br{\widehat{\pi},t}$ satisfies \eqref{eqn:IC} and \eqref{eqn:IR}.
\end{proof}

\subsection{Proof of \texorpdfstring{\cref{thm:characterizationcontinuumtypes}}{Theorem \ref*{thm:characterizationcontinuumtypes}}}

The proof is separated into five parts. We first derive a relaxation of the expert's problem from a necessary local incentive constraint. We then determine the cutoffs and construct a candidate menu. The next two parts show that it minimizes revenue loss in the relaxed problem and satisfies the incentive constraints for all reports. Finally, we prove optimality and uniqueness.

Throughout this section we fix $\mu\in\br{1/2,1}$ and $\lambda>0$, and maintain the assumptions of \cref{thm:characterizationcontinuumtypes}.

\subsubsection{The Relaxed Problem}
We first show that \eqref{eqn:IC} bounds the slope of surplus. For each interior type $\theta$ and fixed $\pi\br{\theta}$, define the right- and left-hand derivatives
\[
    V_\theta^\pm\br{\theta,\pi\br{\theta}}
    :=\lim_{h\to0^\pm}
    \frac{V\br{\theta+h,\pi\br{\theta}}
          -V\br{\theta,\pi\br{\theta}}}{h}.
\]
Let $c\br{\theta}:=2\mu\br{1-\mu}\lambda u\br{\theta}$ denote the marginal cost of imitation.

\begin{lemma}[Local Incentive Constraints]
    \label{lmm:local-ic}
    For every simple menu $\br{\pi,t}$ that satisfies \eqref{eqn:IC}, the one-sided derivatives $V_\theta^\pm\br{\theta,\pi\br{\theta}}$ exist at every interior type $\theta$. If $V\br{\theta,\pi\br{\theta}}=0$, both derivatives are zero. If $V\br{\theta,\pi\br{\theta}}>0$, they satisfy
    \begin{align*}
        V_\theta^+\br{\theta,\pi\br{\theta}}
        &=\br{\mu\pi_G\br{\theta}+\br{1-\mu}\pi_B\br{\theta}-\mu}u'\br{\theta}
          -c\br{\theta}\br{\pi_G\br{\theta}+\pi_B\br{\theta}-1},\\
        V_\theta^-\br{\theta,\pi\br{\theta}}
        &=\br{\mu\pi_G\br{\theta}+\br{1-\mu}\pi_B\br{\theta}-\mu}u'\br{\theta}
          +c\br{\theta}\br{\pi_G\br{\theta}+\pi_B\br{\theta}-1}.
    \end{align*}
    The surplus $U\br{\theta}=V\br{\theta,\pi\br{\theta}}-t\br{\theta}$ is absolutely continuous, and
    \begin{equation}
        \label{eqn:envelope-bounds}
        V_\theta^+\br{\theta,\pi\br{\theta}}
        \leq U'\br{\theta}
        \leq V_\theta^-\br{\theta,\pi\br{\theta}}
        \quad dF\text{-almost everywhere on }\Theta.
    \end{equation}
\end{lemma}

The lower bound in \eqref{eqn:envelope-bounds} is the local downward incentive constraint and the upper bound is the local upward incentive constraint. When $V\br{\cdot,\pi\br{\theta}}$ is differentiable at $\theta$, the two bounds coincide and give the envelope formula \citep{milgrom2002envelope}. When $V\br{\theta,\pi\br{\theta}}>0$, however, $V\br{\cdot,\pi\br{\theta}}$ has a kink at $\theta$. \eqref{eqn:IC} therefore bounds $U'\br{\theta}$ without determining it uniquely, as in \citet{carbajal2013mechanism}.

\begin{proof}
Fix a simple menu $\br{\pi,t}$ that satisfies \eqref{eqn:IC}.

We first compute the one-sided derivatives. Fix an interior type $\theta$. If $V\br{\theta,\pi\br{\theta}}=0$, then \eqref{eqn:responsive} gives
\[
    \mu\pi_G\br{\theta}+\br{1-\mu}\pi_B\br{\theta}-\mu\geq0.
\]
At type $\theta$, the distance is $\Delta=0$, so \eqref{eqn:general-cross-value} gives
\[
    0=V\br{\theta,\pi\br{\theta}}
    =u\br{\theta}\br{\mu\pi_G\br{\theta}
        +\br{1-\mu}\pi_B\br{\theta}-\mu}.
\]
Since $u\br{\theta}>0$, the expression in parentheses is zero, so \eqref{eqn:responsive} binds. Rearranging this equality yields
\[
    \pi_G\br{\theta}+\pi_B\br{\theta}-1
    =\frac{2\mu-1}{\mu}\pi_B\br{\theta}\geq0.
\]
Now suppose type $\theta+h\in\Theta$ buys $\pi\br{\theta}$. The first three terms inside the positive part in \eqref{eqn:general-cross-value} are independent of $h$, so they remain zero. The expression inside the positive part therefore reduces to
\[
    -2\mu\br{1-\mu}\br{1-e^{-\lambda|h|}}
      \br{\pi_G\br{\theta}+\pi_B\br{\theta}-1}\leq0.
\]
Hence $V\br{\theta+h,\pi\br{\theta}}=0$ for every $\theta+h\in\Theta$, and both one-sided derivatives are zero.

If $V\br{\theta,\pi\br{\theta}}>0$, the expression inside the positive part remains positive for all sufficiently small $|h|$. Differentiating \eqref{eqn:general-cross-value} from the right and from the left gives the formulas stated in the lemma. Thus both one-sided derivatives exist at every interior type.

We now show that $U$ is Lipschitz continuous and prove \eqref{eqn:envelope-bounds}. Fix a report $\widetilde{\theta}$. By \eqref{eqn:responsive} and $\mu>1/2$,
$0\leq\pi_G\br{\widetilde{\theta}}+\pi_B\br{\widetilde{\theta}}-1\leq1$.
In \eqref{eqn:general-cross-value}, the positive-part expression multiplying $u\br{\theta}$ lies in $\sbr{0,1-\mu}$. Moreover, the function $1-e^{-\lambda|\theta-\widetilde{\theta}|}$ is $\lambda$-Lipschitz. Thus $\theta\mapsto V\br{\theta,\pi\br{\widetilde{\theta}}}$ is $K$-Lipschitz for every report $\widetilde{\theta}$, where
\[
    K:=\br{1-\mu}\max_{\theta\in\Theta}|u'\br{\theta}|
    +2\mu\br{1-\mu}\lambda\max_{\theta\in\Theta}u\br{\theta}<\infty.
\]
Then by \eqref{eqn:IC} the surplus function $U$ is $K$-Lipschitz, hence absolutely continuous and differentiable $dF$-almost everywhere.

We can now adapt the argument in the proof of Theorem~1 of \citet{milgrom2002envelope} to the one-sided derivatives of $V\br{\cdot,\pi\br{\theta}}$. Fix an interior $\theta$ where $U'$ exists. Type $\theta+h$ can buy $\pi\br{\theta}$ at $t\br{\theta}$, so \eqref{eqn:IC} gives
\[
    U\br{\theta+h}-U\br{\theta}
    \geq V\br{\theta+h,\pi\br{\theta}}-V\br{\theta,\pi\br{\theta}}.
\]
Dividing by $h>0$ and letting $h\downarrow0$ gives $U'\br{\theta}\geq V_\theta^+\br{\theta,\pi\br{\theta}}$. Dividing by $h<0$ reverses the inequality, and letting $h\uparrow0$ gives $U'\br{\theta}\leq V_\theta^-\br{\theta,\pi\br{\theta}}$, hence proving \eqref{eqn:envelope-bounds}.
\end{proof}

We now rewrite the expert's problem in terms of $\pi_G$ and $U$. By \cref{lmm:one-parameter}, we can set $\pi_B=1$ while keeping transfers and surplus unchanged. Letting $\pi\br{\theta}=\br{\pi_G\br{\theta},1}$ and substituting in \eqref{eqn:general-cross-value} gives
\begin{equation}
    \label{eqn:value}
        V\br{\theta,\pi\br{\widetilde{\theta}}}
        =u\br{\theta}\Bigl(\mu\pi_G\br{\widetilde{\theta}}+1-2\mu
        -2\mu\br{1-\mu}\pi_G\br{\widetilde{\theta}}
          \br{1-e^{-\lambda|\theta-\widetilde{\theta}|}}\Bigr)_+.
\end{equation}
Moreover, \eqref{eqn:responsive} reduces to $\pi_G\br{\theta}\geq\br{2\mu-1}/\mu$. Setting $\widetilde{\theta}=\theta$ in \eqref{eqn:value}, we can drop the positive part by responsiveness and write
\[
    V\br{\theta,\pi\br{\theta}}
    =\br{\mu\pi_G\br{\theta}+1-2\mu}u\br{\theta}
    =\overline V\br{\theta}-\mu u\br{\theta}\br{1-\pi_G\br{\theta}}.
\]
Since $t\br{\theta}=V\br{\theta,\pi\br{\theta}}-U\br{\theta}$, the expert's expected revenue is therefore
\[
        \int_0^1t\br{r}f\br{r}\,dr
        =\int_0^1\overline V\br{r}f\br{r}\,dr-J\br{\pi_G,U},
\]
where $J\br{\pi_G,U}$ is the \textit{revenue loss}, given by
\[
        J\br{\pi_G,U} :=\int_0^1
          \br{\mu u\br{r}\br{1-\pi_G\br{r}}+U\br{r}}f\br{r}\,dr.
\]

Omitting the first term, which is independent of $\pi_G$ and $U$, and substituting for transfers in \eqref{eqn:IC} and \eqref{eqn:IR}, we obtain the equivalent problem
\begin{equation}
    \label{eqn:problem}
    \tag{P}
    \begin{aligned}
        \min_{\br{\pi_G,U}}\quad
        &J\br{\pi_G,U}\\
        \text{subject to}\quad
        &U\br{\theta}-U\br{\widetilde{\theta}}
          \geq V\br{\theta,\pi\br{\widetilde{\theta}}}
               -V\br{\widetilde{\theta},\pi\br{\widetilde{\theta}}}
          &&\text{for all }\theta,\widetilde{\theta}\in\Theta,\\
        &U\br{\theta}\geq0
          \quad\text{and}\quad
          \frac{2\mu-1}{\mu}\leq\pi_G\br{\theta}\leq1
          &&\text{for all }\theta\in\Theta.
    \end{aligned}
\end{equation}

We now use the lower bound in \eqref{eqn:envelope-bounds} to derive a relaxation of \eqref{eqn:IC}. Fix a simple menu with $\pi_B=1$ that satisfies \eqref{eqn:IC}. By \eqref{eqn:envelope-bounds}, $U'\br{\theta}\geq V_\theta^+\br{\theta,\pi\br{\theta}}$ for $dF$-almost every $\theta$. Moreover, at every interior type,
\[
    V_\theta^+\br{\theta,\pi\br{\theta}}
    \geq\br{\mu\pi_G\br{\theta}+1-2\mu}u'\br{\theta}
      -c\br{\theta}\pi_G\br{\theta}.
\]
If $V\br{\theta,\pi\br{\theta}}>0$, this holds with equality by \cref{lmm:local-ic} with $\pi_B\br{\theta}=1$. If $V\br{\theta,\pi\br{\theta}}=0$, then $V_\theta^+\br{\theta,\pi\br{\theta}}=0$ and $\pi_G\br{\theta}=\br{2\mu-1}/\mu$, so the first term on the right is zero and the second is negative. Combining the two inequalities gives, for $dF$-almost every $\theta$,
\begin{equation}
    \label{eqn:local}
    \tag{IC$_{-}$}
    U'\br{\theta}\geq\br{\mu\pi_G\br{\theta}+1-2\mu}u'\br{\theta}
      -c\br{\theta}\pi_G\br{\theta},
\end{equation}
where \eqref{eqn:local} is the relaxed local downward incentive constraint.

To interpret \eqref{eqn:local}, we rewrite its right side. Since $\overline V'\br{\theta}=\br{1-\mu}u'\br{\theta}$,
\begin{equation}
    \label{eqn:ld-decomposition}
    \br{\mu\pi_G\br{\theta}+1-2\mu}u'\br{\theta}
      -c\br{\theta}\pi_G\br{\theta}
    =\br{\overline V'\br{\theta}-c\br{\theta}}
      -\br{\mu u'\br{\theta}-c\br{\theta}}\br{1-\pi_G\br{\theta}}.
\end{equation}
The two terms correspond to the two events in which the expert recommends rejecting the project. The first term arises when she recommends rejecting a bad project, which happens with probability $1-\mu$ and generates the value of full information. The value of this recommendation to a marginally higher type increases at rate $\overline V'\br{\theta}$ because his payoff is higher and falls at rate $c\br{\theta}$ because his own project may be good. The second term arises when she distorts expertise by recommending that the client reject a good project, which happens with probability $\mu\br{1-\pi_G\br{\theta}}$. Each unit of distortion costs a client who follows this recommendation $\mu$ times his payoff from a correct decision. For a marginally higher type, this cost increases at rate $\mu u'\br{\theta}$ because his payoff is higher and falls at rate $c\br{\theta}$ because his own project may be bad. The same $c\br{\theta}$ appears in both terms because the Markov prior is reversible. Per unit distance, the probability that project $\theta$ is bad and the state switches to good is $\br{1-\mu}\mu\lambda$, which equals the probability that it is good and the state switches to bad, $\mu\br{1-\mu}\lambda$. Lastly, a switch reverses whether the recommendation is correct for the client, which changes his gain from following it by twice his payoff, so $c\br{\theta}=2\mu\br{1-\mu}\lambda u\br{\theta}$.

Next, replacing the first constraint in \eqref{eqn:problem} with \eqref{eqn:local} gives the relaxed problem
\begin{equation}
    \label{eqn:loss}
    \tag{RP}
    \begin{aligned}
        \min_{\br{\pi_G,U}}\quad
        &J\br{\pi_G,U}\\
        \text{subject to}\quad
        &\eqref{eqn:local},\quad U\br{\theta}\geq0
          \quad\text{and}\quad \frac{2\mu-1}{\mu}\leq\pi_G\br{\theta}\leq1
          \quad\text{for every }\theta\in\Theta,
    \end{aligned}
\end{equation}
where the minimization is over measurable $\pi_G$ and absolutely continuous $U$.

Every pair $\br{\pi_G,U}$ that is feasible in \eqref{eqn:problem} is also feasible in \eqref{eqn:loss}, because $U$ is absolutely continuous by \cref{lmm:local-ic} and \eqref{eqn:IC} implies \eqref{eqn:local}. Hence any solution to \eqref{eqn:loss} that is feasible in \eqref{eqn:problem} also solves \eqref{eqn:problem}.

The rest of the proof finds such a solution. We construct a candidate menu and the corresponding pair $\br{\pi_G^*,U^*}$. We show that this pair solves \eqref{eqn:loss} in \cref{lmm:revenue-bound} and that it is feasible in \eqref{eqn:problem} in \cref{lmm:global-ic}.

\subsubsection{The Candidate Menu}
We first note two facts used in the rest of the proof. First, since $u$ is concave, $u'/u$ is nonincreasing, so condition \eqref{eqn:continuum-growth} implies $u'\br{\theta}\geq\lambda u\br{\theta}$ for all $\theta\in\Theta$. Hence
\begin{equation}
    \label{eqn:inequality}
    \mu u'\br{\theta}-c\br{\theta}\geq\mu\br{2\mu-1}\lambda u\br{\theta}>0 \quad \text{for all }\theta\in\Theta.
\end{equation}

Second, $\overline V'-c$ is strictly decreasing, because $u'$ is nonincreasing and $c$ is strictly increasing, and $\overline V'\br{0}>c\br{0}$, because \eqref{eqn:fullsurpluscontinuum} does not hold. We define the critical type $\ctheta$ as the unique type at which $\overline V'=c$ if one exists, and as $1$ otherwise. Thus $\overline V'>c$ on $\lsbr{0,\ctheta}$ and $\overline V'<c$ on $\rsbr{\ctheta,1}$.

We next show that \cref{asm:regularity} implies regularity of the informational virtual value at every upper cutoff.

\begin{lemma}[Regularity at Every Upper Cutoff]
    \label{lmm:regularity-cutoff}
    For every $b\in\rsbr{0,1}$, $\varphi_b$ is strictly increasing on $\sbr{0,b}$.
\end{lemma}

\begin{proof}
For $b=1$, this is \cref{asm:regularity}. Fix $b\in\br{0,1}$ and define
\[
    w\br{\theta}:=\frac{F\br{b}-F\br{\theta}}{1-F\br{\theta}}
    \quad\text{for }\theta\in\sbr{0,b}.
\]
The function $w$ takes values in $\sbr{0,1}$ and is nonincreasing. By \eqref{eqn:informational-virtual-value}, for $0\leq\theta<\theta'\leq b$,
\[
    \begin{aligned}
        \varphi_b\br{\theta}
        &=w\br{\theta}\varphi\br{\theta}
          +\br{1-w\br{\theta}}u\br{\theta}\\
        &<w\br{\theta}\varphi\br{\theta'}
          +\br{1-w\br{\theta}}u\br{\theta'}\\
        &\leq w\br{\theta'}\varphi\br{\theta'}
          +\br{1-w\br{\theta'}}u\br{\theta'}\\
        &=\varphi_b\br{\theta'},
    \end{aligned}
\]
where the first inequality uses \cref{asm:regularity} and the fact that $u$ is strictly increasing, and the second uses $w\br{\theta'}\leq w\br{\theta}$ and $u\br{\theta'}\geq\varphi\br{\theta'}$, which follows from $\mu u'\br{\theta'}>c\br{\theta'}$ as shown in \eqref{eqn:inequality}.
\end{proof}

We now show that the cutoff conditions \eqref{eqn:cutoff-a} and \eqref{eqn:cutoff-b} have a unique solution.

\begin{lemma}[Cutoffs]
    \label{lmm:cutoffs}
    There is a unique pair $0\leq\sutheta<\sotheta\leq1$ that satisfies \eqref{eqn:cutoff-a} and \eqref{eqn:cutoff-b}. Moreover, $\sutheta<\ctheta\leq\sotheta$.
\end{lemma}

\begin{proof}
We first show that \eqref{eqn:cutoff-b} determines uniquely the upper cutoff as a function $b\br{a}$ of the lower cutoff $a$. We then show that exactly one $a$ satisfies \eqref{eqn:cutoff-a} with upper cutoff $b\br{a}$.

Since $\overline V\br{\theta}=\br{1-\mu}u\br{\theta}$ and $c\br{\theta}=2\mu\br{1-\mu}\lambda u\br{\theta}$, a pair $0\leq a<b\leq1$ satisfies \eqref{eqn:cutoff-b} if and only if
\[
    \overline V\br{b}-\overline V\br{a}
    \geq\int_a^b c\br{r}\,dr,
    \quad \text{with equality if }b<1.
\]
Any such pair has $a<\ctheta$. If instead $a\geq\ctheta$, then $\overline V'<c$ on $\rsbr{a,b}$, so
\[
    \overline V\br{b}-\overline V\br{a}
    =\int_a^b\overline V'\br{r}\,dr
    <\int_a^b c\br{r}\,dr,
\]
contradicting \eqref{eqn:cutoff-b}.

Now fix $a\in\lsbr{0,\ctheta}$ and define
\[
    b\br{a}:=\max\cbr{
      \theta\in\sbr{a,1}:
      \overline V\br{\theta}-\overline V\br{a}\geq\int_a^\theta c\br{r}\,dr}.
\]
The difference $\overline V\br{\theta}-\overline V\br{a}-\int_a^\theta c\br{r}\,dr$ is zero at $\theta=a$, and its derivative $\overline V'\br{\theta}-c\br{\theta}$ is positive below $\ctheta$ and negative above it. The difference is therefore positive on $\rsbr{a,\ctheta}$ and crosses zero at most once above $\ctheta$. Hence $b\br{a}$ is this crossing point if it exists and $1$ otherwise, so $b\br{a}\geq\ctheta>a$, and $b\br{a}$ is the only $b>a$ that satisfies \eqref{eqn:cutoff-b} with lower cutoff $a$.

We next show that $b$ is nonincreasing and continuous, with $b\br{a}\to\ctheta$ as $a\uparrow\ctheta$. For $a<a'<\ctheta$ and $\theta\geq a'$, note that
\[
    \overline V\br{\theta}-\overline V\br{a'}-\int_{a'}^\theta c\br{r}\,dr
    <\overline V\br{\theta}-\overline V\br{a}-\int_a^\theta c\br{r}\,dr,
\]
because $\overline V'>c$ on $\lsbr{a,a'}$. Hence $b\br{a'}\leq b\br{a}$. For continuity, if $\ctheta=1$, then $b\br{a}=1$ for all $a\in\lsbr{0,\ctheta}$, so $b$ is constant. Otherwise, the difference is continuous in $a$ and $\theta$ and strictly decreasing in $\theta$ above $\ctheta$, so $b$ is continuous, including where $b\br{a}=1$.

If $\ctheta=1$, the limit follows from $b\br{a}=1$. Otherwise, fix any $\theta>\ctheta$. Then,
\[
    \lim_{a\uparrow\ctheta}
    \br{\overline V\br{\theta}-\overline V\br{a}-\int_a^\theta c\br{r}\,dr}
    =\int_{\ctheta}^{\theta}\br{\overline V'\br{r}-c\br{r}}\,dr<0,
\]
because $\overline V'<c$ above $\ctheta$, so $b\br{a}<\theta$ for all $a$ sufficiently close to $\ctheta$. Since $b\br{a}\geq\ctheta$ for every $a$ and $\theta>\ctheta$ was arbitrary, $b\br{a}\to\ctheta$.

We now show that exactly one $a\in\lsbr{0,\ctheta}$ satisfies \eqref{eqn:cutoff-a} with upper cutoff $b\br{a}$. Condition \eqref{eqn:cutoff-a} requires $\varphi_{b\br{a}}\br{a}\geq0$ if $a=0$ and $\varphi_{b\br{a}}\br{a}=0$ if $a>0$. We therefore study the properties of the function $a\mapsto\varphi_{b\br{a}}\br{a}$ on $\lsbr{0,\ctheta}$. We show that it is continuous, strictly increasing, and positive as $a$ approaches $\ctheta$. Its sign at $a=0$ then determines which case of \eqref{eqn:cutoff-a} applies.

The function is continuous because $b$ is continuous. It is also strictly increasing because, for $a_1<a_2<\ctheta$,
\[
    \varphi_{b\br{a_1}}\br{a_1}
    <\varphi_{b\br{a_1}}\br{a_2}
    \leq\varphi_{b\br{a_2}}\br{a_2},
\]
where the first inequality holds because $\varphi_{b\br{a_1}}$ is strictly increasing by \cref{lmm:regularity-cutoff}, and for the second we use \eqref{eqn:inequality} and $b\br{a_2}\leq b\br{a_1}$. Finally, since $b\br{a}\to\ctheta$ as $a\uparrow\ctheta$, $\varphi_{b\br{a}}\br{a}\to u\br{\ctheta}>0$.

Suppose first that $\varphi_{b\br{0}}\br{0}\geq0$. Then $a=0$ satisfies \eqref{eqn:cutoff-a}. We claim that no other lower cutoff does. For $a>0$, condition \eqref{eqn:cutoff-a} requires $\varphi_{b\br{a}}\br{a}=0$. But the function is strictly increasing, so $\varphi_{b\br{a}}\br{a}>\varphi_{b\br{0}}\br{0}\geq0$ for every $a>0$. Hence $a=0$ is the only lower cutoff that satisfies \eqref{eqn:cutoff-a}.

Suppose now that $\varphi_{b\br{0}}\br{0}<0$. Then $a=0$ does not satisfy \eqref{eqn:cutoff-a}. Since the function is continuous and strictly increasing, and is negative at $a=0$ and positive near $\ctheta$, by the intermediate value theorem there exist a unique $a\in\br{0,\ctheta}$ at which $\varphi_{b\br{a}}\br{a}=0$.

In either case, let $\sutheta$ denote this value of $a$ and set $\sotheta=b\br{\sutheta}$. Then the cutoffs $\sutheta$ and $\sotheta$ satisfy \eqref{eqn:cutoff-b} by the definition of $b$ and \eqref{eqn:cutoff-a} by the choice of $\sutheta$. No other pair satisfies both conditions. Indeed, take any pair $0\leq a<b\leq1$ that satisfies \eqref{eqn:cutoff-a} and \eqref{eqn:cutoff-b}. By the first step, \eqref{eqn:cutoff-b} requires $a<\ctheta$ and $b=b\br{a}$. By the second step, $\sutheta$ is the only $a<\ctheta$ that satisfies \eqref{eqn:cutoff-a} with upper cutoff $b\br{a}$. Hence $a=\sutheta$ and $b=b\br{\sutheta}=\sotheta$. Finally, $\sutheta<\ctheta$ and $\sotheta=b\br{\sutheta}\geq\ctheta$, so $\sutheta<\ctheta\leq\sotheta$.
\end{proof}

Using the cutoffs from \cref{lmm:cutoffs}, we now construct the candidate menu and verify that the corresponding pair $\br{\pi_G^*,U^*}$ is feasible in \eqref{eqn:loss}. First, set $\pi_B^*\br{\theta}=1$ for every $\theta$. Then, for $\theta<\sutheta$, setting $U^*\br{\theta}=0$ and imposing equality in \eqref{eqn:local} gives
\[
    \pi_G^*\br{\theta}
    =\frac{\br{2\mu-1}u'\br{\theta}}{\mu u'\br{\theta}-c\br{\theta}}
    =\frac{2\mu-1}{\mu}
        \br{1+\frac{c\br{\theta}}{\mu u'\br{\theta}-c\br{\theta}}}.
\]
By \eqref{eqn:ld-decomposition}, the distortion is
\[
    1-\pi_G^*\br{\theta}
    =\frac{\overline V'\br{\theta}-c\br{\theta}}
          {\mu u'\br{\theta}-c\br{\theta}}.
\]
The ratio $c/\br{\mu u'-c}$ is strictly increasing, because $c$ is strictly increasing and $\mu u'-c$ is nonincreasing and positive by \eqref{eqn:inequality}. Since also $\overline V'>c$ below $\sutheta$, $\pi_G^*$ is strictly increasing below $\sutheta$, and $\br{2\mu-1}/\mu<\pi_G^*\br{\theta}<1$. Hence \eqref{eqn:responsive} holds strictly and $V\br{\theta,\pi^*\br{\theta}}>0$.

For $\theta\geq\sutheta$, set $\pi_G^*\br{\theta}=1$, so the client receives full information. Then define
\[
    U^*\br{\theta} =
    \begin{cases}
        0 & \text{if } 0 \leq \theta < \sutheta,\\
        \displaystyle\overline V\br{\theta}-\overline V\br{\sutheta}
          -\int_{\sutheta}^{\theta}c\br{r}\,dr
          & \text{if } \sutheta \leq \theta \leq \sotheta,\\
        0 & \text{if } \sotheta < \theta \leq 1,
    \end{cases}
\]
where $U^*>0$ on $\br{\sutheta,\sotheta}$ by the definition of $b$ in the proof of \cref{lmm:cutoffs}, and $U^*\br{\sotheta}=0$ if $\sotheta<1$. If $\sotheta=1$, $U^*\br{1}$ may be positive. Thus $U^*$ is nonnegative and absolutely continuous.

We now verify that the pair $\br{\pi_G^*,U^*}$ is feasible in \eqref{eqn:loss}. By construction, \eqref{eqn:local} binds $dF$-almost everywhere below $\sotheta$. Below $\sutheta$, both sides are zero, and between the cutoffs, $\pi_G^*=1$ and $U^{*\prime}=\overline V'-c$. Above $\sotheta$, $U^*=0$, $\pi_G^*=1$, and $\overline V'\leq c$, so \eqref{eqn:local} holds there as well. Since also $U^*\geq0$ and $\br{2\mu-1}/\mu\leq\pi_G^*\leq1$, the pair is feasible in \eqref{eqn:loss}.

Finally, we set $t^*\br{\theta}=V\br{\theta,\pi^*\br{\theta}}-U^*\br{\theta}$, which gives the transfers in the three regions of \cref{thm:characterizationcontinuumtypes}. For $\theta\in\sbr{\sutheta,\sotheta}$, this is $t^*\br{\theta}=\overline V\br{\sutheta}+\int_{\sutheta}^{\theta}c\br{r}\,dr$. Thus $t^{*\prime}\br{\theta}=c\br{\theta}$ for $\theta\in\br{\sutheta,\sotheta}$. Since $c$ is strictly increasing and $\overline V'$ is nonincreasing, $t^*$ is strictly convex and $U^*$ is strictly concave on $\sbr{\sutheta,\sotheta}$, as illustrated in \cref{fig:continuum-menus,fig:continuum-surplus}.

\subsubsection{The Revenue Loss Bound}
We now show that the pair $\br{\pi_G^*,U^*}$ solves the relaxed problem \eqref{eqn:loss}. We refer to this pair as the \textit{candidate}. We use the Lagrangian sufficiency theorem of \citet{luenberger1997optimization}. As in \citet[Appendix~C]{amador2013theory}, we construct the multipliers explicitly and verify that the candidate minimizes the resulting Lagrangian.

Let the feasible set consist of all pairs $\br{\pi_G,U}$ such that $\pi_G$ is measurable with $\br{2\mu-1}/\mu\leq\pi_G\br{\theta}\leq1$ for every $\theta\in\Theta$, and $U$ is absolutely continuous. We assign a multiplier function $z:\Theta\to\mathbb R_+$ to \eqref{eqn:local} and a nondecreasing cumulative multiplier function $\Lambda:\Theta\to\mathbb R_+$ to $U\geq0$. Using the representation in \eqref{eqn:ld-decomposition}, the Lagrangian for \eqref{eqn:loss} is then defined as
\begin{equation}
    \label{eqn:lagrangian}
    \begin{aligned}
        &\mathcal L\br{\pi_G,U\mid z,\Lambda}
        :=J\br{\pi_G,U}-\int_{\sbr{0,1}}U\br{r}\,d\Lambda\br{r}\\
        &\quad-\int_0^1z\br{r}\Bigl[U'\br{r}
          -\br{\overline V'\br{r}-c\br{r}}
          +\br{\mu u'\br{r}-c\br{r}}\br{1-\pi_G\br{r}}\Bigr]\,dr,
    \end{aligned}
\end{equation}
where $z$ is nonnegative and $\Lambda$ is nondecreasing with nonnegative mass $\Lambda\br{0}=z\br{0}$ at zero.

A pair $\br{\pi_G,U}$ feasible in \eqref{eqn:loss} and multipliers $\br{z,\Lambda}$ satisfy \textit{complementary slackness} if both integrals in \eqref{eqn:lagrangian} are zero. Thus, $z$ is zero almost everywhere where \eqref{eqn:local} is slack, and $d\Lambda$ assigns no mass to the set where $U>0$. This means that $\Lambda$ is constant on each interval where $U>0$, but it can be positive because it may include weight assigned to lower types.

By the sufficiency theorem of \citet[Section~8.4, Theorem~1]{luenberger1997optimization}, as stated by \citet[Appendix~B, Theorem~1]{amador2013theory}, a pair that is feasible in \eqref{eqn:loss} solves \eqref{eqn:loss} if there are valid multipliers for which the pair minimizes the Lagrangian over the feasible set and satisfies complementary slackness. We propose the multipliers
\[
    z\br{\theta} =
    \begin{cases}
        \displaystyle\frac{\mu u\br{\theta}f\br{\theta}}
                          {\mu u'\br{\theta}-c\br{\theta}}
          & \text{if } 0 \leq \theta < \sutheta,\\
        F\br{\sotheta}-F\br{\theta}
          & \text{if } \sutheta \leq \theta \leq \sotheta,\\
        0 & \text{if } \sotheta < \theta \leq 1,
    \end{cases}
\]
and $\Lambda\br{\theta}:=F\br{\theta}+z\br{\theta}$ for every $\theta\in\Theta$.

\begin{lemma}[Revenue Loss Bound]
    \label{lmm:revenue-bound}
    Every pair $\br{\pi_G,U}$ that is feasible in \eqref{eqn:loss} satisfies
    \[
        J\br{\pi_G,U}
        \geq\int_0^1z\br{r}\br{\overline V'\br{r}-c\br{r}}\,dr,
    \]
    with equality if and only if $\pi_G=\pi_G^*$ $dF$-almost everywhere and $U=U^*$ everywhere.
\end{lemma}

\begin{proof}
We first simplify the Lagrangian at the proposed multipliers. The simplification also explains how the multipliers are constructed. We then verify that the multipliers are valid. Finally, we prove the revenue loss bound, show that the candidate attains it, and show that equality holds only if $\pi_G=\pi_G^*$ $dF$-almost everywhere and $U=U^*$.

\paragraph{Simplifying the Lagrangian.}
Fix a pair $\br{\pi_G,U}$ in the feasible set. We assume for now that $\Lambda$ is nondecreasing and $z\br{1}=0$, and we verify these properties below when we show that the multipliers are valid. Since $U$ is absolutely continuous,
\[
    U\br{\theta}=U\br{0}+\int_0^{\theta}U'\br{r}\,dr \quad \text{for every } \theta\in\Theta.
\]
Substituting this expression for $U$ and applying Fubini's theorem gives
\[
    \int_{\sbr{0,1}}U\br{r}\,d\Lambda\br{r}
    =U\br{0}\Lambda\br{1}
      +\int_0^1U'\br{r}\br{\Lambda\br{1}-\Lambda\br{r}}\,dr
\]
and
\[
    \int_0^1U\br{r}f\br{r}\,dr
    =U\br{0}
      +\int_0^1U'\br{r}\br{1-F\br{r}}\,dr.
\]
Since $\Lambda=F+z$ and $z\br{1}=0$, we have $\Lambda\br{1}=1$ and $\Lambda\br{1}-\Lambda\br{r}=1-F\br{r}-z\br{r}$. Subtracting the first equality from the second therefore gives
\[
    \int_0^1U\br{r}f\br{r}\,dr
    -\int_{\sbr{0,1}}U\br{r}\,d\Lambda\br{r}
    =\int_0^1z\br{r}U'\br{r}\,dr.
\]
Each integral in \eqref{eqn:lagrangian} is finite, because $U$ is absolutely continuous, the other functions are bounded, and $\Lambda$ is nondecreasing. Substituting the preceding equality in \eqref{eqn:lagrangian}, every term involving $U$ therefore cancels:
\begin{equation}
    \label{eqn:lagrangian-reduced}
    \begin{aligned}
        &\mathcal L\br{\pi_G,U\mid z,\Lambda}
        =\int_0^1z\br{r}\br{\overline V'\br{r}-c\br{r}}\,dr\\
        &\quad+\int_0^1
          \Bigl[\mu u\br{r}f\br{r}
          -\br{\mu u'\br{r}-c\br{r}}z\br{r}\Bigr]
          \br{1-\pi_G\br{r}}\,dr,
    \end{aligned}
\end{equation}
which is independent of $U$ and affine in $\pi_G$.

This simplification explains the construction of the multipliers. Since the feasible set does not restrict the sign of $U$, the Lagrangian has a minimum only if all terms involving $U$ cancel, which we achieve by choosing $\Lambda=F+z$ and $z\br{1}=0$.

The values of $z$ then follow from the two conditions of the sufficiency theorem. Above $\sotheta$, \eqref{eqn:local} is slack at the candidate, so complementary slackness requires $z=0$. Hence $z\br{\sotheta}=0$, by monotonicity of $\Lambda=F+z$ if $\sotheta<1$ and by $z\br{1}=0$ if $\sotheta=1$. Between the cutoffs, $U^*>0$, so complementary slackness requires $\Lambda$ to be constant. Since $\Lambda=F+z$ and $z\br{\sotheta}=0$, the constant is $F\br{\sotheta}$, which gives $z\br{\theta}=F\br{\sotheta}-F\br{\theta}$. Below $\sutheta$, $\pi_G^*\br{\theta}$ lies strictly between its bounds, so the candidate minimizes the Lagrangian only if the coefficient on $1-\pi_G$ in \eqref{eqn:lagrangian-reduced} is zero, which gives
\[
    z\br{\theta}
    =\frac{\mu u\br{\theta}f\br{\theta}}
          {\mu u'\br{\theta}-c\br{\theta}}
    \quad\text{for }\theta<\sutheta.
\]

\paragraph{Validity of the Multipliers.}
We next verify that the proposed multipliers are valid, that is, that $z$ is nonnegative and $\Lambda$ is nondecreasing with $\Lambda\br{0}\geq0$. We also verify that $z\br{1}=0$, as assumed in the simplification.

By \eqref{eqn:inequality}, the multiplier $z$ is nonnegative. Moreover, $z\br{1}=0$, because $z=0$ on $\sbr{\sotheta,1}$. If $\sutheta>0$, $z$ is continuous at $\sutheta$ because \eqref{eqn:cutoff-a} gives
\[
    \frac{\mu u\br{\sutheta}f\br{\sutheta}}
         {\mu u'\br{\sutheta}-c\br{\sutheta}}
    =F\br{\sotheta}-F\br{\sutheta}.
\]

We now show that $\Lambda$ is nondecreasing. It is nondecreasing at and above $\sutheta$, because it is constant at $F\br{\sotheta}$ on $\sbr{\sutheta,\sotheta}$ and equals $F$ above $\sotheta$. To show monotonicity for $\theta<\sutheta$, we need to consider only the case $\sutheta>0$. Then we have
\begin{equation}
    \label{eqn:negative-virtual-value}
    \varphi\br{\theta}
    \leq\varphi_{\sotheta}\br{\theta}
    <\varphi_{\sotheta}\br{\sutheta}
    =0
    \quad\text{for }\theta<\sutheta,
\end{equation}
where the first inequality holds because $\varphi_{\sotheta}\br{\theta}$ is a weighted average of $\varphi\br{\theta}$ and $u\br{\theta}\geq\varphi\br{\theta}$, as shown in the proof of \cref{lmm:regularity-cutoff}, the second holds by \cref{lmm:regularity-cutoff}, and the equality follows from \eqref{eqn:cutoff-a}. By \eqref{eqn:informational-virtual-value},
\[
    \begin{aligned}
        \Lambda\br{\theta}
        &=F\br{\theta}
          +\frac{\mu u\br{\theta}f\br{\theta}}
                {\mu u'\br{\theta}-c\br{\theta}}\\
        &=F\br{\theta}
          +\frac{u\br{\theta}\br{1-F\br{\theta}}}
                {u\br{\theta}-\varphi\br{\theta}}\\
        &=1-\frac{1-F\br{\theta}}
                  {1+u\br{\theta}/\br{-\varphi\br{\theta}}}.
    \end{aligned}
\]
The numerator $1-F\br{\theta}$ is strictly decreasing in $\theta$. The denominator is positive and strictly increasing in $\theta$, because $u$ is strictly increasing and $-\varphi$ is positive by \eqref{eqn:negative-virtual-value} and strictly decreasing by \cref{asm:regularity}. Thus $\Lambda$ is strictly increasing below $\sutheta$. Since $\Lambda=F+z$ is continuous at $\sutheta$, it is nondecreasing on $\Theta$, with $\Lambda\br{0}=z\br{0}\geq0$.

\paragraph{The Revenue Loss Bound.}
We now prove the revenue loss bound and show that the candidate attains it. Substituting the definition of $z$ in the coefficient on $1-\pi_G$ in \eqref{eqn:lagrangian-reduced} gives
\begin{equation}
    \label{eqn:coefficient}
    \mu u\br{\theta}f\br{\theta}
    -\br{\mu u'\br{\theta}-c\br{\theta}}z\br{\theta}
    =\begin{cases}
        0 & \text{if } 0 \leq \theta < \sutheta,\\
        \mu f\br{\theta}\varphi_{\sotheta}\br{\theta}
          & \text{if } \sutheta \leq \theta \leq \sotheta,\\
        \mu u\br{\theta}f\br{\theta}
          & \text{if } \sotheta < \theta \leq 1.
    \end{cases}
\end{equation}
Since $\varphi_{\sotheta}\br{\sutheta}\geq0$ by \eqref{eqn:cutoff-a} and $\varphi_{\sotheta}$ is strictly increasing by \cref{lmm:regularity-cutoff}, this coefficient is nonnegative, with strict inequality for $\theta>\sutheta$.

Now fix a pair $\br{\pi_G,U}$ that is feasible in \eqref{eqn:loss}. Then
\begin{equation}
    \label{eqn:bound-inequalities}
    J\br{\pi_G,U}
    \geq\mathcal L\br{\pi_G,U\mid z,\Lambda}
    \geq\int_0^1z\br{r}\br{\overline V'\br{r}-c\br{r}}\,dr,
\end{equation}
where the first inequality holds by definition of the Lagrangian, and the second inequality holds because the second integral in \eqref{eqn:lagrangian-reduced} is nonnegative, since the coefficient on $1-\pi_G$ is nonnegative and $\pi_G\leq1$.

At the candidate, both inequalities in \eqref{eqn:bound-inequalities} hold with equality. For the first, both integrals in \eqref{eqn:lagrangian} are zero, that is, complementary slackness holds. The first integral is zero because $U^*\br{0}=0$ and $\Lambda$ is constant wherever $U^*>0$, and the second is zero because, by construction, \eqref{eqn:local} binds $dF$-almost everywhere below $\sotheta$ and $z=0$ above $\sotheta$. For the second, the second integral in \eqref{eqn:lagrangian-reduced} is zero, because the coefficient on $1-\pi_G$ is zero below $\sutheta$ and $\pi_G^*=1$ at and above $\sutheta$. Thus the candidate minimizes the Lagrangian over the feasible set and satisfies complementary slackness, and
\[
    J\br{\pi_G^*,U^*}
    =\int_0^1z\br{r}\br{\overline V'\br{r}-c\br{r}}\,dr.
\]
Since the candidate is feasible in \eqref{eqn:loss}, it therefore solves \eqref{eqn:loss}.

\paragraph{Uniqueness in the Relaxed Problem.}
It remains to show that equality holds only if $\pi_G=\pi_G^*$ $dF$-almost everywhere and $U=U^*$. We use this in the proof of \cref{thm:characterizationcontinuumtypes} to show that the optimal menu is unique.

Suppose a pair $\br{\pi_G,U}$ that is feasible in \eqref{eqn:loss} attains the bound. Then both inequalities in \eqref{eqn:bound-inequalities} hold with equality. The first requires both integrals in \eqref{eqn:lagrangian} to be zero, because each is nonnegative, and the second requires the second integral in \eqref{eqn:lagrangian-reduced} to be zero. In particular,
\begin{equation}
    \label{eqn:zero-surplus-integral}
    \int_{\sbr{0,1}}U\br{r}\,d\Lambda\br{r}=0.
\end{equation}
Since $U\geq0$, this requires $U$ to be zero $d\Lambda$-almost everywhere. Since $d\Lambda$ has an atom of mass $\Lambda\br{0}=z\br{0}>0$ at zero, we have $U\br{0}=0$. Here $z\br{0}>0$ because $z\br{0}=\mu u\br{0}f\br{0}/\br{\mu u'\br{0}-c\br{0}}$ is positive by \eqref{eqn:inequality} if $\sutheta>0$, and $z\br{0}=F\br{\sotheta}>0$ if $\sutheta=0$.

Next, if $\sutheta>0$, we showed that $\Lambda$ is strictly increasing on $\sbr{0,\sutheta}$. If $\sotheta<1$, it is strictly increasing on $\sbr{\sotheta,1}$, where $\Lambda=F$. On each of these intervals, $d\Lambda$ therefore assigns positive mass to every open subinterval. Since $U$ is continuous and nonnegative, \eqref{eqn:zero-surplus-integral} implies that $U$ equals zero on each of these intervals, including the endpoints. Thus $U=U^*$ below $\sutheta$ and above $\sotheta$. Moreover, $U\br{\sutheta}=0$, as the endpoint of $\sbr{0,\sutheta}$ if $\sutheta>0$ and because $U\br{0}=0$ if $\sutheta=0$.

Next, the coefficient in \eqref{eqn:coefficient} is strictly positive for $\theta>\sutheta$. Since $1-\pi_G\geq0$ and the second integral in \eqref{eqn:lagrangian-reduced} is zero, $\pi_G\br{\theta}=1$ $dF$-almost everywhere above $\sutheta$. Similarly, the integrand of the second integral in \eqref{eqn:lagrangian} is nonnegative, because $z\geq0$ and the pair satisfies \eqref{eqn:local}. Since this integral is zero and $z>0$ below $\sotheta$, \eqref{eqn:local} binds $dF$-almost everywhere below $\sotheta$.

Then, below $\sutheta$, $U=0$ gives $U'=0$. Since \eqref{eqn:local} binds $dF$-almost everywhere, its right side is zero $dF$-almost everywhere:
\[
    0=\br{\mu\pi_G\br{\theta}+1-2\mu}u'\br{\theta}
      -c\br{\theta}\pi_G\br{\theta}
    =\br{\mu u'\br{\theta}-c\br{\theta}}\pi_G\br{\theta}
      -\br{2\mu-1}u'\br{\theta}.
\]
Since $\mu u'\br{\theta}-c\br{\theta}>0$ by \eqref{eqn:inequality}, solving for $\pi_G\br{\theta}$ gives
\[
    \pi_G\br{\theta}=\frac{\br{2\mu-1}u'\br{\theta}}{\mu u'\br{\theta}-c\br{\theta}}=\pi_G^*\br{\theta}.
\]
Hence $\pi_G=\pi_G^*$ $dF$-almost everywhere below $\sutheta$. Between the cutoffs, $\pi_G\br{\theta}=1$ $dF$-almost everywhere, so by \eqref{eqn:ld-decomposition} the right side of \eqref{eqn:local} equals $\overline V'\br{\theta}-c\br{\theta}$. Since \eqref{eqn:local} binds, $U'\br{\theta}=\overline V'\br{\theta}-c\br{\theta}$ $dF$-almost everywhere. Since $U$ is absolutely continuous, integrating from $U\br{\sutheta}=0$ gives
\[
    U\br{\theta}
    =\overline V\br{\theta}-\overline V\br{\sutheta}
      -\int_{\sutheta}^{\theta}c\br{r}\,dr
    =U^*\br{\theta}
    \quad\text{for all }\theta\in\sbr{\sutheta,\sotheta}.
\]
Hence $\pi_G=\pi_G^*$ $dF$-almost everywhere and $U=U^*$ everywhere, concluding the proof. \end{proof}

\subsubsection{Global Incentive Compatibility}
By \cref{lmm:revenue-bound}, the pair $\br{\pi_G^*,U^*}$ solves \eqref{eqn:loss}. It remains to show that it is feasible in \eqref{eqn:problem}, that is, that the candidate menu satisfies \eqref{eqn:IC} and \eqref{eqn:IR}.

\begin{lemma}[Global Incentive Compatibility]
    \label{lmm:global-ic}
    The candidate menu $\br{\pi^*,t^*}$ is a simple menu that satisfies \eqref{eqn:IC} and \eqref{eqn:IR}.
\end{lemma}

\begin{proof}
The candidate information products satisfy \eqref{eqn:responsive} by construction, so the candidate menu is simple, and $U^*\geq0$ gives \eqref{eqn:IR}. We verify \eqref{eqn:IC} in three cases: reports below $\sutheta$, reports at or above $\sutheta$ by a lower type, and reports at or above $\sutheta$ by a higher type.

\paragraph{Deviations to Distorted Information Products.}
We first show that each distorted information product in the candidate menu is worth no more to any other type than to the type whose project it concerns.

Fix a report $\widetilde{\theta}<\sutheta$. For a given type $\theta$, let $\Delta=|\theta-\widetilde{\theta}|$. By \eqref{eqn:value},
\[
    V\br{\theta,\pi^*\br{\widetilde{\theta}}}
    =u\br{\theta}\br{C+De^{-\lambda\Delta}}_+,
\]
where
\[
    C=-\br{2\mu-1}\br{1-\mu\pi_G^*\br{\widetilde{\theta}}}<0
    \quad\text{and}\quad
    D=2\mu\br{1-\mu}\pi_G^*\br{\widetilde{\theta}}>0.
\]
The expression $C+De^{-\lambda\Delta}$ is the gain in the probability of a correct decision for a type at distance $\Delta$ who follows both recommendations. At $\theta=\widetilde{\theta}$, this gain is $C+D=\mu\pi_G^*\br{\widetilde{\theta}}+1-2\mu>0$. For a lower type $\theta<\widetilde{\theta}$, the gain is at most $C+D$. Hence
\[
    V\br{\theta,\pi^*\br{\widetilde{\theta}}}
    \leq u\br{\theta}\br{C+D}
    <u\br{\widetilde{\theta}}\br{C+D}
    =V\br{\widetilde{\theta},\pi^*\br{\widetilde{\theta}}},
\]
where the strict inequality uses $u\br{\theta}<u\br{\widetilde{\theta}}$ and $C+D>0$.

For a higher type $\theta>\widetilde{\theta}$ with $V\br{\theta,\pi^*\br{\widetilde{\theta}}}>0$, we have $0<Ce^{\lambda\Delta}+D<C+D$. Hence
\begin{equation}
    \label{eqn:distorted-log-bound}
    \begin{aligned}
        \frac{\partial}{\partial\theta}
        \log V\br{\theta,\pi^*\br{\widetilde{\theta}}}
        &=\frac{u'\br{\theta}}{u\br{\theta}}
          -\frac{\lambda D}{Ce^{\lambda\Delta}+D}\\
        &<\frac{u'\br{\widetilde{\theta}}}{u\br{\widetilde{\theta}}}
          -\frac{\lambda D}{C+D}\\
        &=0,
    \end{aligned}
\end{equation}
where the inequality uses $Ce^{\lambda\Delta}+D<C+D$ and the fact that $u'/u$ is nonincreasing. The last equality holds because $\br{C+D}u'\br{\widetilde{\theta}}-\lambda D u\br{\widetilde{\theta}}$ is the right side of \eqref{eqn:local} at $\widetilde{\theta}$, which is zero by the definition of $\pi_G^*\br{\widetilde{\theta}}$. Hence $V\br{\theta,\pi^*\br{\widetilde{\theta}}}$ is strictly decreasing in $\theta$ for $\theta>\widetilde{\theta}$ wherever it is positive, and it stays zero once it reaches zero, because $C+De^{-\lambda\Delta}$ is strictly decreasing in $\Delta$.

Together with the comparison for lower types, this gives
\[
    V\br{\theta,\pi^*\br{\widetilde{\theta}}}
    \leq V\br{\widetilde{\theta},\pi^*\br{\widetilde{\theta}}}
    =t^*\br{\widetilde{\theta}}
    \quad\text{for every }\theta\in\Theta,
\]
where the equality holds because $U^*\br{\widetilde{\theta}}=0$. Hence deviations to $\widetilde{\theta}$ yield at most zero surplus.

\paragraph{Upward Deviations to Full Information.}
We next rule out upward deviations to reports at or above $\sutheta$. A type $\theta$ who makes such a report buys full information about a higher project. The argument has two steps. We first show that $t^*$ is strictly increasing on $\sbr{\sutheta,1}$, so a higher report costs more. We then use the fact that full information about another project is worth at most $\overline V\br{\theta}$ to type $\theta$ to show that the report yields less surplus than truth-telling.

Since $t^{*\prime}\br{\theta}=c\br{\theta}>0$ for $\theta\in\br{\sutheta,\sotheta}$, $t^*$ is strictly increasing on $\sbr{\sutheta,\sotheta}$. It is also strictly increasing on $\rsbr{\sotheta,1}$, where $t^*=\overline V$. If $\sotheta<1$, $t^*$ is continuous at $\sotheta$ because $U^*\br{\sotheta}=0$. Hence $t^*$ is strictly increasing on $\sbr{\sutheta,1}$, with $t^*\br{\sutheta}=\overline V\br{\sutheta}$.

Now fix a report $\widetilde{\theta}\geq\sutheta$ and a type $\theta<\widetilde{\theta}$. By \eqref{eqn:value} with $\pi_G\br{\widetilde{\theta}}=1$, $V\br{\theta,\overline\pi\br{\widetilde{\theta}}}\leq\overline V\br{\theta}$. If $\theta<\sutheta$, then
\[
    V\br{\theta,\overline\pi\br{\widetilde{\theta}}}
    \leq\overline V\br{\theta}
    <\overline V\br{\sutheta}
    \leq t^*\br{\widetilde{\theta}},
\]
so the report yields negative surplus, while $U^*\br{\theta}=0$. If $\theta\in\lsbr{\sutheta,\widetilde{\theta}}$, then $\pi^*\br{\theta}=\overline\pi\br{\theta}$ and $t^*\br{\theta}<t^*\br{\widetilde{\theta}}$, so the report yields strictly less surplus than truth-telling.

\paragraph{Downward Deviations to Full Information.}
We finally rule out downward deviations to reports at or above $\sutheta$. Fix a report $\widetilde{\theta}\geq\sutheta$ and consider a type $\theta>\widetilde{\theta}$, with $\Delta=\theta-\widetilde{\theta}$. The surplus of type $\theta$ from truth-telling less his surplus from this report is
\begin{equation}
    \label{eqn:downward-surplus-difference}
    U^*\br{\theta}
    -\br{V\br{\theta,\overline\pi\br{\widetilde{\theta}}}
           -t^*\br{\widetilde{\theta}}},
\end{equation}
which is zero at $\theta=\widetilde{\theta}$. We show that the expression in \eqref{eqn:downward-surplus-difference} is nonnegative for every higher type.

We split the higher types by whether the information product is valuable to them or not. If $V\br{\theta,\overline\pi\br{\widetilde{\theta}}}=0$, the report yields $-t^*\br{\widetilde{\theta}}<0\leq U^*\br{\theta}$, because $t^*\br{\widetilde{\theta}}\geq\overline V\br{\sutheta}>0$, so the deviation is unprofitable. It remains to consider types for which the value is positive. By \eqref{eqn:value}, once the value is zero, it stays zero for all higher types, so these types form an interval that starts at $\widetilde{\theta}$. On this interval, \eqref{eqn:downward-surplus-difference} is zero at $\theta=\widetilde{\theta}$, and we show that it never becomes negative as $\theta$ increases.

We proceed in two steps, because $U^*$ takes a different form below and above $\sotheta$. Below $\sotheta$, $U^*$ has derivative equal to $\overline V'-c$, and we show that \eqref{eqn:downward-surplus-difference} is nondecreasing in $\theta$ while $V\br{\theta,\overline\pi\br{\widetilde{\theta}}}>0$. Above $\sotheta$, $U^*=0$, and we show that $V\br{\theta,\overline\pi\br{\widetilde{\theta}}}$ is nonincreasing in $\theta$, so the surplus from the report stays at most zero.

For $\widetilde{\theta}<\theta<\sotheta$ with $V\br{\theta,\overline\pi\br{\widetilde{\theta}}}>0$, \eqref{eqn:value} and $U^{*\prime}\br{\theta}=\overline V'\br{\theta}-c\br{\theta}$ give
\begin{equation}
    \label{eqn:downward-full-information}
    \begin{aligned}
        \frac{d}{d\theta}
        \br{U^*\br{\theta}-V\br{\theta,\overline\pi\br{\widetilde{\theta}}}}
        &=\overline V'\br{\theta}-c\br{\theta}\\
        &\quad-\br{1-\mu}\Bigl[
             \br{1-2\mu+2\mu e^{-\lambda\Delta}}u'\br{\theta}
             -2\mu\lambda e^{-\lambda\Delta}u\br{\theta}\Bigr]\\
        &=2\mu\br{1-\mu}\br{1-e^{-\lambda\Delta}}
          \br{u'\br{\theta}-\lambda u\br{\theta}}\geq0,
    \end{aligned}
\end{equation}
where the inequality uses condition \eqref{eqn:continuum-growth} and the fact that $u'/u$ is nonincreasing.\footnote{With $u$ concave, condition \eqref{eqn:continuum-growth} is a \textit{sorting condition}: the gain from a fully informative product about a closer project is weakly larger for higher types.} Thus \eqref{eqn:downward-surplus-difference} is nondecreasing while $\widetilde{\theta}<\theta<\sotheta$ and $V\br{\theta,\overline\pi\br{\widetilde{\theta}}}>0$. Since it is zero at $\theta=\widetilde{\theta}$ and continuous in $\theta$, it is nonnegative throughout this interval. Together with the types for which the value is zero, this shows that \eqref{eqn:downward-surplus-difference} is nonnegative for every $\theta\in\rsbr{\widetilde{\theta},\sotheta}$.

For $\theta>\sotheta$, $\overline V'\br{\theta}\leq c\br{\theta}$, that is, $u'\br{\theta}\leq2\mu\lambda u\br{\theta}$. By the bound on its derivative in \eqref{eqn:full-information-monotonicity}, $V\br{\theta,\overline\pi\br{\widetilde{\theta}}}$ is nonincreasing in $\theta$ above $\sotheta$ wherever it is positive. If $\widetilde{\theta}\leq\sotheta<\theta$, then
\[
    V\br{\theta,\overline\pi\br{\widetilde{\theta}}}
    \leq V\br{\sotheta,\overline\pi\br{\widetilde{\theta}}}
    \leq t^*\br{\widetilde{\theta}},
\]
where the first inequality holds because $V\br{\theta,\overline\pi\br{\widetilde{\theta}}}$ is nonincreasing in $\theta$ above $\sotheta$, and the second because \eqref{eqn:downward-surplus-difference} is nonnegative at $\sotheta$ and $U^*\br{\sotheta}=0$, since $\sotheta<1$. Thus the report yields at most zero surplus, while $U^*\br{\theta}=0$. If $\widetilde{\theta}>\sotheta$, the same argument applies with $\widetilde{\theta}$ in place of $\sotheta$, because \eqref{eqn:downward-surplus-difference} is zero at $\widetilde{\theta}$ and $U^*\br{\widetilde{\theta}}=0$. Hence \eqref{eqn:downward-surplus-difference} is nonnegative for every $\theta>\widetilde{\theta}$.

These three cases cover every deviation, so the candidate menu satisfies \eqref{eqn:IC} and \eqref{eqn:IR}.
\end{proof}

\subsubsection{Completion of the Proof}
\begin{proof}[Proof of \cref{thm:characterizationcontinuumtypes}]
Fix a simple menu $\br{\pi,t}$ that satisfies \eqref{eqn:IC} and \eqref{eqn:IR}, and let $U$ denote its surplus. By \cref{lmm:one-parameter}, there is a simple menu $\br{\widehat\pi,t}$ with $\widehat\pi_B=1$ that satisfies \eqref{eqn:IC} and \eqref{eqn:IR} and has the same transfers and surplus. Its recommendation probability $\widehat\pi_G$ satisfies
\[
    \widehat\pi_G\br{\theta}
    =\pi_G\br{\theta}
      -\frac{1-\mu}{\mu}\br{1-\pi_B\br{\theta}}.
\]
The pair $\br{\widehat\pi_G,U}$ is feasible in \eqref{eqn:problem}, hence in \eqref{eqn:loss}, and the expert's expected revenue from $\br{\pi,t}$ equals
\[
\int_0^1\overline V\br{r}f\br{r}\,dr-J\br{\widehat\pi_G,U}.
\]
We use this pair to prove optimality and then uniqueness.

\paragraph{Optimality.}
By \cref{lmm:revenue-bound}, the candidate solves \eqref{eqn:loss}. By \cref{lmm:global-ic}, it is feasible in \eqref{eqn:problem}, so it also solves \eqref{eqn:problem}. Therefore $J\br{\widehat\pi_G,U}\geq J\br{\pi_G^*,U^*}$, so the menu $\br{\pi,t}$ earns at most the revenue of the candidate menu $\br{\pi^*,t^*}$. Hence $\br{\pi^*,t^*}$ is optimal, and its cutoffs are unique by \cref{lmm:cutoffs}.

\paragraph{Uniqueness.}
Suppose now that $\br{\pi,t}$ is optimal. Then $J\br{\widehat\pi_G,U}=J\br{\pi_G^*,U^*}$, so the pair attains the revenue loss bound, and \cref{lmm:revenue-bound} gives $\widehat\pi_G=\pi_G^*$ $dF$-almost everywhere and $U=U^*$ everywhere.

The revenue loss bound therefore determines $U$ and $\widehat\pi_G$. It does not yet determine the original information products, because $\widehat\pi_G$ combines $\pi_G$ and $\pi_B$. To recover them, it remains to show that $\pi_B\br{\theta}=1$ $dF$-almost everywhere. We show this separately for types below and above $\sutheta$.

We start from the lower bound in \eqref{eqn:envelope-bounds} for the original menu. To use the formula for $V_\theta^+$ in \cref{lmm:local-ic}, we need the value $V\br{\theta,\pi\br{\theta}}$ to be positive. This holds $dF$-almost everywhere, because $V\br{\theta,\pi\br{\theta}}=\br{\mu\widehat\pi_G\br{\theta}+1-2\mu}u\br{\theta}$, $\widehat\pi_G=\pi_G^*$ $dF$-almost everywhere, and $V\br{\theta,\pi^*\br{\theta}}>0$ for every $\theta$. Substituting
\[
    \pi_G\br{\theta}+\pi_B\br{\theta}-1
    =\widehat\pi_G\br{\theta}
      -\frac{2\mu-1}{\mu}\br{1-\pi_B\br{\theta}}
\]
in the formula for $V_\theta^+$, the lower bound in \eqref{eqn:envelope-bounds} becomes
\[
    \begin{aligned}
        U'\br{\theta}
        &\geq\br{\mu\widehat\pi_G\br{\theta}+1-2\mu}u'\br{\theta}
          -c\br{\theta}
            \br{\widehat\pi_G\br{\theta}-\frac{2\mu-1}{\mu}\br{1-\pi_B\br{\theta}}}\\
        &=\br{\mu\widehat\pi_G\br{\theta}+1-2\mu}u'\br{\theta}
          -c\br{\theta}\widehat\pi_G\br{\theta}
          +\frac{2\mu-1}{\mu}c\br{\theta}\br{1-\pi_B\br{\theta}}
    \end{aligned}
\]
$dF$-almost everywhere. The first two terms after the equality are the right side of \eqref{eqn:local} at $\widehat\pi_G\br{\theta}$. Below $\sutheta$, \eqref{eqn:local} binds for the candidate. Since $U=U^*$ and $\widehat\pi_G=\pi_G^*$, $U'\br{\theta}$ equals the sum of these two terms $dF$-almost everywhere below $\sutheta$. Hence
\[
    0\geq\frac{2\mu-1}{\mu}c\br{\theta}
      \br{1-\pi_B\br{\theta}}\geq0
    \quad dF\text{-almost everywhere below }\sutheta.
\]
Since $\br{2\mu-1}c\br{\theta}/\mu>0$, $\pi_B\br{\theta}=1$ $dF$-almost everywhere below $\sutheta$. Above $\sutheta$, $\widehat\pi_G\br{\theta}=\pi_G^*\br{\theta}=1$ $dF$-almost everywhere. Since
\[
    1=\widehat\pi_G\br{\theta}
      =\pi_G\br{\theta}
       -\frac{1-\mu}{\mu}\br{1-\pi_B\br{\theta}}
      \leq\pi_G\br{\theta}\leq1,
\]
both inequalities hold with equality. The second gives $\pi_G\br{\theta}=1$. The first gives
\[
    \frac{1-\mu}{\mu}\br{1-\pi_B\br{\theta}}=0,
\]
and since $\mu<1$, this means $\pi_B\br{\theta}=1$. Hence $\pi_G\br{\theta}=\pi_B\br{\theta}=1$ $dF$-almost everywhere above $\sutheta$. Thus $\pi_B=1$ and $\pi_G=\widehat\pi_G=\pi_G^*$ $dF$-almost everywhere. Finally,
\[
    t\br{\theta}
    =\br{\mu\widehat\pi_G\br{\theta}+1-2\mu}u\br{\theta}-U\br{\theta}
\]
gives $t=t^*$ $dF$-almost everywhere. This concludes the proof of \cref{thm:characterizationcontinuumtypes}.
\end{proof}

\end{document}